\documentclass[acmsmall,10pt]{acmart}
\acmJournal{PACMPL}
\acmVolume{1}
\acmNumber{CONF} % CONF = POPL or ICFP or OOPSLA
\acmArticle{1}
\acmYear{2019}
\acmMonth{7}
\acmDOI{} % \acmDOI{10.1145/nnnnnnn.nnnnnnn}
\startPage{1}

\setcopyright{none}
\usepackage{booktabs}   %% For formal tables:
\usepackage{subcaption} %% For complex figures with subfigures/subcaptions
\usepackage{amsmath,amssymb,amsfonts}
\usepackage{algorithmicx}
\usepackage{graphicx}
\usepackage{subcaption}
\usepackage{textcomp}
\usepackage{xcolor}
\usepackage{wrapfig}
    
\usepackage{algorithm}
\usepackage[noend]{algpseudocode}
\usepackage{listings}
\usepackage{enumitem}
\usepackage{balance}
\DeclareMathSizes{10}{10}{10}{10}

\usepackage{multicol} %For multicolumns in algorithms
\algtext*{EndWhile}
\algtext*{EndFor}
\algtext*{EndIf}
\algtext*{EndFunction}

\newcommand\TRUE{\textnormal{\textbf{true}}}
\newcommand\FALSE{\textnormal{\textbf{false}}}

\algrenewcommand\algorithmicindent{1em}
\algdef{S}[IF]{IfNoThen}[1]{\algorithmicif\ #1}

\algblockdefx[StructBlock]{Struct}{EndStruct} [1]{\textbf{struct} #1} [0]{}
\algtext*{EndStruct}
\algblockdefx[EnumBlock]{Enum}{EndEnum} [1]{\textbf{enum} #1} [0]{}
\algtext*{EndEnum}
\algblockdefx[ClassBlock]{Class}{EndClass} [1]{\textbf{class} #1} [0]{}
\algtext*{EndClass}
\algblockdefx[MacroBlock]{Define}{EndDefine} [1]{\textbf{define} #1} [0]{}
\algtext*{EndDefine}
\algblockdefx[InlineBlock]{Inline}{EndInline} [2]{\textbf{inline function} \textsc{#1}(#2)} [0]{}
\algtext*{EndInline}
\algblockdefx[SwitchBlock]{Switch}{EndSwitch} [1]{\textbf{switch} (#1)} [0]{}
\algtext*{EndSwitch}
\algblockdefx[CaseBlock]{Case}{EndCase} [1]{\textbf{case: #1}} [0]{}
\algtext*{EndCase}

\usepackage[utf8]{inputenc}
\usepackage{tikz} 
\usepackage[T1]{fontenc}
\usetikzlibrary{%
  arrows,%
  calc,
  shapes,
  arrows,
  shapes.misc,% wg. rounded rectangle
  shapes.arrows,%https://v2.overleaf.com/project/5b9d0f171f63c01d1a5b8f67
  chains,%
  matrix,%
  positioning,% wg. " of "
  scopes,%
  decorations.pathmorphing,% /pgf/decoration/random steps | erste Graphik
  shadows%
}
\usepackage{amsmath}
\usepackage{relsize}  

\usepackage{amsthm}
\theoremstyle{definition}
\newtheorem{principle}{Principle}

\begin{document}

%% Title information
\title[Quantifiability]{Quantifiability: Concurrent Correctness from First Principles}         %% [Short Title] is optional;
                                        %% when present, will be used in
                                        %% header instead of Full Title.
%\titlenote{This work was funded by the National Science Foundation (NSF) under grant numbers 1717515 and 1740095.}             %% \titlenote is optional;
                                        %% can be repeated if necessary;
                                        %% contents suppressed with 'anonymous'
%\subtitle{Subtitle}                     %% \subtitle is optional
%\subtitlenote{with subtitle note}       %% \subtitlenote is optional;
                                        %% can be repeated if necessary;
                                        %% contents suppressed with 'anonymous'

%% Author information
%% Contents and number of authors suppressed with 'anonymous'.
%% Each author should be introduced by \author, followed by
%% \authornote (optional), \orcid (optional), \affiliation, and
%% \email.
%% An author may have multiple affiliations and/or emails; repeat the
%% appropriate command.
%% Many elements are not rendered, but should be provided for metadata
%% extraction tools.

%% Author with single affiliation.
\author{Victor Cook}
%\authornote{with author1 note}          %% \authornote is optional;
                                        %% can be repeated if necessary
\orcid{0000-0002-9852-2581}             %% \orcid is optional
\affiliation{
%  \position{Position1}
  \department{Department of Computer Science}              %% \department is recommended
  \institution{University of Central Florida}            %% \institution is required
  \streetaddress{4000 Central Florida Blvd.}
  \city{Orlando}
  \state{Florida}
  \postcode{32816}
  \country{USA}                    %% \country is recommended
}
\email{victor.cook@knights.ucf.edu}          %% \email is recommended

%% Author with two affiliations and emails.
\author{Christina Peterson}
%\authornote{with author2 note}          %% \authornote is optional;
                                        %% can be repeated if necessary
\orcid{0000-0002-8070-7633}             %% \orcid is optional
\affiliation{
%  \position{Position2a}
  \department{Department of Computer Science}             %% \department is recommended
  \institution{University of Central Florida}           %% \institution is required
  \streetaddress{4000 Central Florida Blvd.}
  \city{Orlando}
  \state{Florida}
  \postcode{32816}
  \country{USA}                   %% \country is recommended
}
\email{clp8199@knights.ucf.edu}         %% \email is recommended

\author{Zachary Painter}
\orcid{0000-0001-8334-8237}             %% \orcid is optional
\affiliation{
%  \position{Position2b}
  \department{Department of Computer Science}             %% \department is recommended
  \institution{University of Central Florida}           %% \institution is required
  \streetaddress{4000 Central Florida Blvd.}
  \city{Orlando}
  \state{Florida}
  \postcode{32816}
  \country{USA}                   %% \country is recommended
}
\email{zacharypainter@knights.ucf.edu}         %% \email is recommended

\author{Damian Dechev}
\orcid{0000-0002-0569-3403}             %% \orcid is optional
\affiliation{
%  \position{Position2b}
  \department{Department of Computer Science}             %% \department is recommended
  \institution{University of Central Florida}           %% \institution is required
  \streetaddress{4000 Central Florida Blvd.}
  \city{Orlando}
  \state{Florida}
  \postcode{32816}
  \country{USA}                   %% \country is recommended
}
\email{dechev@cs.ucf.edu}         %% \email is recommended

%% Abstract
%% Note: \begin{abstract}...\end{abstract} environment must come
%% before \maketitle command
\begin{abstract}
%The correctness of a concurrent system should not rely upon creating sequential histories. 
Architectural imperatives due to the slowing of Moore's Law, the broad acceptance of relaxed semantics and the O($n!$) worst case verification complexity of generating sequential histories motivate a new approach to concurrent correctness.
%Verifying of a concurrent system should not rely upon creating sequential histories of interleaving methods with a worst-case time complexity of O($n!$).
Desiderata for a new correctness condition are that it be independent of sequential histories, compositional, flexible as to timing, modular as to semantics and free of inherent locking or waiting. 

We propose \textit{Quantifiability}, a novel correctness condition based on intuitive first principles.   Quantifiability models a system in vector space to launch a new mathematical analysis of concurrency.  The vector space model is suitable for a wide range of concurrent systems and their associated data structures.  
%This paper formally defines quantifiablity with its system model and demonstrates useful properties such as composability and the significance of vector norms.
This paper formally defines quantifiability and demonstrates useful properties such as compositionality.
Analysis is facilitated with linear algebra, better supported and of much more efficient time complexity than traditional combinatorial methods.  
%We present results showing that quantifiable data structures are highly scalable, with relaxed semantics an explicit implementation trade-off, not a condition of correctness. 
We present results showing that quantifiable data structures are highly scalable due to the usage of relaxed semantics and propose \textit{entropy} to evaluate the implementation trade-offs permitted by quantifiability. 
\end{abstract}

%% 2012 ACM Computing Classification System (CSS) concepts
%% Generate at 'http://dl.acm.org/ccs/ccs.cfm'.

%% End of generated code

%% Keywords
%% comma separated list
  %% \keywords are mandatory in final camera-ready submission

%% \maketitle
%% Note: \maketitle command must come after title commands, author
%% commands, abstract environment, Computing Classification System
%% environment and commands, and keywords command.
\maketitle

\section{Introduction}
%,``Multicore processors are about to revolutionize the way we design and use data structures.''\cite{Shavit2011-gd}.  
As predicted \cite{Shavit2011-gd}, concurrent data structures have arrived at a tipping point where change is inevitable.  These drivers converge to motivate new thinking about correctness:
\begin{itemize}
    \item Architectural demands to utilize multicore and distributed resources \cite{National_Research_Council2011-ia}.
    \item General acceptance of relaxed semantics \cite{Gruber2016-rz, Haas2013-bn, Wimmer2015-yh, Henzinger2013-pv, Alistarh2018-tw, Adhikari2013-iu, Shavit2015-ce, Rihani2015-cf, Derrick2014-cl, Afek2010-xv}.
    \item  The intractable $O(n!)$ complexity of concurrent system models \cite{Alur1996-mj} prompting the search for reductions \cite{Adhikari2013-iu, Emmi2017-pl,  Liang2013-eh, Derrick2011-hy, OHearn2010-dn, Guerraoui2012-hw, Baier2008-wy, Amit2007-ti, Elmas2010-bw, Derrick2007-te, Tofan2014-ja, Khyzha2016-lv, Baumler2011-jc, Bouajjani2017-ar, Wen2018-le, Singh2016-hh, Alistarh2018-tw, Schellhorn2014-mr, Feldman2018-be, Khyzha2017-ez}.
\end{itemize}
There are a number of correctness conditions for concurrent systems~\cite{papadimitriou1979serializability, Lamport1979-rd, herlihy1990linearizability, Herlihy2011-yj, aspnes1994counting, Afek2010-xv,ou2017checking}. 
%define correct behavior of concurrent programs that is appropriate for the system managing multicore and distributed resources.
%A discussion of these concurrent correctness conditions is presented in Section~\ref{RelatedWork:CorrectnessCondition}.  Here follows a brief synopsis.
%Correctness conditions presented in literature require that a concurrent history is equivalent to a legal sequential history.
The difference between the correctness conditions resides in the allowable method call orderings.
Serializability~\cite{papadimitriou1979serializability} places no constraints on the method call order.
Sequential consistency~\cite{Lamport1979-rd} requires that each method call takes effect in program order.
Linearizability~\cite{herlihy1990linearizability} requires that each method call takes effect at some instant between its invocation and response.
A correctness condition $\mathcal{P}$ is \textit{compositional} if and only if, whenever each object in the system satisfies $\mathcal{P}$, the system as a whole satisfies $\mathcal{P}$~\cite{Herlihy2011-yj}.
Linearizability is a desirable correctness condition for systems of many shared objects where the ability to reason about correctness in a compositional manner is essential.
Sequential consistency is suitable for a standalone system without compositional objects that requires program order of method calls to be preserved such as a hardware memory interface.
Other correctness conditions~\cite{aspnes1994counting,Afek2010-xv,ou2017checking} are defined that permit relaxed behaviors of method calls to obtain performance enhancements in concurrent programs.

These correctness conditions require a concurrent history to be equivalent to a sequential history.
While this way of defining correctness enables concurrent programs to be reasoned about using verification techniques for sequential programs~\cite{guttag1978abstract, hoare1978proof}, it imposes several inevitable limitations on a concurrent system.
Such limitations include 1) requiring the specification of a concurrent system to be described as if it were a sequential system, 2) restricting the method calls to respect data structure semantics and to be ordered in a way that satisfies the correctness condition, leading to performance bottlenecks, and 3) burdening correctness verification with a worst-case time complexity of $O(n!)$ to compute the sequential histories for the possible interleavings of $n$ concurrent method calls.
Some correctness verification tools have provided optimizations~\cite{vechev2009experience,ou2017checking} integrated into model checkers that accept user annotated linearization points to reduce the search space of possible sequential histories to $O(n)$ time in addition to the $O(p \cdot d \cdot r)$ time to perform model checking with dynamic partial-order reductions, where $p$ is the number of processes, $d$ is the maximum size of the search stack, and $r$ is the number of transitions explored~\cite{flanagan2005dynamic}.
However, this optimization technique for correctness verification is only effective if all methods have fixed linearization points.

This paper proposes \textit{Quantifiability}, a new definition of concurrent correctness that does not require reference to a sequential history.
Freeing analysis from this historical artifact of the era of single threaded computation and establishing a purely concurrent correctness condition is the goal of this paper.  
% Yoda is my idol
%The goal of this paper is to establish a purely concurrent correctness condition whose analysis is free from this historical artifact of the era of single threaded computation.
Quantifiability eliminates the necessity of demonstrating equivalence to sequential histories by evaluating correctness of a concurrent history based solely on the outcome of the method calls.
%Quantifiability is strict enough to be useful and enables modular analysis with its functional and locally defined properties.  
Like other conditions it does require atomicity of method calls and enables modular analysis with its compositional properties.  
%It uses concurrent reasoning to reduce complexity and the awkwardness of mapping to sequential histories.
%Quantifiability supports separation of concerns: Principles of correctness are not mixed with timing or data structure semantics. 

Quantifiability supports separation of concerns. Principles of correctness are not mixed with real-time ordering or data structure semantics.  These are modifiers or constraints on the system.     
%As a thought exercise while reading this paper substitute quantifiability for ``relaxed semantics'' and  
%Quantifiability is a new correctness condition that does not require reference to a sequential history.  
%Freeing analysis from this historical artifact of the era of single threaded computation and establishing a purely concurrent correctness condition is the goal of this paper.  
%The presentation draws unabashedly on the work of Herlihy, who shaped the way a generation thinks about concurrency.
%disregard of real-time order and
The conservation of method calls enables concurrent histories to be represented in vector space.  
Although it is not possible to foresee all uses of the vector space model, this paper will demonstrate the use of linear algebra to efficiently verify concurrent histories as quantifiable.

The following presentation of quantifiability and its proposed verification technique draws unabashedly on the work of Herlihy~\cite{Herlihy2011-yj}, who shaped the way a generation thinks about concurrency.
%This paper unifies recent thinking about concurrent correctness already in evidence in dozens of papers, many cited here. 
This paper repositions concurrent correctness in a way that overcomes the inherent limitations associated with defining correctness through equivalence to sequential histories.  Contributions to the field are:
\begin{enumerate}

%\item We propose quantifiability as a concurrent correctness condition and contrast it with other correctness conditions.
\item We propose quantifiability as a concurrent correctness condition and illustrate its benefits over other correctness conditions.
%\item Prove that quantifiability is local and non-blocking. 
%\item Examples of quantifiable semantics are
\item We show that quantifiability is compositional and non-blocking.
\item We introduce linear algebra as a formal tool for reasoning about concurrent systems.
%by representing concurrent histories in vector space and applying linear algebra to verify quantifiability.
%\item Present a formal method for reasoning about correctness of a concurrent history without referencing a sequential history.
\item We present a verification algorithm for quantifiability 
with a significant improvement in time complexity compared to
%in comparison to the complexity of verification for 
%traditional correctness conditions that compute
%are required to compute all possible
analyzing all possible sequential histories.
\item A quantifiably correct concurrent stack and queue are implemented and shown to scale well. 
%Q Container achieves a 155\% speedup over the elimination backoff stack~\cite{Hendler2004-ho}.
\end{enumerate}

%The state of a concurrent system is that of objects and the methods pending on them.    
%Concurrent systems must make progress while ensuring that methods never cause an inconsistent state.  This requirement is divided into a \textit{progress} condition (liveness) and a \textit{correctness} condition (safety).  The rate of progress is \textit{performance}, which limits the size of the computation that can be accomplished in a given interval.  When all pending methods have completed, the system is said to be \textit{quiescent}. At such time its state depends only upon the objects.  Concurrent systems that are far from quiescent are of interest.  Loaded with pending methods, they demand that multiprocessing provide the speedup formerly granted by Moore's Law.  

%The analogous metric for correctness is \textit{complexity}, which limits the size of the system to which it can be applied.  Complexity typically arises from the size of the state space as a concurrent history is exploded into multiple sequential histories.

\subsection{First Principles}
\label{SubSection:FirstPrinciples}
%Among principles that constrain correct system behavior some are of primary importance, defining what things happened while secondary constraints control their timing or order. 
Among principles that describe concurrent system behavior, \textit{first principles} 
%(also referred to as \textit{primary principles}) 
define what things happen while \textit{secondary principles} such as timing and order are modifiers on them. 
The conditions defined in secondary principles do not make sense without the first, but the reverse is not the case \cite{Descartes1903-wy}.
%First principles are defined \textit{ab initio}, while
%First principles define things that stand alone, while secondary principles such as order make no sense without referring to the things being ordered .  
This view accords with intuition: Tardiness to a meeting is secondary to there being a meeting at all.  In a horse race, first principles define that the jockeys and horses are the same ones who started; only then does the order of finish makes sense.  The intuition that the events themselves are more important than their order, and that conservation is a prerequisite for ordering will be motivated with the following examples.
%Being late to a meeting is less significant than there being no meeting at all, or having two uninvited people show up instead.   

%Although there are cases such as catching a train where being five minutes late may mean missing it (if the train is on time), the principle stands.  Where timing is critical 

%*****************TODO: OPTIONAL************************
%Brief definitions for the concurrent correctness conditions discussed in the following examples are provided here to give the reader intuition regarding the variance of ordering constraints placed upon method calls by the correctness conditions.

\def\pt{1.90}
\def\pc{1.75}
\def\pb{1.60}

\def\qt{1.15}
\def\qc{1.0}
\def\qb{0.85}

\def\rt{0.40}
\def\rc{0.25}
\def\rb{0.10}

\begin{wrapfigure}{r}{0.55\textwidth}
%\begin{figure}[h]
\centering
\begin{tikzpicture}[draw=black, scale=1, transform shape]

%\draw[help lines] (0,0) grid (10,1);
\node[align=left] at (0.5,\pc) {P0};
\node[align=left] at (0.5,\qc) {P1};
%\node[align=left] at (0.5,0.5) {P2};

\draw [thick] (1.0,\pc) -- (3.0,\pc);
\draw (1.0,\pb) -- (1.0,\pt);
\node[align=center, above] at (2.0,\pb+0.1) {$push^x_0(7)$};
\draw (3.0,\pb) -- (3.0,\pt);

\draw [thick] (5.5,\pc) -- (7.5,\pc);
\draw (5.5,\pb) -- (5.5,\pt);
\node[align=center, above] at (6.5,\pb+0.1) {$pop^x_0(7)$};
\draw (7.5,\pb) -- (7.5,\pt);

\draw [thick] (3.25,\qc) -- (5.25,\qc);
\draw (3.25,\qb) -- (3.25,\qt);
\node[align=center, above] at (4.25,\qb+0.1) {$push^x_1(8)$};
\draw (5.25,\qb) -- (5.25,\qt);

\node[align=center, above] at (4.0,-0.5) {\textbf{History H1:} Sequentially consistent and ``almost'' \\ linearizable.}; 
\end{tikzpicture}
%\end{figure}
\end{wrapfigure}

A concurrent history $H$ defines the events in a system of concurrent processes and objects.
An object subhistory, $H|O$, of a history $H$ is the subsequence of all events in $H$ that are invoked on object $O$~\cite{herlihy1990linearizability}.
Consider history $H1$ on object $x$ with Last-In-First-Out (LIFO) semantics.  
The notation follows the convention $m^{o}_{p}$($i$), where $m$ is the method, $o$ is the object, $p$ is the process, and $i$ is the item to be input or output.
$H1$ is serializable and also sequentially consistent.  $H1$ is not linearizable, but it would be if the interval of $push^x_0(7)$ were slightly extended to overlap $push^x_1(8)$ or a similar adjustment were made relative to $pop^x_0(7)$.  Linearizability requires determining ``happens before'' relationships among all method calls to project them onto a sequential timeline.  Doing this with a shared clock timing the invocation and response of each method call is not feasible \cite{Sheehy2015-uh}.  
%What is generally available is logical time \cite{Lamport1978-ip}.
What is available, given some inter-process communication, is a logical clock \cite{Lamport1978-ip}.
%Once established, linearizability is compositional \cite{Herlihy1990-ms}.
%The assumption that exact intervals can be determined ties analysis to a sequential timeline.  %and underpins methodologies for proving programs correct. 
Linearizability is sometimes ``relaxed'', creating loopholes to enable performance gains.  Without timing changes, $H1$ is linearizable using $k$-LIFO semantics where $k >= 2$ \cite{Shavit2015-ce}.  
%True to the intuition that what happens is primary as opposed to when and how, $H1$ happens to be serializable and sequentially consistent even if the semantics are relaxed to allow First-In-First-Out (FIFO). 

\begin{wrapfigure}{r}{0.55\textwidth}
%\begin{figure}[h]
\centering
\begin{tikzpicture}[draw=black, scale=1, transform shape]

%\draw[help lines] (0,0) grid (10,1);
\node[align=left] at (0.5,\pc) {P0};
\node[align=left] at (0.5,\qc) {P1};
%\node[align=left] at (0.5,0.5) {P2};

\draw [thick] (1.0,\pc) -- (3.0,\pc);
\draw (1.0,\pb) -- (1.0,\pt);
\node[align=center, above] at (2.0,\pb+0.1) {$push^x_0(7)$};
\draw (3.0,\pb) -- (3.0,\pt);

\draw [thick] (5.5,\pc) -- (7.5,\pc);
\draw (5.5,\pb) -- (5.5,\pt);
\node[align=center, above] at (6.5,\pb+0.1) {$pop^x_0(3)$};
\draw (7.5,\pb) -- (7.5,\pt);

\draw [thick] (3.25,\qc) -- (5.25,\qc);
\draw (3.25,\qb) -- (3.25,\qt);
\node[align=center, above] at (4.25,\qb+0.1) {$push^x_1(8)$};
\draw (5.25,\qb) -- (5.25,\qt);

\node[align=center, above] at (4.0,-0.5) {\textbf{History H2:} Not serializable because calls are not \\ conserved.}; 

\end{tikzpicture}
%\end{figure}
\end{wrapfigure}

Consider history $H2$ on the same object $x$.  $H2$ is not serializable, not sequentially consistent, not linearizable, and no changes in timing will allow $H2$ to meet any of these conditions.  Also, there is no practical relaxation of semantics that accepts $H2$.  There is an essential difference in the correctness of $H1$ and $H2$.  What happened in history $H1$ is intuitively acceptable, given some adjustments to when (timing) and how (relaxed semantics) it happened.  What happened in history $H2$ is impossible, as it creates the return value 3 from nothing.  As in the equestrian example, item 3 is not one of the starting horses.  
The method calls on object $x$ are not \textit{conserved}.  
A correctness condition that captures the difference between $H1$ and $H2$ allows separating the concerns of what happened and when according to the (possibly relaxed) semantics.  

\begin{wrapfigure}{r}{0.55\textwidth}
%\begin{figure}[h]
\centering
\begin{tikzpicture}[draw=black, scale=1, transform shape]

%\draw[help lines] (0,0) grid (10,1);
\node[align=left] at (0.5,\pc) {P0};
\node[align=left] at (0.5,\qc) {P1};
%\node[align=left] at (0.5,0.5) {P2};

\draw [thick] (1.0,\pc) -- (3.5,\pc);
\draw (1.0,\pb) -- (1.0,\pt);
\node[align=center, above] at (2.25,\pb+0.1) {$pop^x_0(7)$};
\draw (3.5,\pb) -- (3.5,\pt);

\draw [thick] (3.75,\pc) -- (6.0,\pc);
\draw (3.75,\pb) -- (3.75,\pt);
\node[align=center, above] at (4.75,\pb+0.1) {$push^y_0(8)$};
%\node[align=center, above] at (4.75,\pb+0.1) {$pop^y_0(null)$};
% Can we introduce the null values after your explanation?   
\draw (6.0,\pb) -- (6.0,\pt);

\draw [thick] (2.0,\qc) -- (4.25,\qc);
\draw (2.0,\qb) -- (2.0,\qt);
\node[align=center, above] at (3.0,\qb+0.1) {$pop^y_1(8)$};
\draw (4.25,\qb) -- (4.25,\qt);

\draw [thick] (4.5,\qc) -- (6.75,\qc);
\draw (4.5,\qb) -- (4.5,\qt);
\node[align=center, above] at (5.5,\qb+0.1) {$push^x_1(7)$};
%\node[align=center, above] at (5.5,\qb+0.1) {$pop^x_1(null)$};
\draw (6.75,\qb) -- (6.75,\qt);

\node[align=center, above] at (4.0,-0.5) {\textbf{History H3:} Serializable, not sequentially consistent.};

\end{tikzpicture}
%\end{figure}
\end{wrapfigure}
%TODO Take a look at this figure, it is actually linearizable, and therefore also sequentially consistent. We need to fix this.

History $H3$ has two objects $x$ and $y$. 
The projections $H3|x$ and $H3|y$ are serializable.  The combined history $H3$ is serializable.  
%Unless multiple method calls are cemented into transactions, serializability is compositional because it is fully commutative \cite{Papadimitriou1979-qu}.  
Projections $H3|x$ and $H3|y$ are also sequentially consistent.  However, their composition into $H3$ is not sequentially consistent.  Sequential consistency is not compositional~\cite{Lamport1978-ip}.  Projection $H3|x$ is not linearizable, therefore $H3$ is also not linearizable. 

%Current correctness conditions confound the what, the when and the how of histories.  

%A method is a \textit{function} if it is uniquely defined for every object state; otherwise it is a \textit{relation}~\cite{Gallier2015-ib}.
A method is \textit{total} if it is defined for every object state; otherwise it is \textit{partial}.
\textit{Conditional semantics} are semantics that enable a partial method to return null upon reaching an undefined object state.
History $H4$ is the same as $H3$ with the exception of introducing conditional semantics for $pop$, making explicit a common relaxation of how a stack works.  With the conditional $pop$ $H4$ is sequentially consistent, yielding multiple correct orderings and end states. 

\begin{wrapfigure}{r}{0.55\textwidth}
%\begin{figure}[h]
\centering
\begin{tikzpicture}[draw=black, scale=1, transform shape]

%\draw[help lines] (0,0) grid (10,1);
\node[align=left] at (0.5,\pc) {P0};
\node[align=left] at (0.5,\qc) {P1};
%\node[align=left] at (0.5,0.5) {P2};

\draw [thick] (1.0,\pc) -- (3.5,\pc);
\draw (1.0,\pb) -- (1.0,\pt);
\node[align=center, above] at (2.25,\pb+0.1) {$pop^y_0(null\lor8)$};
\draw (3.5,\pb) -- (3.5,\pt);

\draw [thick] (3.75,\pc) -- (6.0,\pc);
\draw (3.75,\pb) -- (3.75,\pt);
\node[align=center, above] at (4.9,\pb+0.1) {$push^x_0(7)$};
\draw (6.0,\pb) -- (6.0,\pt);

\draw [thick] (2.0,\qc) -- (4.25,\qc);
\draw (2.0,\qb) -- (2.0,\qt);
\node[align=center, above] at (3.0,\qb+0.1) {$pop^x_1(null\lor7)$};
\draw (4.25,\qb) -- (4.25,\qt);

\draw [thick] (4.5,\qc) -- (6.75,\qc);
\draw (4.5,\qb) -- (4.5,\qt);
\node[align=center, above] at (5.65,\qb+0.1) {$push^y_1(8)$};
\draw (6.75,\qb) -- (6.75,\qt);

\node[align=center, above] at (4.0,-.5) {\textbf{History H4:} Conditional $pop$ makes $H3$ sequentially \\ consistent.};

\end{tikzpicture}
%\end{figure}
\vspace{-2em}
\end{wrapfigure}

But conditional $pop$ is not consistent with early formal definitions of the stack abstract data type where $pop$ on an empty stack threw an error \cite{Guttag1976-nq}
 % John Guttag, 668 citations
 or a signal \cite{Liskov1994-rn}. 
% Liskov and Wing, 1464 citations
The semantics of these exceptions were taken seriously \cite{Gogolla1984-sx}.
% Martin Gogolla, 290 citations
Invariants prevented exceptions, and there was ``no guarantee'' of the result if they were violated \cite{Zaremski1995-vm}. 
% Zaremski and Liskov, 753 citations
The conditional $pop$ can be traced to the literature on performance \cite{Badrinath1987-oe}, 
% Badinrath 1987, 70 citations
where the requirement to handle errors and check invariants is ignored.  Conditional semantics remain prevalent in recent work, extending to proofs of correctness allowing two different linearization points with respect to the same method calls \cite{Amit2007-ti}. 

\iffalse
Allowing a stack pop to conditionally return null is not merely a convenience for performance benchmarks.  Altering the semantics of $pop$ to return null is as significant as changing the order in which items are returned. 
Imagine the task of debugging even a simple sequential process if the result of method calls defined inside the program were conditional upon events occurring outside the program.  For example, define \texttt{x--} as a method $\lambda x.(x-1)$, which a process calls respecting the invariant that \texttt{x>0}.  Now suppose that the method occasionally returns null and fails to decrement x.  This is analogous to the behavior of a concurrent stack with a conditional $pop$.  
%The process sees the $pop$ fail even while it respects the invariant.
%Method invariant semantics (how) and ordering of return values (when) are important at the same level. 

Type theory supports the view that conditional semantics, although permissible, alter the method just as fundamentally as other semantic changes. A well-cited formal definition of an abstract data type confirms this \cite{Zaremski1995-vm}, representing both the empty stack invariant and the ordering of return values for $pop$ as subtypes of an unordered ``bag'' of items.  Each is a separate concern, explicitly implemented within a type system.
\fi

\begin{wrapfigure}{r}{0.55\textwidth}
%\begin{figure}[h]
\centering
\begin{tikzpicture}[draw=black, scale=1, transform shape]

%\draw[help lines] (0,0) grid (10,1);
\node[align=left] at (0.5,\pc) {P0};
\node[align=left] at (0.5,\qc) {P1};
\node[align=left] at (0.5,\rc) {P2};

\draw [thick] (1.0,\pc) -- (3.0,\pc);
\draw (1.0,\pb) -- (1.0,\pt);
\node[align=center, above] at (2.0,\pb+0.1) {$pop^z_0(null\lor1$)};
\draw (3.0,\pb) -- (3.0,\pt);

\draw [thick] (5.5,\pc) -- (7.5,\pc);
\draw (5.5,\pb) -- (5.5,\pt);
\node[align=center, above] at (6.5,\pb+0.1) {$pop^z_0(null\lor1)$};
\draw (7.5,\pb) -- (7.5,\pt);

\draw [thick] (3.25,\qc) -- (5.25,\qc);
\draw (3.25,\qb) -- (3.25,\qt);
\node[align=center, above] at (4.25,\qb+0.1) {$pop^z_1(null\lor1)$};
\draw (5.25,\qb) -- (5.25,\qt);

\draw [thick] (3.25,\rc) -- (5.25,\rc);
\draw (3.25,\rb) -- (3.25,\rt);
\node[align=center, above] at (4.25,\rb+0.1) {$push^z_2(1)$};
\draw (5.25,\rb) -- (5.25,\rt);

\node[align=center, above] at (4.0,-0.75) {\textbf{History H5:} $P0$ keeps trying to $pop$.};

\end{tikzpicture}
%\end{figure}
\vspace{-1em}
\end{wrapfigure}

History $H5$ illustrates another problem with the conditional $pop$.  Consider a stack that allocates a scarce resource.  $P0$ issued a request before $P1$ and repeats it soon after, but gets nothing.  $H5$ might be repeated many times with $P1$ and $P2$ exchanging the item.   The scheduler allocates \textit{twice} as many requests per cycle to $P0$ as either $P1$ or $P2$, so why is there starvation? It is because conditional $pop$ is \textit{inherently unfair}.  %History $H5$ is serializable, sequentially consistent and linearizable, 
Although $P0$ is not blocked in the sense of waiting to complete the method \cite{Herlihy1991-jf}, conditional $pop$ causes it to repeatedly lose its place in the ordering of requests.  It might be called ``progress without progress.'' 
%Recall that sequential consistency causes \textit{inherent blocking} and this was used to show the benefits of linearizability \cite{herlihy1990linearizability}. 
%NOTE: This is definitely serializability that is said to cause inherent blocking according to Herlihy's Linearizability paper. See page 11, last paragraph.
Recall that serializability causes \textit{inherent blocking} and this was used to show the benefits of linearizability \cite{herlihy1990linearizability}.  
A new correctness condition should be free of \textit{inherent unfairness} as well. 
%In defense of these conditions, this is considered a progress issue, but as $P0$ sees it something is wrong.  
%In a quantifiably correct system the first $pop$ request would remain active, not blocking unless the thread desires to wait, in the program $pop(null\lor1)$ can be made asynchronous by allocating an address as an inbox so it reads $pop(\&v)$ and checking as desired to see if the address is updated.  

\subsection{Desirable Properties for Quantifiability}

Multicore programming is considered an art \cite{Herlihy2011-yj} and is generally regarded as difficult.  
%The insight of projection onto a sequential timeline provided an abstraction to understand systems with multiple processes and led to definitions of their correctness.  
The projection of a concurrent history onto a sequential timeline provided an abstraction to understand systems with a few processes and led to definitions of their correctness.  
%Linearizability, being compositional, ensured that system complexity is linear in the number of objects. 
Linearizability, being compositional, ensures that reasoning about system correctness is linear with respect to the number of objects. 
%But linearizable complexity is not linear in the number of method calls, and systems may have thousands of processes and millions of method calls.  
But verifying linearizability for an individual object is not linear in the number of method calls.  Systems today may have thousands of processes and millions of method calls, far beyond the capacity of current verification tools. 
The move from art to an engineering discipline requires a new correctness condition with the following desirable properties:

%of a purely concurrent correctness condition:  

\begin{itemize} 
\item \textbf{Conservation} What happens is what the methods did.  Return values cannot be pulled from thin air (history $H2$). Method calls cannot disappear into thin air (history $H5$).  
\item \textbf{Measurable} Method calls have a certain and measurable impact on system state, not sometimes null (history $H4$).  
\item \textbf{Compositional} Demonstrably correct objects and their methods may be combined into demonstrably correct systems (history $H3$).
\item \textbf{Unconstrained by Timing} Correctness based on timing limits opportunities for performance gains and incurs verification overhead when comparing the method call invocation and response times to determine which method occurs first in the history (history $H1$).
%(history $H1$ and $H4$). 
%Correctness based on timing has high complexity and can be handled in the type system rather than in the correctness condition 
%\item \textbf{Locking, Waiting, Starvation} Design of the correctness condition should not cause these known structural failures. 
\item \textbf{Lock-free, Wait-free, Deadlock-Free, Starvation-Free} Design of the correctness condition should not limit or prevent system progress. 
\end{itemize}

\section{Definition}
Quantifiability is concerned with the impact of method calls on the system as opposed to the projection of a method call onto a sequential timeline.
The \textit{configuration} of an arbitrary element comprises the values stored in that element.
The \textit{system state} is the configuration of all the objects that represents the outcome of the method calls by the processes.

\subsection{Principles of Quantifiability}
\label{Section:Principles}

A familiar way to introduce a correctness condition is to state principles that must be followed for it to be true \cite{Herlihy2011-yj}.  Quantifiability embodies two principles.

\begin{principle}
\textbf{Method conservation:  Method calls are first class objects in the system that must succeed, remain pending, or be explicitly cancelled.}
\end{principle}
Principle~1 requires that every instance of a process calling a method, including any arguments and return values specified, is part of the system state.  Method calls are not ephemeral requests, but ``first class'' \cite{Abelson_undated-vk,Strachey1967-qn} members of the system.  All remain pending until they succeed or are explicitly cancelled. Method calls may not be cancelled implicitly as in the conditional $pop$. 
%(although a conditional $popc$ might be defined, it does not offer the guarantees of a $pop$).  
Actions expected from the method by the calling process must be completed.  
%The calling process must get a result conforming to the return type.  
This includes returning values (if any) and making the expected change to the state of the concurrent object on which the method is defined. 

Duplicate method calls can be handled in several ways conforming to Principle~1.  A duplicate call might be considered a syntactic shorthand for ``cancel the first operation and resubmit'', or it could throw a run time error to have identical calls on the same address.   Alternatively an index could be added to the method call to uniquely identify it such that it can be distinguished from other identical calls. 
\iffalse
Quantifiable semantics support these options in blocking and non-blocking versions well suited to the type systems of programming languages.  %TODO What are we trying to say here?
Useful extensions of this include permitting threads to pre-request some number of method calls, and the ability to assign priority to the method call independent of the order represented in the data structure.  
For example a quantifiable priority queue could prioritize the items in the queue and prioritize the incoming requests.
%For example the priority queue demonstrated in this paper implements the priority of items in the queue and the priority of incoming requests.
\fi

\begin{principle}
\textbf{Method quantifiability: Method calls have a measurable impact on the system state. }
\end{principle}
Principle~2 requires that every method call owns a scalar value, or metric, that reflects its impact on system state.  There is some total function that computes this metric for each instance of a method call.  
%If the method is cancelled and not completed, it may be zero. 
%Additionally, if conditional semantics for a method are desired such that the method returns true if it successfully performs its operation and false otherwise, then the metric for the method when false is returned can also be assigned a value of zero.
Building on Principle~1 that conserves the method calls themselves, Principle~2 requires that a value can be assigned to the method call.  All method calls ``count.''  The conservation of method calls along with the measurement of their impact on system state is what gives quantifiability its name.  Values are assigned as part of the correctness analysis. As with concepts such as linearization points, these values are not necessarily part of the data structure, but are artifacts for proving correctness.   

The assignment of values to the method calls may be straightforward.  
%In the stack example presented in Section~\ref{SubSection:SystemModel}, $push$ is set to +1 and $pop$ is set to -1.  
For the stack abstract data type, $push$ is set to +1 and $pop$ is set to -1.  
Sometimes value assignments are subtle:  Principle~1 requires that reads are first class members of the system state, so performing them is a state change.  Reads will have a small but measurable impact, unlike reads in other system models that are considered to have no effect on system state.  Probes such as a $contains$ method on a set data type also have a value.
%Without this the linear algebra of correctness would not be effective.

%******************TODO: OPTIONAL**************************
%These two principles enable computation on the concurrent state so that quantifiability can be verified without reference to a sequential history.  
\iffalse
These two principles enable computation of the system state so that quantifiability can be verified without reference to a sequential history.  
Completed methods have a value. Pending methods too are first class members of the system state.  Methods are conserved and have an indefinite lifespan, although they may be explicitly cancelled.  
%There is no need write them in a sequential history to make a statement about correctness.
There is no need to record them in a sequential history to make a statement about correctness.
\fi

An informal definition of quantifiability is now presented to provide the reader with intuition regarding the meaning of quantifiability. The informal definition is followed by a description of the system model in Section~\ref{SubSection:SystemModel} and a formal definition of quantifiability in Section~\ref{SubSection:FormalDefinition}.

\begin{definition}
\label{Definition:Informal}
(Informal). A history $H$ is \textit{quantifiable} if each method call in $H$ succeeds, remains pending, or is explicitly canceled, and the effect of each method call appears to execute atomically and in isolation. Furthermore the effect of every completed method call makes a measurable contribution to the system state.
\end{definition}

It may appear as though quantifiability is equivalent to serializability.  However, quantifiability does not permit method calls to return null upon reaching an undefined object state while serializability does permit this behavior.  Quantifiability measures the outcome of every method call by virtue of its completion.
This subtle difference directly impacts the complexity of analysis, and can lead to throughput increases for quantifiability when designing a data structure.  Quantifiable implementations learn from relaxed semantics to ``save'' a method call in an accumulator rather than discarding it due to a data structure configuration where the method call could not be immediately fulfilled. Quantifiable data structure design is discussed in further details in Section~\ref{Section:Implementation}. 

\subsection{System Model}
\label{SubSection:SystemModel}
A concurrent system is defined here as a finite set of methods, processes, objects and items.  Methods define \textit{what} happens in the system.  Methods are defined on a class of objects but affect only the instances on which they are called.   Processes are the actors \textit{who} call the methods, either in a predetermined sequence or asynchronously driven by events.  Objects are encapsulated containers of concurrent system state.  Objects are \textit{where} things happen.  Items are data passed as arguments to and returned as a result from completed method calls on the concurrent objects.  
%Items are passive and may not have methods defined on them. 
Method invariants and semantics place constraints such as order, defining \textit{how} things happen. Quantifiable concurrent histories are serializable so every method call takes effect during the interval spanning the history, meaning that \textit{when} method calls occur may be reordered to achieve correctness. 
\iffalse
A summary of the components of the system model are presented in Table~\ref{Table:SystemModelRoles}.

%Convolutions of the system model components are prohibited, so for example methods may not create or modify processes, objects cannot be items, items cannot morph into each other, processes are not themselves items.
%Convolutions of the above are prohibited, so for example methods may not create or modify processes, objects cannot create new items spontaneously,  processes are not themselves items.
%The domain of items is often infinite, such as real numbers or strings in a language.  There are infinite possibilities of objects and methods defined on them.  Processes too may spawn indefinitely.   However the actual concurrent systems to be considered in the model presented here are finite. 

\def\arraystretch{1.5}
\begin{table}[]
\centering
\begin{tabular}{|c|p{6cm}|c}
\cline{1-2}
\textbf{Component} & \textbf{Role} &  \\ \cline{1-2}
Processes &  Processes call methods  &  \\ \cline{1-2}
Methods &   Methods act on instances of objects, only on whose abstract data type that method is defined  &  \\ \cline{1-2}
Objects &   Objects are instances of an abstract data type  &  \\ \cline{1-2}
Items &   Items are data passed as input to methods or returned as output from methods &  \\ \cline{1-2}
\end{tabular}
\vspace{4mm}
\caption{System model components and roles} \label{Table:SystemModelRoles}
\end{table}
\fi

A \textit{method call} is a pair consisting of an invocation and next matching response~\cite{herlihy1990linearizability}. 
An invocation is \textit{pending} in history $H$ if no matching response follows the invocation~\cite{herlihy1990linearizability}.
Each method call is specified by a tuple (Method, Process, Object, Item).
A method call with input or output that comprises multiple items can be represented as a single structured item.
An execution of a concurrent system is modeled by a \textit{concurrent history} (or simply \textit{history}), which is a multiset of method calls~\cite{herlihy1990linearizability}. 
%A concurrent history is a multi-set of method calls, each repeated a finite number of times. 
%Method calls in the history are done \textit{by} a process \textit{on} an object and further specified with items passed or returned.  
%Multiple items from the input and output may be considered a single structured item.  
%The \textit{system state} is the configuration of all the objects and items that represents the outcome of the method calls by the processes.
%Although the domain of possible methods, processes, objects and permutations of input and output items is infinite, actual concurrent histories contain a small subset.  
Although the domain of possible methods, processes, objects, and items is infinite, actual concurrent histories are a small subset of these.

It is not unusual when discussing concurrent histories to speak of, ``the projection of a history onto objects.'' However the focus from there has always been on building sequential histories, so the literature does not extend this language to bring the analysis of concurrent histories formally into the realm of linear algebra.  
Quantifiability facilitates this extension, with fruitful consequences.
%This section introduces a vector notation.

\subsection{Vector Space}
\label{SubSection:VectorSpace}
A \textit{vector} is an ordered $n$-tuple of numbers, where $n$ is an arbitrary positive integer.
%A \textit{row vector} is a vector with a row-by-column dimension of 1 by $n$.
A \textit{column vector} is a vector with a row-by-column dimension of $n$ by 1.
In this section our system model is mapped to a vector space over the field of real numbers $\mathbb{R}$.  
The system model is isomorphic to vector spaces over $\mathbb{R}$ described in matrix linear algebra textbooks~\cite{Beezer2008-eq}.
From this foundation, analysis of concurrent histories can proceed using the tools of linear algebra.

The components of the system model are represented as dimensions in the vector space, written in the order Methods (M), Processes (P), Objects(O) and Items (I).  
The \textit{basis vector} of the history is the Cartesian product of the dimensions $M \times P \times O \times I$.
Each unique configuration of the four components defines a basis for a vector space over the real numbers.
The spaces thus defined are of finite dimension.
In this model, \textit{orthogonal} means that dimensions are independent.
%For example, the projection of history $H$ onto object $x$ ($H|x$) is separate from the projection of history $H$ onto object $y$ ($H|y$).
An \textit{orthogonal basis} is a basis whose vectors are orthogonal.
It is necessary to define an orthogonal basis because each non-interacting method call with distinct objects, processes, and items is an independent occurrence from every other combination.  

%TODO Explain how to create a matrix of method call vectors
A history is represented by a vector with \textit{elements} corresponding to the basis vector uniquely defined by the concurrent system. 
%Every concurrent history is defined as a vector whose coefficients correspond to the ordered basis that uniquely defines the concurrent system.
Principle~2 states that each method call has a value. These are the values represented in the elements of the history vector.
On a LIFO stack, $push$ and pop methods are inverses of each other.  An important difference is that $push$ is completed without dependencies in an unbounded stack, whereas $pop$ returns the next available item, which may not arrive for some time.  A history that shows a completed $pop$ must account for the source of the item being returned, either in the initial state or in the history itself.  The discussion of history $H2$ in Section~\ref{SubSection:FirstPrinciples} showed this is common to the analysis of serializability, sequential consistency, and linearizability.

Concurrent histories can be written as column vectors whose elements quantify the occurrences of each unique method call, that is, a vector of coordinates over $\mathbb{R}$ acting on a basis constructed of the Methods, Processes, Objects and Items involved. History $H1$ in Section~\ref{SubSection:FirstPrinciples} can be written: 
   \begin{align} 
    H1  &= \begin{pmatrix}
           1 \\
           0 \\
           0 \\
           1 \\
           -1 \\
           0 \\
           0 \\
           0 
         \end{pmatrix}
         & basis = 
         & \begin{bmatrix}
         push, \hspace{3pt}  P0, \hspace{3pt}  x, \hspace{3pt}  7 \\
         push, \hspace{3pt} P0, \hspace{3pt} x, \hspace{3pt} 8 \\
         push, \hspace{3pt} P1, \hspace{3pt} x, \hspace{3pt} 7 \\
         push, \hspace{3pt} P1, \hspace{3pt}  x, \hspace{3pt} 8 \\
         pop, \hspace{3pt} P0, \hspace{3pt} x, \hspace{3pt} 7 \\
         pop, \hspace{3pt} P0, \hspace{3pt} x, \hspace{3pt} 8 \\
         pop, \hspace{3pt} P1, \hspace{3pt} x, \hspace{3pt} 7 \\
         pop, \hspace{3pt} P1, \hspace{3pt} x, \hspace{3pt} 8 
         \end{bmatrix}
  \end{align}

The history vector has the potential to be large considering that the basis vector is defined according to the Cartesian product of the dimensions $M \times P \times O \times I$.
However, the history vector will likely be sparse unless all possible combinations in which the items passed to the methods invoked by the processes on the objects occur in the history.
If the history vector is sparse, then the algorithms analyzing the history vector can be compressed such that the non-zero values are stored in a compact vector and for each element in the compact vector, the corresponding index in the original history vector is stored in an auxiliary vector~\cite{williams2007optimization}.
%optimized to take advantage of the sparse structure.  
With a dense  history vector where the majority of the elements in the history vector represent a method that actually occurs in the history, the complexity remains contained as standard linear algebra can be applied in the analysis of the history vector. 
%An example analysis is demonstrated in III-C using the Python Numpy library.

\subsection{Formal Definition}
\label{SubSection:FormalDefinition}
The formal definition of quantifiability is described using terminology from mathematics, set theory, and linear algebra in addition to formalisms presented by Herlihy et al.~\cite{herlihy1990linearizability} to describe concurrent systems. 
%The principles of quantifiability enable us to provide a formal definition without reference to a legal sequential history. 
%A reference to a legal sequential history is omitted by representing a history as a set of vectors and defining quantifiability based on properties of the set of vectors. The key benefit of this approach is that it avoids the $O(n!)$ growth rate of possible legal sequential histories to consider when evaluating program correctness.
Methods are classified according to the following convention.
A \textit{producer} is a method that generates an item to be placed in a data structure.
A \textit{consumer} is a method that removes an item from the data structure.
A \textit{reader} is a method that reads an item from the data structure.
A \textit{writer} is a method that writes to an existing item in the data structure.

A \textit{method call set} is an unordered set of method calls in a history.
A \textit{producer set} is a subset of the method call set that contains all its producer method calls.
A \textit{consumer set} is a subset of the method call set that contains all its consumer method calls.
A \textit{writer set} is a subset of the method call set that contains all its writer method calls.
A \textit{reader set} is a subset of the method call set that contains all its reader method calls.
Since a method call is a pair consisting of an invocation and next matching response, no method in the method call set will be pending.
Quantifiability does not discard the pending method calls from the system state nor does it place any constraints on their behavior while they remain pending.

The history vector described in Section~\ref{SubSection:VectorSpace} is transformed such that the method calls in a history are represented as a set of vectors.
%If the method calls in a history are represented as a set of vectors, it can be verified that a history is quantifiable using linear algebra.
%To maintain consistency in the application of linear algebra to verify quantifiability, all method calls are represented as column vectors.
To maintain consistency in defining quantifiability as a property over a set of vectors, all method calls are represented as column vectors.
Each position of the vector represents a unique combination of process, object, and input/output parameters that are encountered by the system, where this representation is uniform among the set of vectors.
%Since quantifiability places no constraints on the methods based on the process calling the method, the process dimension described in Section~\ref{SubSection:VectorSpace} is omitted.
Given a system that encounters $n$ unique combinations of process, object, and input/output parameters, each method call is represented by an $n$-dimensional column vector.

%A value assignment scheme is established for the elements of the vectors representing the method calls.
The value assignment scheme is chosen such that the changes to the system state by the method calls are ``quantified.''
For all cases, let $\vec{V_i}$ be a column vector that represents method call $op_i$ in a concurrent history. Each element of $\vec{V_i}$ is initialized to 0.
\\ \textbf{Case} ($op_i \in producer~set$): Let $j$ be the position in $\vec{V_i}$ representing the combination of input parameters passed to $op_i$ and the object that $op_i$ operates on. Then $\vec{V_i}[j] = 1$.
\\ \textbf{Case} ($op_i \in consumer~set$): Let $j$ be the position in $\vec{V_i}$ representing the combination of output parameters returned by $op_i$ and the object that $op_i$ operates on. Then $\vec{V_i}[j] = -1$.
\\ \textbf{Case} ($op_i \in writer~set$): Let $j$ be the position in $\vec{V_i}$ representing the combination of input parameters passed to $op_i$ and the object that $op_i$ operates on. Let $k$ be the position in $\vec{V_i}$ representing the combination of input parameters that correspond to the previous value held by the object that is overwritten by $op_i$. 
%Create a second vector $\vec{V_i'}$ that represents method call $op_i$'s consume effect on the atomic object's previous value. Then $\vec{V_i}[j] = 1$, $\vec{V_i'}[k] = -1$, and $consumer~set = consumer~set~ \cup \vec{V_i'}$.
If $j\ne k$, then $\vec{V_i}[j] = 1$ and $\vec{V_i}[k] = -1$, else $\vec{V_i}[j] = 0$.
\\ \textbf{Case} ($op_i \in reader~set$): Let $j$ be the position in $\vec{V_i}$ representing the combination of output parameters returned by $op_i$ and the object that $op_i$ operates on. Let $\vec{I}$ be a column vector representing a \textit{read index} for each combination of output parameters returned by a reader method, where each element of $\vec{I}$ is initialized to 0. Then $\vec{I}[j] = \vec{I}[j]+1$, $\vec{V_i}[j] = - \left(~\frac{1}{2}~ \right)^{\vec{I}[j]}$.

%Each assigned value captures the effect of the method call on the system state.
%The value assignment scheme is chosen such that the changes to the system state by the method calls are ``quantified.''
%In the case for $op_i \in producer~set$, a new item whose initial value is defined by the input parameters passed to $op_i$ is added to the system. 
%Setting $\vec{V_i}[j]$ to 1, where $j$ denotes the position representing the combination of input parameters passed to $op_i$ and the object that $op_i$ operates on, captures the entrance of the new item into the system.
In the case for $op_i \in producer~set$, setting $\vec{V_i}[j]$ to 1, where $j$ denotes the position representing the combination of input parameters passed to $op_i$ and the object that $op_i$ operates on, captures the entrance of the new item into the system.
%In the case for $op_i \in consumer~set$, an item whose final value is defined by the output parameters is removed from the system. 
%Setting $\vec{V_i}[j]$ to -1, where $j$ denotes the position representing the combination of output parameters returned by $op_i$ and the item that $op_i$ operates on, captures the removal of the item from the system.
In the case for $op_i \in consumer~set$, setting $\vec{V_i}[j]$ to -1, where $j$ denotes the position representing the combination of output parameters returned by $op_i$ and the object that $op_i$ operates on, captures the removal of the item from the system.
%Establishing vectors for the producer and consumer methods using this value assignment scheme enables us to verify L2 of Definition~\ref{Definition:Quantifiability} by applying vector addition to the method call vectors and checking that each element of the resulting vector is greater than or equal to zero, since consuming an object that no longer exists in the system (i.e. a non-injective mapping) leads to an element with a negative value in the resulting vector.

In the case for $op_i \in writer~set$, an item that exists in the system is overwritten with the input parameters passed to $op_i$.
A writer method accomplishes two different things in one atomic step: 1) it consumes the previous value held by the item and 2) it produces a new value for the item. 
This state change is represented by setting the position in $\vec{V_i}$ representing the corresponding object and combination of the input parameters to be written to an item to 1 and by setting the position in $\vec{V_i}$ representing the corresponding object and combination of input parameters corresponding to the previous value held by the item to -1. 
If the input parameters corresponding to the previous value held by the item are identical to the input parameters to be written to an item, then the position in $\vec{V_i}$ representing the combination of input parameters to be written to the item is set to zero since no change has been made to the system state. 
\iffalse
Determining the previous value held by an item when performing an atomic write can be accomplished through an \texttt{atomic\_exchange} from the C++ Atomic Operations Library.
Since an \texttt{atomic\_exchange} is much more expensive than an \texttt{atomic\_store} that performs an atomic write without retrieving the previous value held by an item, it is advised that the \texttt{atomic\_exchange} only be used when verifying quantifiability.
This can be elegantly handled through a macro that either calls \texttt{atomic\_store} for normal use cases or calls \texttt{atomic\_exchange} for verification use cases.
\fi

To separate the two actions performed by the writer method into separate vectors, linear algebra can be applied to $\vec{V_i}$ in the following way.
Let $\vec{V_i}\_prod$ be the vector representing the producer effect of the writer method. 
Then $\vec{V_i}\_prod = \lfloor \left( \vec{V_i} + \vec{1} \right) \cdot ~\frac{1}{2}~\rfloor$.
Let $\vec{V_i}\_cons$ be the vector representing the consumer effect of the writer method. 
Then $\vec{V_i}\_cons = \lceil \left( \vec{V_i} + \vec{-1} \right) \cdot ~\frac{1}{2}~\rceil$.

The addition of $\vec{1}$ to $\vec{V_i}$ when computing $\vec{V_i}\_prod$ will cause all elements with a -1 value to become 0, and the multiplication of the scalar $~\frac{1}{2}~$ will revert all elements with a value of 2 back to 1.
The floor function is applied to revert elements with a value of $~\frac{1}{2}~$ back to 0.
A similar reasoning can be applied to the computation of $\vec{V_i}\_cons$.

%The value assignment scheme for writer methods enables us to verify L4 of Definition~\ref{Definition:Quantifiability} because the elements of the resulting vector obtained by adding the method call vectors will only be greater than or equal to zero if each writer method can be mapped to some method in the producer set that initially produced the object being written to by the writer methods.

In the case for $op_i \in reader~set$, $op_i$ returns the state of an item that exists in the system as output parameters. 
Multiple reads are permitted for an item with the constraint that the output parameters returned by a reader reflect a state of the item that was initialized by a producer method or updated by a writer method.
This behavior is accounted for by setting $\vec{V_i}[j] = - \left(~\frac{1}{2}~ \right)^{\vec{I}[j]}$, where $j$ denotes the position representing the corresponding object and the combination of output parameters returned by $op_i$ and $\vec{I}[j]$ represents the read count for the combination of output parameters returned by $op_i$.
The series $~\frac{1}{2} + \frac{1}{4} + \frac{1}{8} ...$ is a geometric series, where $\sum\limits_{n=1}^\infty \left(~\frac{1}{2}~ \right)^{n}=1$.
Since a concurrent history will always contain a finite number of methods, the elements of the resulting vector obtained by taking the sum of the reader method vectors will have a value in the range of $\left(-1, -~\frac{1}{2}~\right]$.
If this vector is further added with the sum of the producer method vectors and the writer method vectors $\vec{V_i}\_prod$, the elements of the resulting vector will always be greater than zero given that the output of all reader methods corresponds with a value that was either initialized by a producer method or updated by a writer method. 
%The value assignment scheme for reader methods enables us to verify L3 of Definition~\ref{Definition:Quantifiability} because when the reader method vectors are added to the producer method vectors and writer method vectors, the elements of the resulting vector will always be greater than or equal to zero given that the output parameters returned by the reader method were either initialized by a producer method or updated by a writer method. 

\begin{definition}
\label{Definition:Quantifiability}
Let $\vec{P}$ be the vector obtained by applying vector addition to the set of vectors for the producer set of history $H$. 
Let $\vec{W}\_prod$ be the vector obtained by applying vector addition to the set of vectors $\vec{V_i}\_prod$ for the writer set of history $H$. 
Let $\vec{W}\_cons$ be the vector obtained by applying vector addition to the set of vectors $\vec{V_i}\_cons$ for the writer set of history $H$. 
Let $\vec{R}$ be the vector obtained by applying vector addition to the set of vectors for the reader set of history $H$.
Let $\vec{C}$ be the vector obtained by applying vector addition to the set of vectors for the consumer set of history $H$. 
Let $\vec{H}$ be a vector with each element initialized to 0.
\\ \textbf{For each} element $i$,
\\ \textbf{if} $\left( \vec{P}[i] + \vec{W}\_prod[i] \right) \ge 1$ \textbf{then}
%\\ \indent 
$\vec{H}[i]$ = $\lceil \vec{P}[i] + \vec{W}\_prod[i] + \vec{R}[i] \rceil$  + $\vec{W}\_cons[i]$ + $\vec{C}[i]$ 
%\\ \textbf{if} $\vec{P}[i] \ge 1$, 
%\\ \indent $\vec{H}[i]$ = $\lceil \vec{P}[i] + \vec{R}[i] \rceil$ + $\vec{W}\_prod[i]$ + $\vec{W}\_cons[i]$ + $\vec{C}[i]$, 
%\\ \textbf{else if} $\vec{W}\_prod[i] \ge 1$, 
%\\ \indent $\vec{H}[i]$ = $\lceil \vec{W}\_prod[i] + \vec{R}[i] \rceil$ + $\vec{P}[i]$ + $\vec{W}\_cons[i]$ + $\vec{C}[i]$, 
\\ \textbf{else} 
%\\ \indent 
$\vec{H}[i]$ = $\vec{P}[i]$ + $\vec{W}\_prod[i]$ + $\vec{W}\_cons[i]$ + $\vec{R}[i]$ + $\vec{C}[i]$.
\vspace{3pt}
\\ History $H$ is \textit{quantifiable} if for each element $i$, $\vec{H}[i] \ge 0$.
\end{definition}

%TODO: Explain reasoning behind formal definition of quantifiability
Informally, if all vectors representing the methods in the method call set of history $H$ are added together, the value of each element should be greater than or equal to zero.
%This property indicates that the net effect of all methods invoked upon the system is compliant with the requirement that no duplicate or non-existent items have been removed from the system (indicated by an element in $\vec{H}$ that is less than zero) and no duplicate items have been inserted into the system (indicated by an element in $\vec{H}$ that is greater than one).
This property indicates that the net effect of all methods invoked upon the system is compliant with the conservation requirement that no non-existent items have been removed, updated, or read from the system.
%In other words, all method calls are ``conserved.''
%There is, however, a challenging dilemma that must be addressed to guarantee that each element in the sum of the method call vectors in a quantifiable history is greater than or equal to zero.
%The dilemma being how to handle a method that observes the state of the system but does not change the state of the system, i.e. the reader method.
The values of the vectors for the reader method are assigned such that each element in the sum of the reader method vectors is always greater than -1.
As long as the output of the reader method is equivalent to the input of a producer method or writer method, then the reader method has observed a state of the system that corresponds to the occurrence of a producer method or writer method.
The ceiling function is applied to $\vec{P}[i] + \vec{W}\_prod[i] + \vec{R}[i]$ if $\left( \vec{P}[i] + \vec{W}\_prod[i] \right) \ge 1$ which yields a value that is also $\ge 1$.
Once the reader method vectors have been added appropriately, the remaining method call vectors can be directly added to compute $\vec{H}$ for history $H$.

If any element of $\vec{H}$ is less than zero, then a consume action has been applied to either an item that does not exist in the system state (the item was previously consumed) or an item that never existed in the system state (the item was never produced or written), which is not quantifiable due to a violation of Principle~1.

A notable difference between defining correctness as properties over a set of vectors and defining correctness as properties of sequential histories is the growth rate of a set of vectors versus sequential histories when the number of methods called in a history is increased.
The size of a set of vectors grows at the rate of $O(n)$ with respect to $n$ methods called in a history.
The number of sequential histories grows at the rate of $O(n!)$ with respect to $n$ methods called in a history.
This leads to significant time cost savings when verifying a correctness condition defined as properties over a set of vectors since analysis of $n$ $c$-dimensional vectors using linear algebra can be performed in $O(n + c)$ time ($n$ time to assign values and compute separate vectors that each represent a sum of the producer, consumer, writer, and reader method call vectors, and $c$ time to add the elements of the vectors representing the sum of the producer, consumer, writer, and reader method call vectors).

\section{Proving that a Concurrent History is Quantifiably Correct with Tensor Representation}
\label{MatrixAnalysis}
%% Justin
A tensor is the higher-dimension generalization of the matrix. Just as matrices are composed of rows and columns, tensors are composed of \textit{fibers}, obtained by fixing all indices of the tensor except for one. The \textit{order} of a tensor is the number of dimensions. 

When proving that a concurrent history is quantifiable, it is useful to reshape the history vector into a higher-order tensor. Any tensor of order $d$, including order-1 tensors (vectors), can be reshaped into a tensor of higher order $m$, where $m > d$. Such a reshaping is known as the \textit{tensorization} or \textit{folding} of the original tensor. There are many tensorization techniques for vectors based on the desired structure of the resultant tensor \cite{debals2015stochastic}. Here, we use the segmentation technique to map consecutive segments of the vector to the tensor. In particular, we follow the method of Grasedyck \cite{grasedyck2010polynomial}; given a vector $x \in \mathbb{R}^{I_1 \cdots I_N}$, we define the bijection
\[
\mu~:~\mathbb{R}^{I_{1} \cdots I_{N}} \mapsto \mathbb{R}^{I_{1} \times \cdots \times I_{N}}
\]
for all indices $i_d \in \{1,\ldots,I_d\}$, $d=1,\ldots,N$, by
\[
(\mu(x))_{i_1,\ldots,i_{N}} \mapsto (x)_{j},
\]
where
\[
j = i_1 + \sum\limits_{k=2}^{N}(i_k - 1)\prod\limits_{m=1}^{k-1}I_m.
\]
This mapping maps each segment of $I_1$ consecutive vector elements of $x$ to each mode-1 fiber of the tensor $\mu(x)$.

Let $H \in \mathbb{R}^{I \cdot O \cdot P \cdot M}$ be the concurrent history vector for a given system. The general concurrent history tensor $\mathcal{H}$ is then obtained by $\mathcal{H} = \mu(H) \in \mathbb{R}^{I \times O \times P \times M}$. Because the value assignment scheme provides scalar quantities to the method calls, it can be useful to eliminate the inside dimension $M$ by summation, yielding a 3-way tensor $\mathcal{H}_{iop} \in \mathbb{R}^{I \times O \times P}$ which is the net result of method calls for each process for every object-item pair.  The process dimension may further be eliminated by summation and the order-2 tensor (matrix) $\mathcal{H}_{io} \in \mathbb{R}^{I \times O}$ is the net result of all method calls for every object-item pair.  This matrix represents a quantifiable history if and only if all of the resulting elements are greater than or equal to zero.  

In summary, quantifiability can be determined by tensorizing the history vector into an order-4 tensor and summing along the method and process dimensions to flatten it into a matrix. If all the values in this matrix are non-negative then the history is quantifiable. Other properties can be shown. For example, summing the absolute values along the method and process dimensions creates a heatmap of the busy object-item pairs in the history.

\section{Proving that a Data Structure is Quantifiably Correct}
To prove that a concurrent data structure is quantifiability correct, it must be shown that each of the methods preserve atomicity (the method takes effect entirely or not at all), isolation (the method's effects are indivisible), and conservation (every method call either completes successfully, remains pending, or is explicitly cancelled).
There is a large body of research on automatic verification of atomicity for transactions or method calls in a concurrent history, including an atomic type system~\cite{flanagan2003type}, inference of operation dependencies~\cite{flanagan2008velodrome}, dynamic analysis tools~\cite{flanagan2004atomizer} based on Lipton's theory of reduction~\cite{lipton1975reduction}, and modular testing of client code~\cite{shacham2011testing}.
There are fewer techniques presented in literature for proving atomicity and isolation for a concurrent object. 
Such techniques include Lipton's theory of reduction~\cite{lipton1975reduction} for reasoning about sequences of statements that are indivisible, occurrence graphs that represent a single computation as a set of interdependent events~\cite{best1981formal}, Wing's methodology~\cite{wing1989verifying} for demonstrating that a concurrent object's behavior is equivalent to its sequential specification, and simulation mappings between the implementation and specification automata~\cite{chockler2005proving}.

Proving that a concurrent object is linearizable~\cite{Herlihy1990-ms} requires an abstraction function~$\mathcal{A}:REP \rightarrow ABS$ to be defined, where $ABS$ is an \textit{abstract type} (the type being implemented), $REP$ is a \textit{representation type} (the type used to implement $ABS$), and $\mathcal{A}$ is defined for the subset of $REP$ values that are legal representations of $ABS$.
An implementation~$\rho$ of an abstract operation~$\alpha$ is shown to be correct by proving that whenever~$\rho$ carries one legal $REP$ value~$r$ to another~$r$', $\alpha$ carries the abstract value from $\mathcal{A}(r)$ to $\mathcal{A}$($r$').

Since Lipton's approach ~\cite{lipton1975reduction} is focused on lock-based critical sections, occurrence graphs~\cite{best1981formal} do not model data structure semantics, and Wing's approach~\cite{wing1989verifying}, simulation mappings~\cite{chockler2005proving}, and formal proofs of linearizability require reference to sequential histories, they are not sufficient for proofs of quantifiability.
However, informal proofs of linearizability reason about program correctness by identifying a single instruction for each method in which the method call takes effect, referred to as a \textit{linearization point}.
Proving that a data structure is quantifiably correct can be performed in a similar fashion by defining a \textit{visibility point} for each method.
A visibility point is a single instruction for a method in which the entire effects of the method call become visible to other method calls.
Unlike a linearization point, a visibility point does not need to occur at some instant between a method call's invocation and response.

Establishing a visibility point for a method demonstrates that its effects preserve atomicity and isolation, but it still remains to be shown that the method call's effects are conserved.
A method call's effects are conserved if it returns successfully or its pending request is stored in the data structure and will be fulfilled by a future method call.
The proof for conservation of method calls requires demonstrating that 1) a method completes its operation on the successful code path and 2) a method's pending request is stored in the data structure on the unsuccessful code path.
Additionally, statements must be provided for each method that prove that its invocation is guaranteed to fulfill a corresponding pending request if one exists.

%TODO Verification Algorithm
\section{Verification Algorithm}

Algorithm~\ref{alg:verification} presents the verification algorithm for quantifiability derived from the corresponding formal definition. 
Line~\ref{alg:verification}.\ref{l:MAX} is defined as a constant that is the total number of unique configurations comprising the process, object, and input/output that are encountered by the system.
The $P$ array tracks the running sum of the producer method vectors.
The $W\_prod$ array tracks the running sum of the new values written by the writer method vectors.
The $W\_cons$ array tracks the running sum of the previous values overwritten by the writer method vectors.
The $R$ array tracks the running sum of the reader method vectors.
The $C$ array tracks the running sum of the consumer method vectors.
The $I$ array tracks the running sum of the read index for the reader method vectors.
The $H$ array tracks the final sum of all method call vectors.
The \textsc{VerifyHistory} function accepts a set of method calls as an argument on line~\ref{alg:verification}.\ref{l:verify}.
The for-loop on line~\ref{alg:verification}.\ref{l:forloop1} iterates through the methods in the method call set and adds a value to the appropriate array according to the value assignment scheme discussed in Section~\ref{SubSection:FormalDefinition}.
\iffalse
and summarized as follows.
\\ \textbf{Producer}: For each producer method encountered in the method call set that inserts an item with configuration $i$ into the system, $P[i]$ is incremented by one on line~\ref{alg:verification}.\ref{l:producer}. 
At the end of the for-loop on line~\ref{alg:verification}.\ref{l:forloop1}, $P[i]$ is the total number of items with configuration $i$ that have entered the system.
\\ \textbf{Consumer}: For each consumer method encountered in the method call set that removes an item with configuration $i$ from the system, $C[i]$ is decremented by one on line~\ref{alg:verification}.\ref{l:consumer}. 
At the end of the for-loop on line~\ref{alg:verification}.\ref{l:forloop1}, $C[i]$ is the total number of items with configuration $i$ that have been removed from the system.
\\ \textbf{Writer}: For each writer method encountered in the method call set that updates an item in the system with a previous configuration of $j$ to a new configuration $i$, $W\_prod[i]$ is incremented by one (produce effect) on line~\ref{alg:verification}.\ref{l:writerprod} and $W\_cons[j]$ is decremented by one (consume effect) on line~\ref{alg:verification}.\ref{l:writercons}. 
\\ \textbf{Reader}: For each reader method encountered in the method call set that reads an item with configuration $i$ in the system, the read count $I[i]$ is incremented by one on line~\ref{alg:verification}.\ref{l:readcount} and $R[i]$ is decremented by $\left( ~\frac{1}{2}~ \right)^{I[i]}$ on line~\ref{alg:verification}.\ref{l:reader}.
\fi

The for-loop on line~\ref{alg:verification}.\ref{l:forloop2} iterates through each of the configurations and sums the method call vectors according to Definition~\ref{Definition:Quantifiability} to obtain the final vector $\vec{H}$.
If any element of $\vec{H}$ is less than zero or greater than one (line~\ref{alg:verification}.\ref{l:check}), then the history is not quantifiable.
Otherwise, if all elements of $\vec{H}$ are greater than or equal to zero, then the history is quantifiable.
%The correctness of the verification algorithm is presented in Appendix~\ref{Appendix:Correctness}.

%\begin{wrapfigure}{R}{0.5\textwidth}
%\begin{minipage}{0.5\textwidth}
\begin{algorithm}
\footnotesize
	\caption{Quantifiability Verification}
	\label{alg:verification}
	\begin{algorithmic}[1]
	    \State $\textbf{\#define}~MAX~constant$ \Comment{Total number of process/object/input/output configurations} \label{l:MAX}
	    \State $\textbf{int}~P[MAX]$, $W\_prod[MAX]$, $W\_cons[MAX]$, $R[MAX]$, $C[MAX]$, $I[MAX]$, $H[MAX]$
		\Function{VerifyHistory} {$\textbf{set}~methods$} \label{l:verify}
		\State $\textbf{set \textless Method \textgreater::iterator}~it$
		\For{$it=methods.begin();~it!=methods.end();~++it$} \label{l:forloop1}
		    \If {$it.type == Producer$} 
		        \State $\textbf{int}~j = \Call{ParamsToIndex}{it.object, it.input}$
		        \State $P[j] = P[j] + 1$ \label{l:producer}
		    \ElsIf{$it.type == Writer$}
		        \State $\textbf{int}~j = \Call{ParamsToIndex}{it.object, it.input}$
		        \State $\textbf{int}~k = \Call{ParamsToIndex}{it.object, it.prevVal}$
		        \State $W\_prod[j] = W\_prod[j] + 1$ \label{l:writerprod}
		        \State $W\_cons[k] = W\_cons[k] - 1$ \label{l:writercons}
		    \ElsIf{$it.type == Reader$}
		        \State $\textbf{int}~j = \Call{ParamsToIndex}{it.object, it.output}$
		        \State $I[j] = I[j] + 1$ \label{l:readcount}
		        \State $R[j] = R[j] - \left( ~\frac{1}{2}~ \right) ^{I[j]}$ \label{l:reader}
		    \ElsIf {$it.type == Consumer$} 
		        \State $\textbf{int}~j = \Call{ParamsToIndex}{it.object, it.output}$
		        \State $C[j] = C[j] - 1$ \label{l:consumer}
		    \EndIf
		\EndFor
		\For{$\textbf{int}~i = 0;~i<MAX;~i++$} \label{l:forloop2}
		    \If { $\left( P[i] + \vec{W}\_prod[i] \right) \ge 1$ }
		        \State $\vec{H}[i]$ = $\lceil \vec{P}[i] + \vec{W}\_prod[i] + \vec{R}[i] \rceil$ + $\vec{W}\_cons[i]$ + $\vec{C}[i]$
		    %\ElsIf { $W\_prod[i] \ge 1$ }
		    %    \State $\vec{H}[i]$ = $\lceil \vec{W}\_prod[i] + \vec{R}[i] \rceil$ + $\vec{P}[i]$ + $\vec{W}\_cons[i]$ + $\vec{C}[i]$
		    \Else
		        \State $\vec{H}[i]$ = $\vec{P}[i]$ + $\vec{W}\_prod[i]$ + $\vec{W}\_cons[i]$ + $\vec{R}[i]$ + $\vec{C}[i]$
		    \EndIf
		    %\If {$\left( \vec{H}[i] < 0 \right) || \left( \vec{H}[i] > 1 \right)$} \label{l:check}
		    \If {$\vec{H}[i] < 0 $} \label{l:check}
		        \State \Return \FALSE
		    \EndIf
		\EndFor
		\State \Return \TRUE
		\EndFunction
	\end{algorithmic}
\end{algorithm}
%\end{minipage}
%\end{wrapfigure}

\subsection{Time Complexity of Verification Algorithm}
Let $n$ be the total number of methods in a history and let $c$ be the total number of configurations determined according to the input/output of each method and the object to be invoked on by the method. 
The for-loop on line~\ref{alg:verification}.\ref{l:forloop1} takes $O(n)$ time to iterate through all methods in the method call set.
The for-loop on line~\ref{alg:verification}.\ref{l:forloop2} takes $O(c)$ time to iterate through all possible configurations.
Let $i$ be the total number of input/output combinations and let $j$ be the total number of objects.
The total number of configurations is $i \cdot j$.
Therefore, the total time complexity of \textsc{VerifyHistory} is $O(n + i \cdot j)$.

\section{Properties of Quantifiability}
The system model presented in Section~\ref{SubSection:SystemModel} is mapped to a vector space.  We do not claim that the axioms of a vector space hold for all possible concurrent systems.  We do propose a mapping from most concurrent systems to the mathematical ideal of a vector space.  Concurrent systems fitting the model define a vector space and their histories are the vectors in that space. For concurrent systems fitting the model, properties of a vector space become axiomatic and have a variety of uses.

\subsection{Compositionality} 

To show compositionality, it must be shown that the composition of two quantifiable histories is quantifiable, and that the decomposition of histories, i.e. the projection of the history on any of its obejcts, is also a quantifiable history.  This is formally stated in the following theorem.

\begin{theorem}
History $H$ is quantifiable if and only if, for each object $x$, $H|x$ is quantifiable.
\end{theorem}
\begin{proof}
It first must be shown that if each history $H|x$ for object $x$ is quantifiable, then history $H$ is quantifiable.
Since the addition of quantifiable histories is closed under addition, it follows that the composition of quantifiable object subhistories $H|x$ is also quantifiable.
Therefore, $H$ is quantifiable.

It now must be shown that if history $H$ is quantifiable, then each history $H|x$ for object $x$ is quantifiable.
Since $H$ is quantifiable, then each element of the vector $\vec{H} \ge 0$.
Each position in $\vec{H}$ corresponds to a unique configuration representing the process, object, and input/output that the method is invoked upon.
Since each element of the vector $\vec{H} \ge 0$, then for each element $i$ associated with object $x$, $\vec{H}[i] \ge 0$.
Each history $H|x$ for object $x$ is therefore quantifiable.
\end{proof}

\subsection{Non-Blocking and Non-Waiting Properties}
\label{SubSection:NonBlocking}
A correctness condition may inherently cause blocking, as is the case with serializability applied to transactions \cite{herlihy1990linearizability}.  Quantifiability shares with linearizability the non-blocking property, and for the same reason: it never forces a process with a pending invocation to block.
\textit{Lock-freedom} is the property that some thread is guaranteed to make progress. \textit{Wait-freedom} is the property that all threads are guaranteed to make progress.
%Quantifiable data structures can be implemented in a way that satisfies either lock-freedom or wait-freedom.
%The strategies for guaranteeing lock-freedom or wait-freedom in quantifiable data structures is no different than non-blocking data structures presented in literature~\cite{Hendler2004-ho, zhang2016efficient, barrington2015scalable, feldman2016efficient}.
Quantifiability is compatible with the existing synchronization methods for lock-freedom and wait-freedom because it is a non-blocking correctness property.

The requirement that all methods must succeed or be explicitly cancelled raises the question of how this is non-blocking.  Indeed a thread might choose to block if there is no way it can proceed without the return value or the state change resulting from the method.  It is a matter for the application to decide, not an inherent property of Quantifiability Principle~1.  For example, consider thread 1 calling a $pop$ method on a concurrent stack, $T1:s.pop() \rightarrow x$.  This can be written as \texttt{<Type> v = s.pop();} which is blocking in C.  Or it may be invoked as a call by reference in the formal parameters \texttt{s.pop(<Type> \&v);} which is non-blocking.  
The second invocation also permits the thread to block if desired by spinning on the address to check if a result is available.  If address \&v is not pointing to a value of \texttt{<Type>}, the method has not yet succeeded.   
Alternatively, instead of spin-waiting, a thread can do a ``context switch'' and proceed with other operations while waiting for the pending operation to succeed.
The thread can still perform other operations on the same data structure despite an ongoing pending operation.
Since quantifiability does not enforce program order, it is possible for operations called by the same thread to be executed out-of-order.
And if the thread decides the method is no longer needed, it can be cancelled.

The concept of retrieving an item to be fulfilled at a later time is implemented in C++11, C\#, and Java as \textit{promises} and \textit{futures}~\cite{Stroustrup2013-yo}. 
%If an item is desired but not yet available to be accessed, a promise object can be created for this item that will be set to a value by a thread at some point in the future.
%Each promise object has a future object associated with it that enables a thread to retrieve the item set by the promise object.
The \texttt{async} function in C++ calls a specified function and returns a future object without waiting for the specified function to complete.
The return value of the specified function can be accessed using the future object.
The \texttt{wait\_for} function in C++ is provided by the future object that enables a thread to wait for the item in the promise object to be set to a value for a specified time duration.
Once the \texttt{wait\_for} function returns a \textit{ready} status, the item value is retrieved by the future object through the \texttt{get} function.
The disadvantage of the \texttt{get} function is that it blocks until the item value is set by the promise object.
Once the \texttt{get} function returns the item, the future object is no longer valid, leading to undefined behavior if other threads invoke \texttt{get} on this future object.

Due to the semantics of the \texttt{get} function, 
%it is advised to 
we advise implementing retrieval of an item to be fulfilled at a later time with a shared object used to announce information, 
%for an operation
referred to as a \textit{descriptor object}~\cite{harris2002practical, Dechev2006-cb}.
%A discussion on progress guarantees is provided below.
%This can be achieved by using a descriptor object to create a pending item to be fulfilled by another thread. 
Once the pending item is fulfilled, it is updated to a non-pending item using the atomic instruction Compare-And-Swap (CAS). 
CAS accepts as input a memory location, an expected value, and an update value. If the data referenced by the memory location is equivalent to the expected value, then the data referenced by the memory location is changed to the update value and true is returned. Otherwise, no change is made and false is returned.
Since CAS will only fail if another thread successfully updates the data referenced by the memory location, quantifiability can be achieved in a lock-free manner.
%with \texttt{s.\_pop(<Type> \&v)}.   The address, formal argument name and calling process identify the method call.  

Wait-freedom is typically achieved using helping schemes in conjunction with descriptor objects to announce an operation to be completed in a table such that all threads are required to check the announcement table and help a pending operation prior to starting their own operation~\cite{kogan2012methodology}. 
However, as a consequence of the relaxed semantics allowed by quantifiability, contention avoidance can be utilized (discussed in Section~\ref{SubSection:QStack}) that allows threads to make progress on their own operations without interference from other threads.

\iffalse
The length of time a method call may remain pending is a progress condition.  For example wait-free progress requires that the method call would succeed in a finite amount of time.  The idea that method calls can happen at the same time is not so hard to accept because they may fulfill each other, in the manner of the elimination backoff stack \cite{Hendler2004-ho}.  Accepting that they may happen in any order is to accept the reality that since the days of the shared bus are gone, there is nothing like an ``external observer'' to determine if a history is acceptable.  Latency is a fact of parallel and distributed systems.  Quantifiable systems may include synchronization schemes, important in many domains, but the correctness condition should be that methods succeed, whatever may be done to give an illusion of real-time ordering. 
\fi

\iffalse
Transactions require multiple atomic steps to complete~\cite{dechev2006lock,zhang2016lock}.
For these circumstances, a descriptor object can be used to post the required atomic steps for an operation and a helping scheme can be employed to enable other threads to help complete the pending operation in a lock-free manner.
\fi

\iffalse
The proof of the non-blocking property for quantifiability is presented below.
%in Appendix~\ref{Appendix:NonBlocking}.
\fi

\subsection{Proof of Non-Blocking Property}
\label{NonBlocking}

The following proof shows that quantifiability is non-blocking; that is, it does not require that it wait for another pending operation to complete.

\begin{theorem}
Let $inv$ be an invocation of a method $m$. If $\langle x~inv~P \rangle$ is a pending invocation in a quantifiable history $H$ with a corresponding vector $\vec{H}$, then either there exists a response $\langle x~res~P \rangle$ such that either $H \cdot \langle x~res~P \rangle$ is quantifiable or $H \setminus \langle x~inv~P \rangle$ is quantifiable.
\end{theorem}
\begin{proof}
If method $m$ is a producer method that produces item with configuration $i$, then there exists a response $\langle x~res~P \rangle$ such that $H \cdot \langle x~res~P \rangle$ is quantifiable because $\vec{H}[i] + 1$ is greater than zero since $\vec{H}[i] \ge 0$ by the definition of quantifiability.
If method $m$ is a consumer method that consumes an item with configuration $i$, then a response $\langle x~res~P \rangle$ exists if $\vec{H}[i] \ge 1$.
If method $m$ is a writer method that updates item with configuration $i$ to a new configuration $j$, then a response $\langle x~res~P \rangle$ exists if $\vec{H}[i] \ge 1$ and $\vec{H}[j] \ge 0$.
If method $m$ is a reader method that reads an item with configuration $i$, then a response $\langle x~res~P \rangle$ exists if $\vec{H}[i] \ge 1$.
If a response for method $m$ does not exist, then method $m$ can be cancelled.
Upon cancellation, $\langle x~inv~P \rangle$ is removed from history $H$.
Since quantifiability places no restrictions on the behavior of pending method calls, $H \setminus \langle x~inv~P \rangle$ is quantifiable.
\end{proof}

\section{Related Work}
%Quantifiability makes sense in light of recent advances in concurrency research.  
Quantifiability is motivated by recent advances in concurrency research.  
Frequently cited works \cite{TSQueue2014, Hendler2004-ho} are already moving in the direction of the two principles stated in Section~\ref{Section:Principles}.   This section places quantifiability in context of several threads of research: the basis of concurrent correctness conditions, complexity of proving correctness and design of related data structures.  
%Quantifiability builds on previous work with the goal of enabling formal reasoning about correctness that is not tied to sequential thinking, less complex correctness proofs, reliable progress guarantees, and higher performance data structures.

\subsection{Relationship to Other Correctness Conditions}
\label{RelatedWork:CorrectnessCondition}
%A survey paper described all the correctness conditions in general use today...
%Quantifiability is different from all of them...
A \textit{sequential specification} for an object is a set of sequential histories for the object~\cite{Herlihy2011-yj}.
A sequential history $H$ is \textit{legal} if each subhistory in $H$ for object $x$ belongs to the sequential specification for $x$. 
Many correctness conditions for concurrent data structures are proposed in literature~\cite{papadimitriou1979serializability,Lamport1979-rd,herlihy1990linearizability,Herlihy2011-yj,aspnes1994counting,Afek2010-xv,ou2017checking}, all of which reason about concurrent data structure correctness by demonstrating that a concurrent history is equivalent to a legal sequential history.

\iffalse
Reasoning about concurrent data structures in terms of a legal sequential history is appealing because it builds upon previous techniques for verifying correctness of sequential objects~\cite{guttag1978abstract, hoare1978proof}.
However, generating all possible legal sequential histories of a concurrent history has a worst case time complexity of $O(n!)$, leading to inefficient verification techniques for determining correctness of concurrent data structures.
%A brief overview of correctness conditions for concurrent data structures in literature is provided, followed by an explanation of how quantifiability distinguishes itself from the other correctness conditions.
\fi

%Serializability ~\cite{papadimitriou1979serializability} is a correctness condition for database systems that requires that a sequence of updates/retrievals by an arbitrary number of atomic users to appear as though the users took turns, in some order, each executing their entire transaction indivisibly. 
%The formal definition provided for serializability states that a history $h$ is serializable if and only if there is a serial history $h_s$ such that $h$ is equivalent to $h_s$. 
Serializability ~\cite{papadimitriou1979serializability} is a correctness condition such that a history $h$ is serializable if and only if there is a serial history $h_s$ such that $h$ is equivalent to $h_s$.
A history $h$ is \textit{strictly serializable} if there is a serial history $h_s$ such that $h$ is equivalent to $h_s$, and an atomic write ordered before an atomic read in $h$ implies that the same order be retained by $h_s$. 
Papadimitriou draws conclusions implying that there is no efficient algorithm that distinguishes between serializable and non-serializable histories~\cite{papadimitriou1979serializability}.
\iffalse
A \textit{decision problem} is a problem that is a yes-no question of the input values. A decision problem is in the \textit{Nondeterministic Polynomial time} (NP) complexity class if ``yes'' answers to the decision problem can be verified in Polynomial time (P). A decision problem $x$ is in \textit{NP-complete} if $x$ is in NP, and every problem in NP is reducible to $x$ in polynomial time. Papadimitriou proved that testing whether a history is serializable is NP-complete. An implication of this result is that, unless $P = NP$, there is no efficient algorithm that distinguishes between serializable and non-serializable histories~\cite{papadimitriou1979serializability}.
\fi

Sequential consistency~\cite{Lamport1979-rd} is a correctness condition for multiprocessor programs such that the result of any execution is the same as if the operations of all processors were executed sequentially, and the operations of each individual processor appear in this sequence in program order. 
%Lamport indicates that the sequential consistency of an individual process does not guarantee that the program as a whole is sequentially consistent.
%Lamport proposes that compositionality for sequentially consistent objects can be achieved by requiring that the memory requests from all processors issued to an individual memory module are serviced from a single FIFO queue which guarantees a total ordering for the program execution.
Since sequential consistency is not compositional, Lamport proposes that compositionality for sequentially consistent objects can be achieved by requiring that the memory requests from all processors are serviced from a single FIFO queue. 
However, a single FIFO queue is a sequential bottleneck that limits the potential concurrency for the entire system.

%Linearizability~\cite{herlihy1990linearizability} is a correctness condition for concurrent objects that captures the real-time ordering of methods invoked upon the concurrent object. 
%A history $h$ is linearizable if $h$ is equivalent to a legal sequential history, and each method call appears to take effect instantaneously at some moment between its invocation and response.
Linearizability~\cite{herlihy1990linearizability} is a correctness condition such that a history $h$ is linearizable if $h$ is equivalent to a legal sequential history, and each method call appears to take effect instantaneously at some moment between its invocation and response.
Herlihy et al.~\cite{Herlihy2011-yj} suggest that linearizability can be informally reasoned about by identifying a linearization point in which the method call appears to take effect at some moment between the method's invocation and response.
Identifying linearization points avoids reference to a legal sequential history when reasoning about correctness, but such reasoning is difficult to perform automatically.
%mention how this avoids reference to a legal sequential history, but is difficult to perform automatically
Herlihy et al.~\cite{herlihy1990linearizability} compared linearizability with strict serializability by noting that linearizability can be viewed as a special case of strict serializability where transactions are restricted to consist of a single operation applied to a single object. 
Comparisons of correctness conditions for concurrent objects to serializability are made in a similar manner by considering a special case of serializability where transactions are restricted to consist of a single method applied to a single object.
\iffalse
Herlihy compared linearizability with serializability by grouping multiple method calls into a transaction making the histories strictly serializable.  Here we consider the serializability of method calls, not transactions.   
\fi

Quiescent consistency~\cite{aspnes1994counting} is a correctness condition for counting networks that establishes a safety property for a network of two-input two-output computing elements such that the inputs will be forwarded to the correct output wires at any quiescent state. 
A \textit{step property} is defined to describe a property over the outputs that is always true at the quiescent state.
%Each computing element is referred to as a \textit{balancer}.
%The authors prove that a balancing network with $m$ balancers satisfies the step property in all executions if and only if it satisfies it in all sequential executions in which at most $2^m$ tokens traverse the network.
%A lemma is presented that extends a balancing network to a counting network.
Quasi-linearizability~\cite{Afek2010-xv} builds upon the formal definition of linearizability to include a sequential specification of an object that is extended to a larger set that includes sequential histories that are not legal, but are within a bounded distance $k$ from a legal sequential history. 
\iffalse
Ou et al.~\cite{ou2017checking} present non-deterministic linearizability, a correctness model for concurrent data structures that utilize the relaxed semantics of the C/C++ memory model. 
A specification $S$ for a concurrent object $O_c$ comprises a non-deterministic, sequentialized version of the concurrent object $O_s$ and an admissibility function that permits two methods to be unordered in the concurrent execution.
A \textit{valid sequential history} is a sequential history that is legal in the C/C++ memory model. 
%Since the sequentialized version of the concurrent object is non-deterministic, multiple behaviors are admitted for a given method call.
The notion of justified behaviors is introduced to account for method calls exhibiting non-deterministic behavior in a concurrent history.
A concurrent object $O_c$ is non-deterministic linearizable on a valid sequential history $H$ for a specification $S$ if and only if each method call in $H$ is justified and returns a value as specified by its non-deterministic specification.
A concurrent object $O_c$ is non-deterministic linearizable on an execution $E$ for a specification $S$ if and only if $O_c$ is non-deterministic linearizable for all valid sequential histories of $E$ for $S$.
\fi

Unlike the correctness conditions proposed in literature, quantifiability does not define correctness of a concurrent history by referencing an equivalent legal sequential history.
Quantifiability requires that the method calls be conserved, enabling correctness to be proven by quantifying the method calls and applying linear algebra to the method call vectors. %TODO May need to revise this sentence
This fundamental difference enables quantifiability to be verified more efficiently than the existing correctness conditions because applying linear algebra can be performed in $O(n$ + $c)$ time ($n$ is the number of method calls and $c$ is the number of process, object, and input/output configurations), while deriving the legal sequential histories has a worst case time complexity of $O(n!)$.

\subsection{Proving Correctness}
Verification tools are proposed~\cite{vechev2009experience,burckhardt2010line,zhang2015round,ou2017checking} to enable a concurrent data structure to be checked for correctness according to various correctness conditions.
Vechev et al.~\cite{vechev2009experience} present an approach for automatically checking linearizability of concurrent data structures. 
%The approach incorporates two methods for checking linearizability: 1) automatic linearization and 2) linearization points. 
%The automatic linearization technique searches for a permitted linearization by exploring all permutations of a concurrent execution.
%The linearization point technique accepts user-provided annotations indicating the linearization points of an algorithm so that an order between overlapping method calls can be determined, eliminating the need to explore all permutations of a concurrent execution.
Burckhardt et al.~\cite{burckhardt2010line} present Line-Up, a tool that checks deterministic linearizability automatically. 
%Line-up derives a sequential specification for a concurrent data structure by recording all sequential histories for a finite test and checking that each concurrent execution is equivalent to one of the recorded sequential histories.
Zhang et al.~\cite{zhang2015round} present Round-up, a runtime verification tool for checking quasi-linearizability violations of concurrent data structures.
%Round-up systematically executes a test program with a standard sequential data structure to compute all possible legal sequential histories.
%Each concurrent history is first checked for linearizability by determining if it is equivalent to a legal sequential history. 
%If the concurrent history is linearizable, it is also quasi-linearizable by definition.
%Otherwise, the order of the method calls in each legal sequential history are rearranged up to a specified distance $k$ to generate a quasi-linearization in which the concurrent history is compared with to determine if it is quasi-linearizable.
Ou et al.~\cite{ou2017checking} develop a tool that checks non-deterministic linearizability for concurrent data structures designed using the relaxed semantics of the C/C++ memory model.
%Ordering point annotations are utilized to identify a specific instruction in which the effects of a method become visible to the system.
%During model checking, the method calls are organized in a directed acyclic graph where a method $m_1$ has a directed edge to $m_2$ if $m_1$'s ordering point occurs before $m_2$'s ordering point in the generated concurrent history.
%A topological sort of the directed acyclic graph yields a legal sequential history for the concurrent history.
%To reduce the large time complexity required to derive all possible legal sequential histories, an option is provided to the user that randomly generates and checks a customized number of legal sequential histories.

These verification tools are all faced with the computationally expensive burden of generating all possible legal sequential histories of a concurrent history since this is the basis of correctness for the correctness conditions in literature.
The correctness verification tools presented by Vechev et al.~\cite{vechev2009experience} and Ou et al.~\cite{ou2017checking} accept user annotated linearization points to eliminate the need for deriving a legal sequential history for a reordering of overlapping methods.
Although this optimization is effective for fixed linearization points, it could potentially miss valid legal sequential histories for method calls with non-fixed linearization points in which the linearization point may change based on overlapping method calls. 

%Observational refinement captures program correctness with respect to a sequential specification such that given two libraries $L_1$ and $L_2$ implementing the methods of some concurrent object, $L_1$ refines $L_2$ if and only if every computation of every program using $L_1$ would also be possible if $L_2$ were used instead.
Bouajjani et al.~\cite{bouajjani2015tractable} present an approximation-based approach for detecting observational refinement violations that assigns intervals to method calls such that a method call $m_1$ happens before method call $m_2$ if $m_1$'s interval ends before $m_2$'s interval ends.
This approach is able to detect observational refinement violations in polynomial time, but it suffers from enumeration of possible histories constrained by the interval order for an execution.
Emmi et al.~\cite{emmi2015monitoring} present a verification technique for observational refinement that uses symbolic reasoning engines instead of explicit enumerations of linearizations.
This technique is limited to atomic collections, locks, and semaphores.
Sergey et al.~\cite{sergey2016hoare} propose a Hoare-style logic for reasoning about the program's inputs and outputs directly without referencing sequential histories.
Nanevski et al.~\cite{nanevski2019specifying} apply structure-preserving functions on resources to achieve proof reuse in separation logic.
The authors use their proposed logic to reason about program correctness using the heap rather than sequential histories.
While the Hoare-style specifications~\cite{sergey2016hoare} and structure preserving functions~\cite{nanevski2019specifying} are tailored for the non-linearizable objects they are describing, Quantifiability is designed such that correctness can be verified automatically and efficiently for arbitrary abstract data types.

\subsection{Design of Related Data Structures}
%Concern (1) given an undefined object state, do we want to return an error indication right away or to wait for a matching push?
Several data structure design strategies are presented in literature that motivate the principles of quantifiability.
The concern of defining the behavior of partial methods when reaching an undefined object state is addressed by dual data structures~\cite{scherer2004nonblocking}.
Dual data structures are concurrent object implementations that hold reservations in addition to data to handle conditional semantics.
%Partial methods are divided into a request method and a follow-up method that is used to determine if the request has been fulfilled.
%If the dual data structure is empty, a request is inserted into the data structure to be fulfilled by another thread.
Dual data structures are linearizable and can be implemented to provide non-blocking progress guarantees including lock-freedom, wait-freedom, or obstruction-freedom.
The main difference between dual data structures and quantifiable data structures is the allowable order in which the requests may be fulfilled.
The relaxed semantics of quantifiability provides an opportunity for performance gains over the dual data structures.

%Concern (2) do we want to minimize latency for individual operations, or are we willing to increase the latency of calls in hopes of reducing contention (e.g., by combining with a matching call in another thread)?
Other data structure designs observe that contention can be reduced by allowing operations to be matched and eliminated if the combined effect does not change the abstract state of the data structure.
The elimination backoff stack (EBS) \cite{Hendler2004-ho} uses an elimination array where $push$ and $pop$ method calls are matched to each other at random within a short time delay if the main stack is suffering from contention.  
%This concept of waiting to match up mutually satisfying calls provided inspiration for quantifiability.  
In the algorithm, the delay is set to a fraction of a second, which is a sufficient amount of time to find a match during busy times. If no match arrives, the method call retries its operation on the central stack object. When operating on the central stack object, the $pop$ method is at risk of failing if the stack is empty. However, if the elimination array delay time is set to infinite, the elimination backoff stack implements Quantifiability Principle~1, and all the method calls wait until they succeed.

The TS-Queue \cite{TSQueue2014} is one of the fastest queue implementations, claiming twice the speed of the elimination backoff version.  The TS Queue also relies on matching up method calls, enabling methods that would otherwise fail when reaching an undefined state of the queue to instead be fulfilled at a later time. In the TS Queue, rather than a global delay, there is a tunable parameter called \textit{padding} added to different method calls.  By setting an infinite time padding on all method calls, the TS Queue follows Quantifiability Principle~1. 

The EBS and TS-Queue share in common that they significantly improve performance by using a window of time in which pending method calls are conserved until they can succeed.  Quantifiability Principle~1 extends this window of time for conservation of method calls indefinitely, while allowing threads to cancel them as needed for specific applications.  

%Concern (3) Are we willing to relax object semantics in hopes of improving performance (e.g., by allowing a pop to return any object that is _near_the top of the stack)?
Contention due to frequently accessed elements in a data structure can be further reduced by relaxing object semantics.
The $k$-FIFO queue~\cite{kirsch2013fast} maintains $k$ segments each consisting of $k$ slots implemented as either an array for a bounded queue or a list for an unbounded queue. This design enables up to $k$ $enqueue$ and $dequeue$ operations to be performed in parallel and allows elements to be dequeued out-of-order up to a distance $k$. 
Quantifiability takes the relaxed semantics of the $k$-FIFO queue a step further by allowing method calls to occur out-of-order up to any arbitrary distance, leading to significant performance gains as demonstrated in Section~\ref{SubSection:Performance}. 

%These related works have demonstrated that applying the principles of quantifiability is the surest way to achieve high performance.  Accepting it as right correctness condition for concurrency will untangle our thinking from sequential histories.  It is time to recognize that the path to high performance is quantifiability. 

%Others have relaxed correctness conditions for performance, but their work is different because...

\iffalse
\begin{figure*}[t]
\centering
\newcommand{\chainlabel}[2]{\path [<-, draw] (#1) -> node [at end] {#2} ++(0,0.5);}

\begin{tikzpicture}[scale=1.25, transform shape, every node/.style={rectangle split, rectangle split parts=2, rectangle split}, node distance=1em, start chain, every join/.style={<-}]
 \node[draw, on chain, join] { $pop(\&a1)$ \nodepart{second} -2 };
 \node[draw, on chain, join, label={top @time2}] { $pop(\&a2)$ \nodepart{second} -1 };
 %\node[label={\small 1/16}] (a) at (0,0) {a};

 \node[draw, on chain, join, label={start}] { $empty$  \nodepart{second} 0 };
 \node[draw, on chain, join] { $push(5)$ \nodepart{second} 1  };
 \node[draw, on chain, join] { $push(9)$ \nodepart{second} 2  };
 \node[draw, on chain, join, label={top @time1}] { $push(7)$  \nodepart{second} 3 };
 %\node[draw, on chain, join] {};
\chainlabel{chain-6.north}{};
\chainlabel{chain-3.north}{};
\chainlabel{chain-2.north}{};

\end{tikzpicture}

\caption{Negative stack formed @time2 after 3 $push$ followed by 5 $pop$ calls.}
\label{fig:negativeStack}
\end{figure*}
\fi

\section{Implementation}
\label{Section:Implementation}

A quantifiable stack and quantifiable queue are implemented to showcase the design of quantifiable data structures.
Since the quantifiable stack and quantifiable queue have similar design strategies, the implementation details are only provided for the stack.
Quantifiability is applicable to other abstract data types that deliver additional functionality beyond the standard producer/consumer methods provided by queues and stacks.
Consider a reader method such as a $read$ operation for a hashmap or a $contains$ operation for a set.
If the item to be read does not exist in the data structure, a \textit{pending item} is created and placed in the data structure at the same location where the item to be read would be placed if it existed.
If a pending item already exists for the item to be read, the reader method references this pending item.
Once a producer method produces the item for which the pending item was created, the pending item is updated to a regular (non-pending) item. 
Since the reader methods hold a reference to this item, they may check the address when desired to determine if the item of interest is available to be read.
A similar strategy can be utilized for writer methods.

\subsection{QStack}
\label{SubSection:QStack}
The quantifiable stack (QStack) is designed to conserve method calls while avoiding contention wherever possible. Consider the state of a stack receiving two concurrent $push$ operations. Assume a stack contains only Node~$1$. Two threads concurrently push Node~$2(a)$ and Node~$2(b)$. The state of the stack after both operations have completed is shown in Figure~\ref{fig:stack1}. The order is one of two possibilities: $5,7,9$, or $5,9,7$. Based on this quantifiable implementation, either $7$ or $9$ are valid candidates for a $pop$ operation.

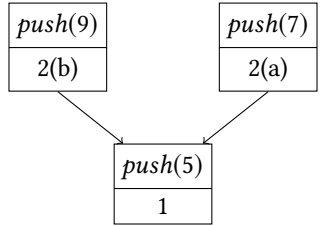
\begin{figure}
    \centering
    \begin{tikzpicture}[
    rectnode/.style={rectangle split, rectangle split parts=2, draw},
    ]
    \node[rectnode] at (0,0) (a) { $push(5)$ \nodepart{second} 1  };
    \node[rectnode] at (1.4,1.8) (b) { $push(7)$ \nodepart{second} 2(a)  };
    \node[rectnode] at (-1.4,1.8) (c) { $push(9)$ \nodepart{second} 2(b)  };
    
    \draw[->] (b.south) -- (a.45);
    \draw[->] (c.south) -- (a.135);
    \end{tikzpicture}
    
    \caption{Concurrent push representation of nodes "2(a)" and "2(b)"}
    \label{fig:stack1}
\end{figure}

The QStack is structured as a doubly-linked tree of nodes. Two concurrent $push$ method calls are both allowed to append their nodes to the data structure, forming a fork in the tree (Figure~\ref{fig:stack1}). $Push$ and $pop$ are allowed to insert or remove nodes at any ``leaf'' node. To facilitate this design we add a descriptor pointer to each node in the stack. At the start of each operation, a thread creates a descriptor object with all the details necessary for an arbitrary thread to carry out the intended operation.

\begin{algorithm}
\footnotesize
    \caption{Stack: Definitions}
    \label{alg:StackDefintion}
    \begin{multicols}{2}
    \begin{algorithmic}[1]
        \State \textbf{Struct} Node $\{$
        \State \textbf{T} value;
        \State \textbf{Op} op;
        \State \textbf{Node *} nexts[];
        \State \textbf{Node *} prev;
        \State  $\};$
        \State \textbf{Struct} Desc $\{$
        \State \textbf{T} value;
        \State \textbf{Op} op;
        \State \textbf{bool} active = true;
        \State $\};$
    \end{algorithmic}
    \end{multicols}
\end{algorithm}

Algorithm \ref{alg:StackDefintion} contains type definitions for the QStack. $Node$ contains the fields $value$, $op$, $nexts$ and $prev$. The $value$ field represents the abstract type being stored in the data structure. The $op$ field identifies the node as either a pushed value or an unsatisfied pop operation. The $nexts$ field is an array holding references to the children of the node, while $prev$ contains a reference to its parent. $Descriptor$ contains the $value$ and $op$ fields, as well as $active$. The $active$ field designates whether the associated operation for the descriptor object is currently pending, or if the thread performing that operation has completed it. The stack data structure has a global array $tail$, which contains all leaf nodes in the tree. 
The stack data structure also has a global variable $forkRequest$ that is used to indicate that another branch should be added to a node in the stack and is initialized to null.
The tree is initialized with a sentinel node in which the $active$ flag is set to false.

In order to conserve unsatisfied pops, we generalize the behaviour of $push$ and $pop$ operations with $insert$ and $remove$. If a $pop$ is made on an empty stack, we instead begin a stack of waiting $pop$ operations by calling $insert$ and designating the inserted node as an unfulfilled $pop$ operation. Similarly, if we call $push$ on a stack that contains unsatisfied pops, we instead use $remove$ to eliminate an unsatisfied $pop$ operation, which then finally returns the value provided by the incoming $push$.

\begin{algorithm}
\footnotesize
	\caption{Stack: Insert}
	\label{alg:StackInsert}
	\begin{algorithmic}[1]
		\Function{Insert} {$\textbf{Node *}\ cur,\ \textbf{Node *}\ elem,\ \textbf{int}\ index$}
		    \State \textbf{Desc*} $d$ = \textbf{new} Desc(v, op) \label{insert:desc}
		    \State \textbf{Node *} $curDesc = cur.desc$
		        
	        \If{$curDesc.active == true$} \label{insert:checkDescriptor}
	            \State \textbf{return} $false$
	        \EndIf
		        
	        \If{$cur.desc.CAS(currDesc, d)$} \label{insert:CAS}
	            \If{$top[index]\ != cur$}
    	                \State $d.active = false$
    	                \State \textbf{return} $false$
    	        \EndIf
    	            
	            \If{$cur.nexts.isEmpty()\ \&\ top.count(cur) == 1$} \label{insert:isLeaf}
		            \State $elem.prev = cur$
		            \State $cur.nexts.add(elem)$
		            \State $tail[index] = elem$
		            \State \textbf{Node *} $helperNode$ = $forkRequest$
		                \If{$helperNode\ != null\ \&\ forkRequest.CAS(helperNode, null)$}
		                \label{insert:forkRequest}
		                    \If{$helperNode.op == cur.op$}
        		                \State $helperNode.prev = cur$
        		                \State $cur.nexts.add(helperNode)$
        		                \State initialize a new tail pointer and set it equal to $helperNode$
        		            \EndIf
		                \EndIf
		            \State $d.active = false$ \label{insert:visibility}
		            \State \textbf{return} $true$ 
		        \Else
		            \State Remove dead branch \label{insert:removeBranch}
	            \EndIf
	        \EndIf
	        \State $d.active = false$
	        \State \textbf{return} $false$
		\EndFunction
	\end{algorithmic}
\end{algorithm}

\begin{algorithm}
\footnotesize
	\caption{Stack: Remove}
	\label{alg:StackRemove}
	\begin{algorithmic}[1]
		\Function{Remove} {$\textbf{Node *}\ cur,\ \textbf{int}\ index$}
		    \State \textbf{Desc*} $d$ = \textbf{new} Desc(op) \label{remove:desc}
		        \State \textbf{Node *} $curDesc = cur.desc$
		        \State \textbf{Node *} $prev = cur.prev$
		        
		        \If{$curDesc.active == true$} \label{remove:checkDescriptor}
		            \State \textbf{return} $false$
		        \EndIf
		        
		        \If{$cur.desc.CAS(currDesc, d)\ \&\ top.count(cur) == 1$} \label{remove:CAS}
    		        \If{$top[index]\ != cur$}
    	                \State $d.active = false$
    	                \State \textbf{return} $false$
    	            \EndIf
		            \If{$cur.nexts.isEmpty()$} \label{remove:isLeaf}
		                \State $v = cur.value$
		                \State $prev.nexts.remove(cur)$
		                \State $tail[index] = prev$
		                \State $d.active = false$ \label{remove:visibility}
		                \State \textbf{return} $true$ 
		            \Else
                        \State Remove dead branch
		            \EndIf
		        \EndIf
		        \State $d.active = false$
		        \State \textbf{return} $false$
		\EndFunction
	\end{algorithmic}
\end{algorithm}

Algorithm \ref{alg:StackInsert} details the pseudocode for the $insert$ operation. A node $cur$ is passed in, which is expected to be a leaf node. In addition, $elem$ is passed in, which is the node to be inserted. We check the descriptor of $cur$ to see if another thread is already performing an operation at this node on line \ref{insert:checkDescriptor}. If there is no pending operation, then we attempt to update the descriptor to point to our own descriptor on line \ref{insert:CAS}. If this is successful, we check on line \ref{insert:isLeaf} if $cur$ is a leaf node by ensuring $cur.nexts$ is empty and that $top$ contains only 1 reference to $cur$. If it is not, that means that $cur$ was previously a fork in the tree, but all nodes from one of the branches has been popped. In this case, we remove the index of the $tail$ array corresponding to the empty branch, effectively removing the fork at $cur$ from the tree. If $cur$ is determined to be a leaf node on line \ref{insert:isLeaf}, we are free to make modifications to $cur$ without interference from other threads. In this case, $elem$ is linked with $cur$ and the tail pointer is updated. 

The $remove$ method is given by Algorithm \ref{alg:StackRemove}. The $remove$ method is similar to the $insert$ method except that after the CAS on line \ref{remove:CAS}, we check if $cur$ is a leaf node before removing it from the tree.

$Push$ and $pop$ methods wrap these algorithms, as both operations need to be capable of inserting or removing a node depending on the state of the stack. Care should be taken that $push$ only removes a node when the stack contains unsatisfied $pop$ operations, while $pop$ should only insert a node when the stack is empty, or already contains unsatisfied $pop$ operations.  

Algorithm~\ref{alg:StackPush} details the $push$ method for the QStack. On line \ref{Push:newNode} we allocate a new node, and set the $value$ and $op$ field. Since a node may represent either a pushed value, or a waiting pop, we need to use $op$ to designate the operation of the node. At line \ref{Push:randomIndex}, we choose an index at which to try and add our node. The $tail$ array contains all leaf nodes. The $getRandomIndex()$ method avoids contention with other threads by choosing a random index.

If a thread is failing to make progress (line \ref{Push:loop}), we update the $forkRequest$ variable to contain the node for the delayed operation. When a successful $insert$ operation finds a non-null value in the $forkRequest$ variable on line \ref{insert:forkRequest}, it inserts that node as a sibling to its own node. This creates a fork at the node $cur$, increasing the chance of success for future $insert$ operations. 

%If the queue contains any unsatisfied $dequeue$ operations, the current $enqueue$ operation will be used to satisfy one. Otherwise, we will insert the node as normal. 
If the node's operation is determined to be a $pop$ on line~\ref{Push:check}, then the $push$ operation will fulfill the unsatisfied $pop$ operation.
Otherwise, the $push$ operation will proceed to insert its node into the stack.
The $pop$ method is given by Algorithm~\ref{alg:StackPop}. 
Similar to push, a random index is selected on line~\ref{Pop:randomindex} and the corresponding node is retrieved on line~\ref{Pop:readTail}.
If the node's operation is determined to be a $push$ on line~\ref{Pop:check} then the node is removed from the top of the stack. 
Otherwise, the stack is empty and the unsatisfied $pop$ operation is inserted in the stack.

\begin{algorithm}
\footnotesize
	\caption{Stack: Push}
	\label{alg:StackPush}
	\begin{algorithmic}[1]
		\Function{Push} {$\textbf{T}\ v$}
            \State \textbf{Node*} $elem$ = \textbf{new} Node(v, PUSH) \label{Push:newNode}
            \State \textbf{bool} $ret = false$
            \State \textbf{int} $loops = 0$
		    \While{$true$} \label{Push:while}
		        \State $loops++$
		        \State \textbf{int} $index = getRandomIndex()$ \label{Push:randomIndex}
		        \State \textbf{Node *} $cur = tail[index]$ \label{Push:readTail}
		        \If{$cur == null$} \label{Push:nullcheck}
		            \State \textbf{Continue}
		        \EndIf
		        \If{$cur.op == POP$} \label{Push:check}
		            \State $ret = remove(cur, v)$ \label{Push:remove}
		        \Else
		            \State $ret = insert(cur, elem, v)$ \label{Push:insert}
		        \EndIf
		        \If{ret}
		            \State \textbf{break}
		        \EndIf
		            
	            \If{$loops > FAIL\_THRESHOLD~\&~!forkRequest$} \label{Push:loop}
                    \State $forkRequest.CAS(null, cur)$ \label{Push:forkRequest}
                    \State \textbf{Break}
		        \EndIf
		    \EndWhile
		\EndFunction
	\end{algorithmic}
\end{algorithm}

\begin{algorithm}
\footnotesize
	\caption{Stack: Pop}
	\label{alg:StackPop}
	\begin{algorithmic}[1]
		\Function{Pop} {$\textbf{T}\ \&v$}
            \State \textbf{Node*} $elem$ = \textbf{new} Node(v, POP)
            \State \textbf{bool} $ret = false$
            \State \textbf{int} $loops = 0$
		    \While{$true$} \label{Pop:while}
		        \State $loops++$
		        \State \textbf{int} $index = getRandomIndex()$ \label{Pop:randomindex}
		        \State \textbf{Node *} $cur = tail[index]$ \label{Pop:readTail}
		        \If{$cur == null$} \label{Pop:nullcheck}
		            \State \textbf{Continue}
		        \EndIf
		        \If{$cur.op == PUSH$} \label{Pop:check}
		        %$sentinel \not\in tail$ 
		            \State $ret = remove(cur, \&v)$ \label{Pop:remove}
		        \Else
		            \State $ret = insert(cur, elem, \&v)$ \label{Pop:insert}
		        \EndIf
		        \If{ret}
		            \State $v = cur.value$
		            \State \textbf{break}
		        \EndIf
		        
		        \If{$loops > FAIL\_THRESHOLD~\&~!forkRequest$} \label{insert:loop}
                    \State $forkRequest.CAS(null, cur)$ \label{Pop:forkRequest}
                    \State \textbf{Break}
		        \EndIf
		    \EndWhile
		\EndFunction
	\end{algorithmic}
\end{algorithm}

\begin{theorem}
The QStack is quantifiable.
\end{theorem}
\begin{proof}
To prove that the QStack is quantifiable it must be shown that each of the methods preserve atomicity, isolation, and conservation.
A visibility point is established for each of the methods that demonstrates that each method preserves atomicity and isolation.
\\ \textbf{Insert:} The $insert$ method creates a new descriptor on line~\ref{insert:desc}, where the $active$ field is initialized to true. When the CAS succeeds on line~\ref{insert:CAS}, any other thread that reads the descriptor on line~\ref{insert:checkDescriptor} when calling $insert$ (or line~\ref{remove:checkDescriptor} of Algorithm~\ref{alg:StackRemove} when calling $remove$) will observe that the $active$ field is true and will continue from the beginning of the while loop on line~\ref{Push:while} of Algorithm~\ref{alg:StackPush} when calling $push$ (or line~\ref{Pop:while} of Algorithm~\ref{alg:StackPop} when calling $pop$). 
When the if statement on line~\ref{insert:isLeaf} succeeds, the current thread sets the descriptor's $active$ field to false on line~\ref{insert:visibility}. 
Since threads that were spinning due to the if statement on line~\ref{insert:checkDescriptor} (or line~\ref{remove:checkDescriptor} of Algorithm~\ref{alg:StackRemove} when calling $remove$) are now able to observe the effects of the operation associated with the previous active descriptor, the visibility point for the $insert$ method is line~\ref{insert:visibility}.
\\ \textbf{Remove:} The $remove$ method creates a new descriptor on line~\ref{remove:desc}, where the $active$ field is initialized to true. When the CAS succeeds on line~\ref{remove:CAS}, any other thread that reads the descriptor on line~\ref{remove:checkDescriptor} when calling $remove$ (or line~\ref{insert:checkDescriptor} of Algorithm~\ref{alg:StackInsert} when calling $insert$) will observe that the $active$ field is true and will continue from the beginning of the while loop on line~\ref{Push:while} of Algorithm~\ref{alg:StackPush} when calling $push$ (or line~\ref{Pop:while} of Algorithm~\ref{alg:StackPop} when calling $pop$). When the if statement on line~\ref{remove:isLeaf} succeeds, the current thread sets the descriptor's $active$ field to false on line~\ref{remove:visibility}. 
Since threads that were spinning due to the if statement on line~\ref{remove:checkDescriptor} (or line~\ref{insert:checkDescriptor} of Algorithm~\ref{alg:StackInsert} when calling $insert$) are now able to observe the effects of the operation associated with the previous active descriptor, the visibility point for the $remove$ method is line~\ref{remove:visibility}.
\\ \textbf{Push:} The $push$ method accesses the node at a random tail index on line~\ref{Push:readTail}. If the operation of the node is a $pop$, then $remove$ is called on line~\ref{Push:remove}, so the visibility point is line~\ref{remove:visibility} of Algorithm~\ref{alg:StackRemove}. Otherwise, $insert$ is called on line~\ref{Push:insert}, so the visibility point is line~\ref{insert:visibility} of Algorithm~\ref{alg:StackInsert}.
\\ \textbf{Pop:} The $pop$ method accesses the node at a random tail index on line~\ref{Pop:readTail}. If the operation of the node is a $push$, then $remove$ is called on line~\ref{Pop:remove}, so the visibility point is line~\ref{remove:visibility} of Algorithm~\ref{alg:StackRemove}. Otherwise, $insert$ is called on line~\ref{Pop:insert}, so the visibility point is line~\ref{insert:visibility} of Algorithm~\ref{alg:StackInsert}.

It now must be shown that the method calls are conserved. 
Since $insert$ and $remove$ are utility functions, only $push$ and $pop$ must be conserved.
\\ \textbf{Push:} The $push$ method checks if the operation of the node at the tail is a $pop$ on line~\ref{Push:check}. If the check succeeds, then the $push$ fulfills the unsatisfied $pop$ by removing it from the stack at line~\ref{Push:remove}. Otherwise, it proceeds with its own operation by calling $insert$ at line~\ref{Push:insert}. Since a $pop$ request is guaranteed to be fulfilled if one exists due to the check on line~\ref{Push:nullcheck}, and $forkRequest$ is updated on line~\ref{Push:forkRequest} to the current node if the loop iterations exceeds the $FAIL\_THRESHOLD$, $push$ satisfies method call conservation.
\\ \textbf{Pop:} The $pop$ method checks if the operation of the node at the tail is a $push$ on line~\ref{Pop:check}. If the check succeeds, then the $pop$ proceeds with its own operation by removing it from the stack at line~\ref{Pop:remove}. Otherwise, it places its unfulfilled request by calling $insert$ at line~\ref{Pop:insert}. Since a $pop$ will only place a request if no nodes associated with a $push$ operation exist in the stack due to the check on line~\ref{Pop:nullcheck}, and $forkRequest$ is updated on line~\ref{Pop:forkRequest} to the current node if the loop iterations exceeds the $FAIL\_THRESHOLD$, $pop$ satisfies method call conservation.
\end{proof}

\subsection{Performance}
\label{SubSection:Performance}

The QStack and QQueue were tested against the fastest available published work, along with classic examples.  Stack results are shown in Figure \ref{fig:QStack}, and queue results in Figure \ref{fig:QQueue}.  The x-axis plots the number of threads available for each run.  The y-axis plots method calls per microsecond.  Plot line color and type show the different implementations.

\begin{figure*}[h]
 %\hfill
 \begin{subfigure}[t]{0.47\textwidth}
 \includegraphics[width=1.05\textwidth]{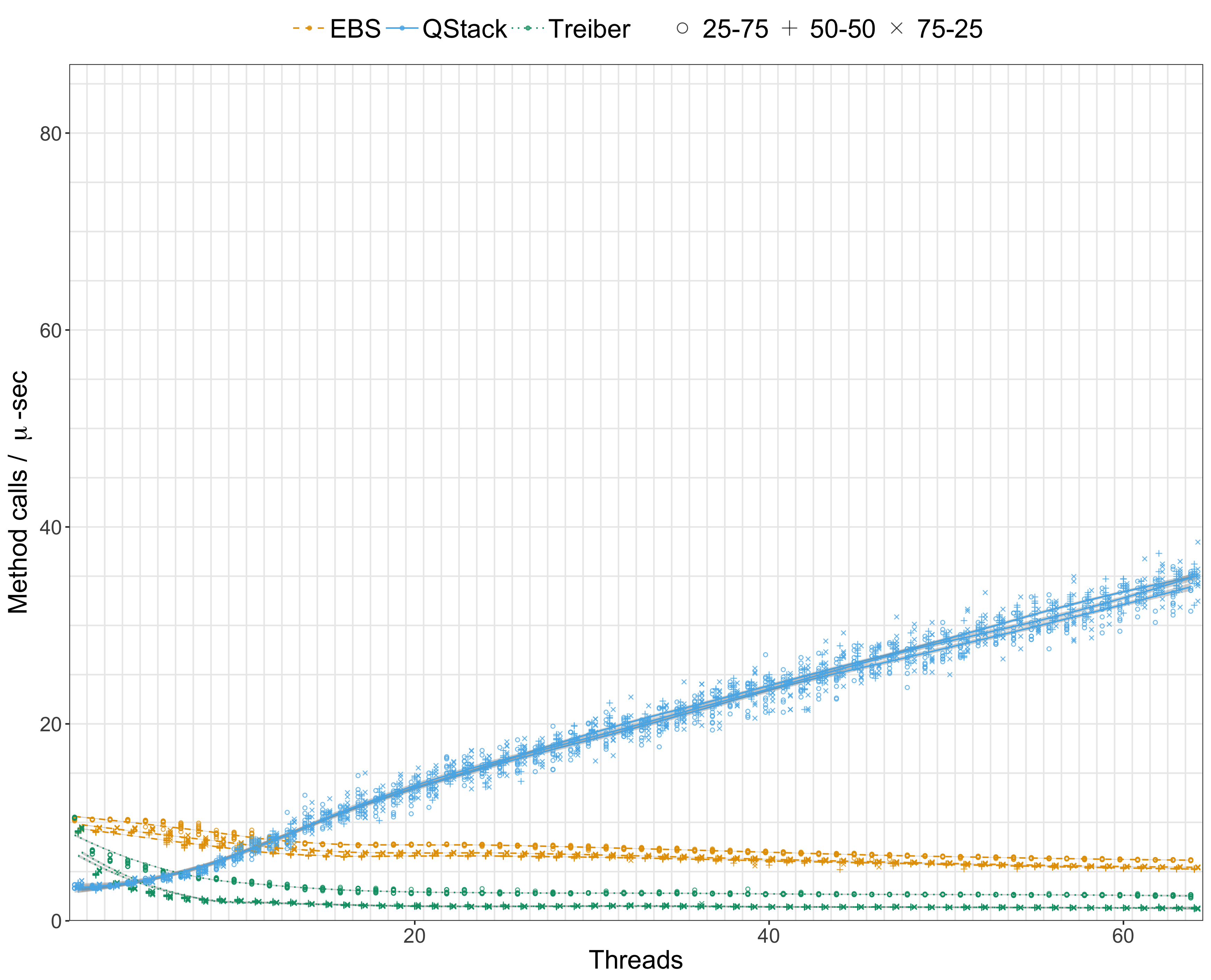}
 \caption{QStack, EBS and Treiber stack.} \label{fig:QStack}
 %\hfill
 \end{subfigure}
 %\hfill
 \hspace{1em}
 \begin{subfigure}[t]{0.47\textwidth}
 \includegraphics[width=1.05\textwidth]{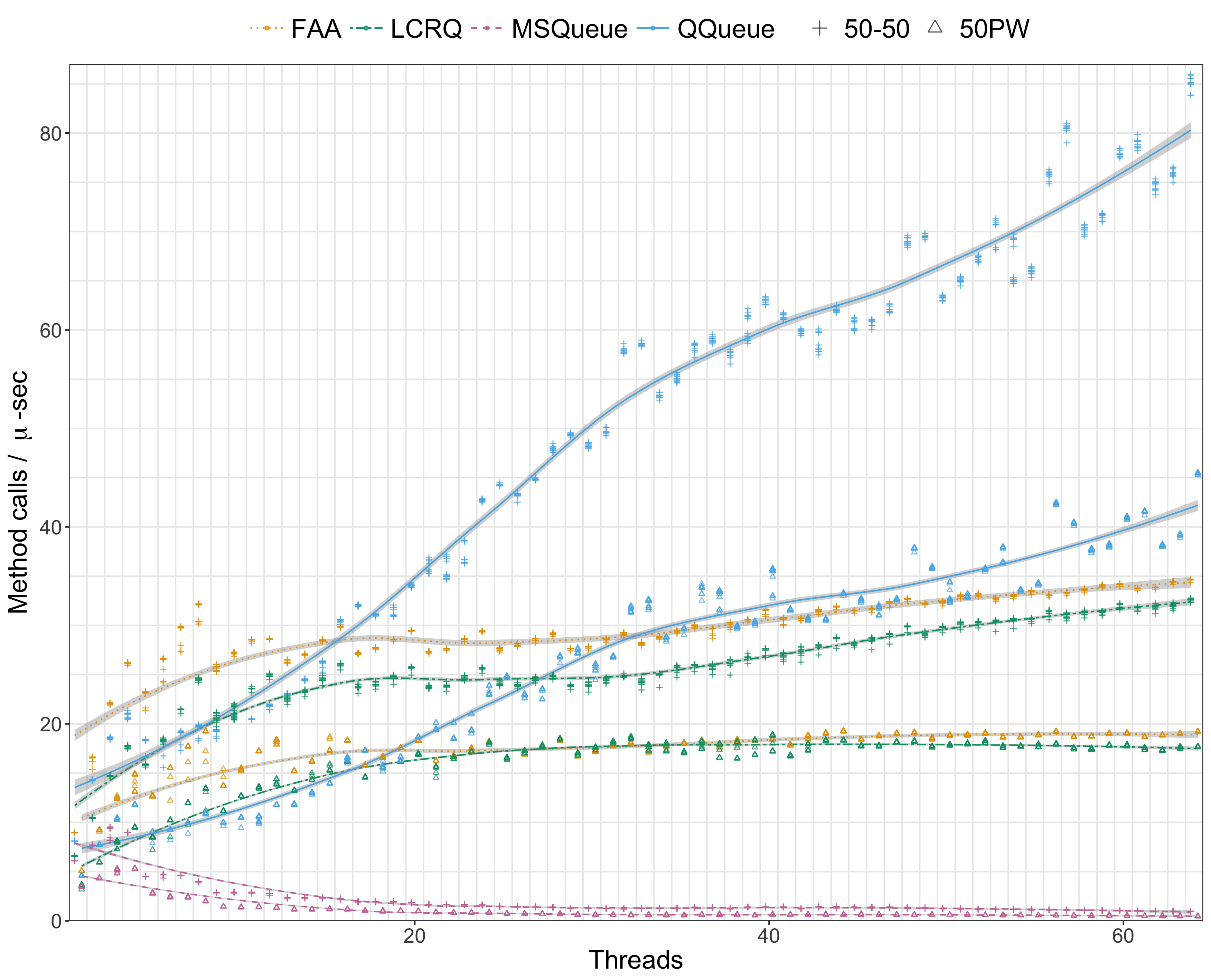}
 \caption{QQueue, FAA queue, LCRQ and MS queue.} \label{fig:QQueue} 
 %\hfill
 \end{subfigure}
\caption{Performance Analysis on the AMD\textsuperscript{\textregistered}  EPYC\textsuperscript{\textregistered} server} 
\end{figure*}

Experiments were run on an AMD\textsuperscript{\textregistered} EPYC\textsuperscript{\textregistered} server of 2GHz clock speed and 128GB memory, with 32 cores delivering a maximum of 64 simultaneous multi-threads.  The operating system is Ubuntu 18.04 LTS and code is compiled with gcc 7.3.0 using -O3 optimizations.

The QStack was compared with the lock-free Elimination Backoff Stack (EBS)  \cite{Bar-Nissan2011-cb} and lock-free Treiber Stack \cite{Treiber1986-cr}.   These were selected as representative of a relaxed semantics stack and a classic linearizable stack.   
With a single thread, Treiber and EBS demonstrate similar performance, while QStack is lower due to the overhead of descriptors, which incur more remote memory accesses. The operations mix made little difference, but the Treiber Stack and EBS showed slightly higher performance at 25-75 because they quickly discard unsatisfied pop calls, returning $null$.  As threads are added, Treiber drops off quickly due to contention.  At five threads QStack overtakes Treiber and at 12 threads becomes faster than EBS.  The salient result is that the QStack continues to scale, achieving over five times EBS performance with 64 threads.  The other implementations consume resources to maintain order at microsecond scale instead of serving requests as quickly as possible with best efforts ordering.

Testing methodology follows those used in the original EBS presentation, going from one to 64 threads with five million operations per thread.  Memory is pre-allocated in the stack experiments, and for each run the program is restarted by a script to prevent the previous memory state from influencing the next run. The Boost library \cite{Beman_Dawes2018-bs} is used to create a uniform random distribution of method calls based on the different mixes.  

Stack $push$-$pop$ mixes of 25-75, 50-50 and 75-25 were tested for each implementation across all threads.  Queue $enqueue$-$dequeue$ mixes were temporal variations on a 50-50 mix.   For both stack and queue, there were a minimum of 10 trials per thread per mix.  The data was smoothed using the LOESS method as implemented in the ggplot2 library. Shaded areas indicate the 95 percent confidence limits for the lines.  Additionally, the data points for every run are shown in both stack and queue plots, with slight x-offsets to the left and right inside the column for readability. 

The QQueue was compared with the lock-free LCRQ \cite{Morrison2013-mc}, the wait-free FAA queue \cite{Yang2016-io} and the lock-free MS queue \cite{Michael1995-vf}.  The LCRQ and FAA are the fastest queues in a recent benchmark framework with ACM verified code artifacts \cite{Yang2018-lw}.  The MS queue is a classic like the Treiber stack.  The framework uses only 50-50 mixes, one random (50-50) and one pairwise (50PW).  The QQueue performs similarly to LCRQ until overtaking it at 14 threads, then overtaking the wait-free FAA queue at 18 threads.  The FAA queue is exceptional as it performs as well or better than the alternative lock-free implementations. The TS-Queue \cite{TSQueue2014} and the Multiqueue \cite{Rihani2015-cf} are queues of interest published more recently than the FAA queue, but were not selected because verified code artifacts have not been published. 
%The results reported for the Multiqueue using highly relaxed semantics are of similar trend and magnitude as the QQueue, albeit on slightly faster hardware, so future work invites a comparison on both speed and error rates.  

Queue experiments follow the methodology of the Yang and Mellor-Crummey framework \cite{Yang2018-lw} and use the queue implementations provided in the source code.  Memory allocation is dynamic within the framework. Benchmarks provided are two variations on 50-50 mixes, one random and the other pairwise. 
The different temporal distributions within the 50-50 mix have more influence on the results than different mixes (25-75, 50-50, 75-25) used in the stack experiments.  

In both the pairwise and random mixes, the QQueue continues to scale to the limit of hardware support, more than double the performance of FAA and LCRQ at 64 threads.

The quantifiable containers continue to scale until all threads are employed, with slightly reduced slope in the simultaneous multi-threading region from 32 to 64 threads.  Other implementations, including those that are linearizable with relaxed semantics, could maintain microscale order only at the cost of scalability.  Furthermore, linearizability may cause unfairness where method calls are not conserved.  It follows that quantifiability allows for the design of fast and highly scalable multiprocessor data structures for modern multi-core computing platforms.

%The performance gains of quantifiability trump the small remaining ability of linearizability to enforce order at this microscopic scale.   
%so the elimination backoff stack (EBS) is faster. Also with few processes EBS has the advantage of a linearizable ordering. 
%We do not propose that quantifiability replace other correctness conditions, rather that there are use cases where one or the other is appropriate. 
%The error rate of the two was not investigated in this paper and is good material for future work.   On the one hand it is clear that the relaxed semantics of quantifiability will not preserve the order...
%It would also be possible to add a time category to the dimensions, creating vector spaces with a M-P-O-I-T basis.   This would allow the techniques discussed in this paper to be applied to linearizability.
%\subsection{Temporal Order of Asynchronous Method Calls}
\subsection{Entropy Applied to Concurrent Objects}

With a large number of threads and therefore many overlapping method calls, quantifiability and linearizability both accept any ordering of the overlapping calls.  The results presented in Section~\ref{SubSection:Performance} show performance gains for prototype quantifiable data structures over their linearizable counterparts.  What is missing is a way to measure the disorder introduced to achieve such gains.  This section will propose a measure of the disorder and apply it to the experimental results shown in Figure \ref{fig:QStack}.

\begin{figure*}[h]
 %\hfill
 \begin{subfigure}[b]{0.47\textwidth}
 % trim=left botm right top
 \includegraphics[clip, trim=0.2cm 5.0cm 0.2cm 5.0cm, width=1.05\textwidth]{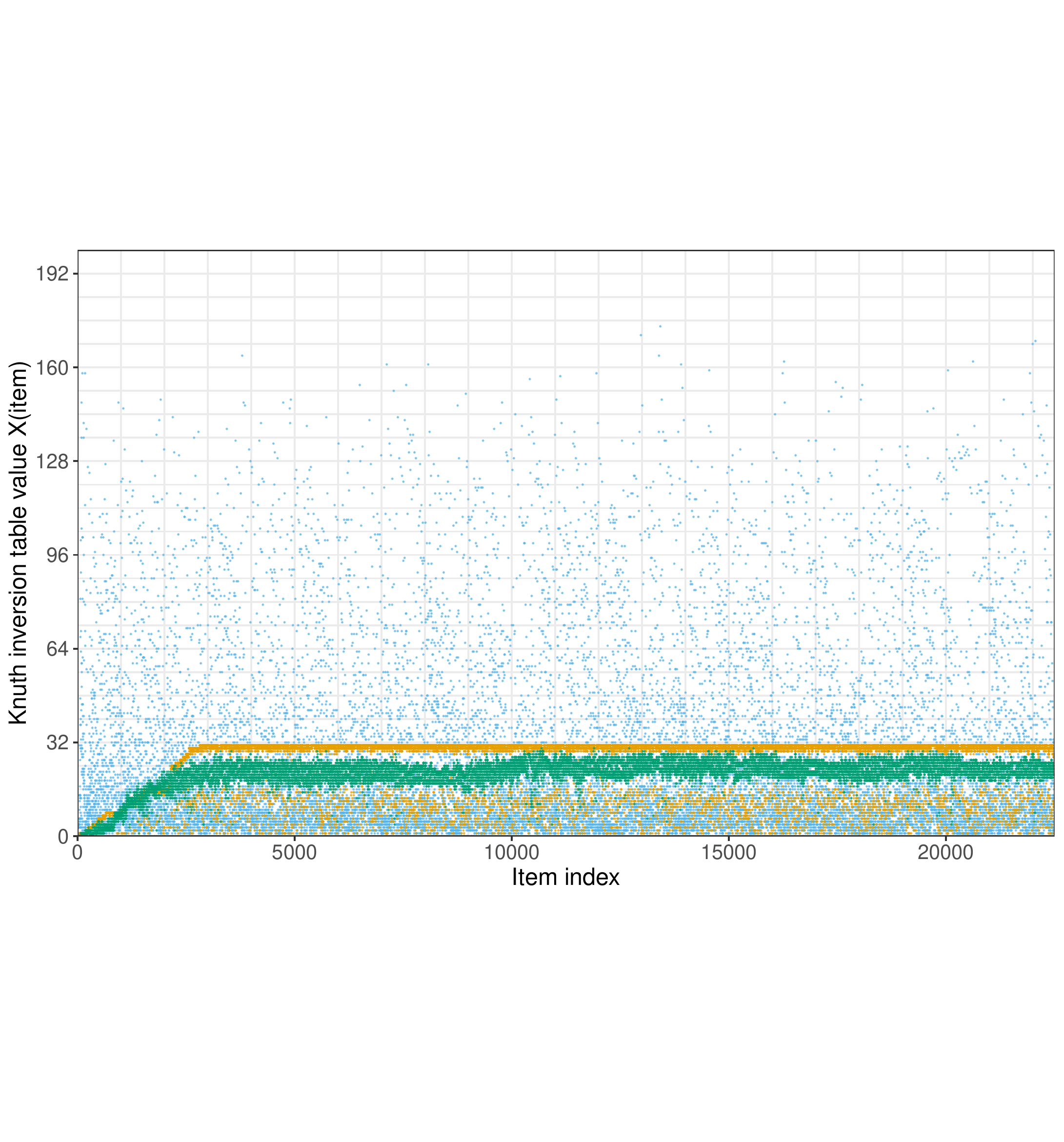}
 \caption{Knuth inversions X(item) 32 threads} 
 \label{fig:Events32}
 %\hfill
 \end{subfigure}
 %\hfill
 \hspace{1em}
 \begin{subfigure}[b]{0.47\textwidth}
 \includegraphics[clip, trim=0.2cm 5.0cm 0.2cm 5.0cm, width=1.05\textwidth]{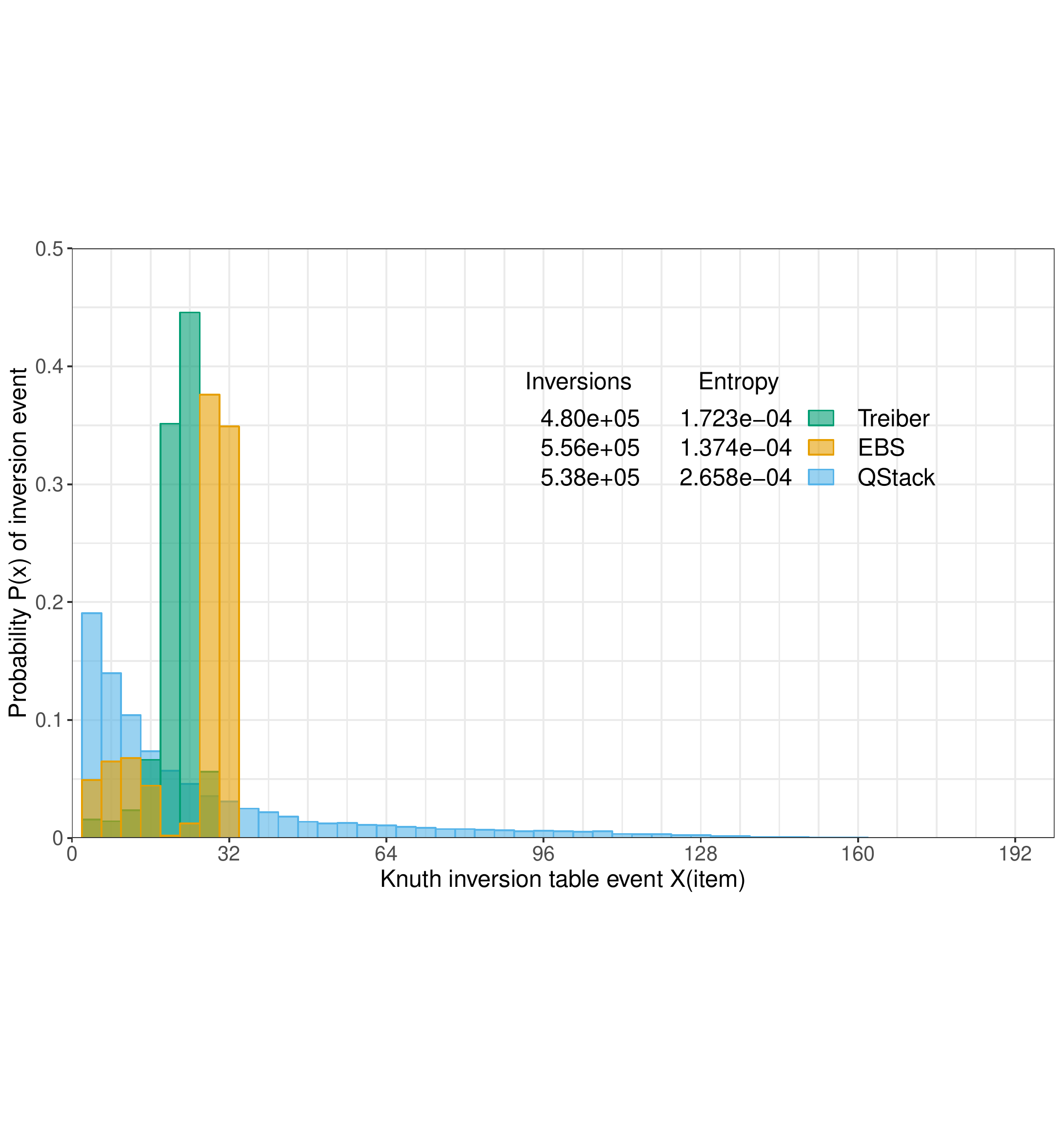}
 \caption{Distribution P(X) of inversions} 
 \label{fig:PMD32} 
 \end{subfigure}
 
 \begin{subfigure}[b]{0.47\textwidth}
 \includegraphics[clip, trim=0.2cm 5.0cm 0.2cm 4.0cm, width=1.05\textwidth]{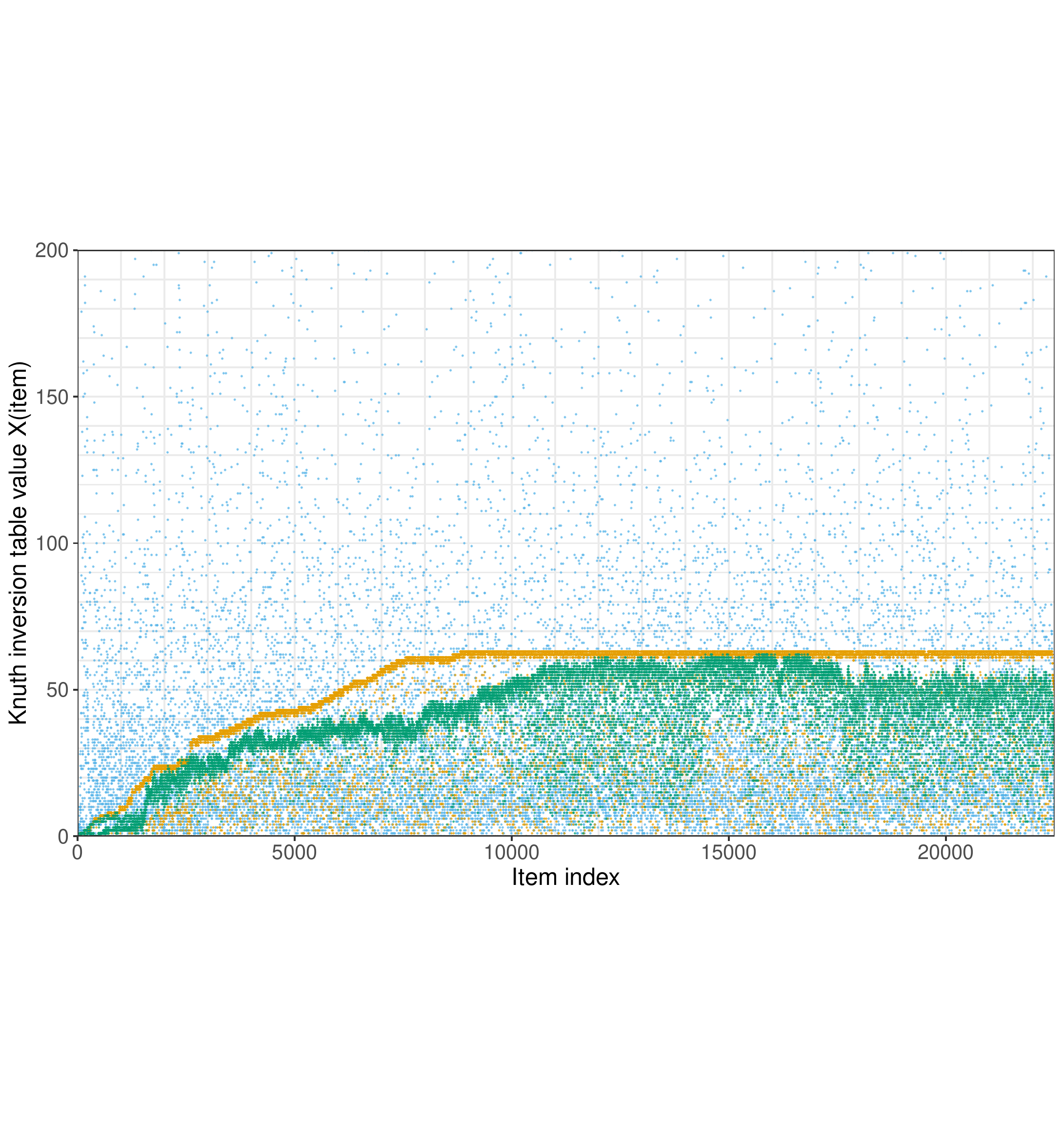}
 \caption{Knuth inversions X(item) 64 threads} 
 \label{fig:Events64}
 %\hfill
 \end{subfigure}
 %\hfill
 \hspace{1em}
 \begin{subfigure}[b]{0.47\textwidth}
 \includegraphics[clip, trim=0.2cm 5.0cm 0.2cm 4.0cm, width=1.05\textwidth]{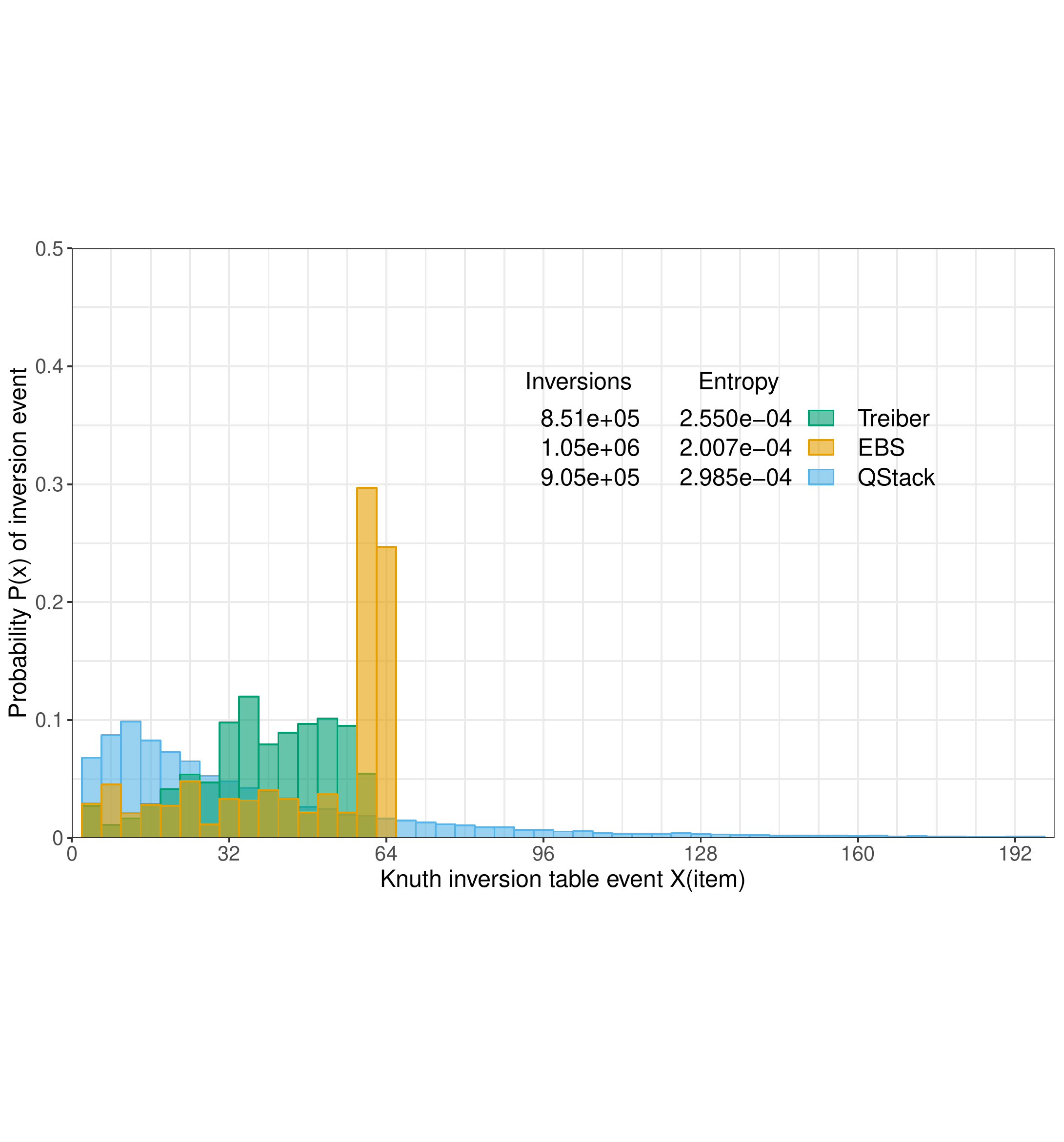}
 \caption{Distribution P(X) of inversions} 
 \label{fig:PMD64} 
 \end{subfigure}
\caption{Inversion events X(item) and distribution P(X) for concurrent stacks, 32 and 64 threads.} 

\end{figure*}

Entropy applied to information systems is called \textit{Shannon entropy}, also referred to as \textit{information entropy}~\cite{shannon1948mathematical}.  
%It is the rate at which information is produced along a channel.   
%Shannon entropy is a probability distribution over a vector of possible outcomes.  
Shannon entropy measures the amount of uncertainty in a probability distribution~\cite{Goodfellow-et-al-2016}. 
%Information entropy is a measure of the predictability of items in a sequence of data. 
With a single process calling a non-concurrent object, the results are completely predictable. Concurrent data structures, even provably linearizable ones, may admit unpredictable results.  A perfectly ordered data stream input sequentially to a linearizable FIFO queue will not necessarily emerge in the same order.  Overlapping method calls together with the rules of linearizability may allow a great many different correct output orders. In the literature this divergence from the real time order is called an \textit{error rate}.   This term is misleading as it implies failure or incorrectness.  Where the unpredictability of a concurrent system is not an error, but lies within the bounds of correctness, we suggest \textit{entropy} as the measure.  In many physical systems, the change in entropy increases with the speed of a reaction.  Computing systems cannot escape the general applications of physical laws. The motivation behind generalizing entropy for concurrent objects is to provide the ability to measure the uncertainty in a concurrent system for the comparison of concurrent correctness conditions. 

If the natural numbers $1,2 ... N$ are sent into a FIFO queue and emerge intact, each successive dequeue method call output will always be greater than the preceding dequeue method call output.  This scenario represents perfect predictability, and zero entropy.  Likewise, a LIFO stack output would perfectly reverse the order.
%The information reveals how out-of-order is the output of a history for a concurrent data structure.  
%We will measure the number of times an interesting event occurs.  
The interesting events are when items come back from our data structures in the ``wrong'' order.   These are called \textit{surprises} in the literature, and their probability distribution of occurrence is called the \textit{surprisal}.
For an ordered list $a_1, a_2, ... a_i$, an inversion \cite{Knuth1998-td} in the list is a surprise, and the inversion count $x(j)$ is defined for each list element $a_j$ as follows:
\begin{equation}
x(j) = count( i<j,  a_i > a_j) + count(i>j, a_i < a_j)
\end{equation}
%To calculate a discrete probability distribution $P(x)$ of surprising events requires that events be defined along the x axis according to their degree of surprisal. 
\iffalse
Let $X$ be a discrete random variable that denotes an observed inversion count for a list element, referred to as an item in the system model.
Let $N$ be the total number of items in the list.
Let $k_X$ be the number of items that measured an inversion count of $X$.
The discrete probability distribution of surprising events is $P(X)=\frac{k_X}{N}$.
An item with a large inversion count is more out-of-order and thus more surprising than an item with a smaller count. 
\fi
Given inversion count $x(j)$ defined over the concurrent history with $n$ items, let $k$ be the number of items with an inversion count of $X$. The possible values of $X$ are \{ $X_0, ..., X_n$\}.  The discrete probability distribution of surprising events is $P(X)=\frac{k}{n}$.
The entropy for a set of concurrent histories generated in an experiment is: %provided in Equation~\ref{eq:entropy}.

\begin{equation}
\label{eq:entropy}
H(X) = - \sum_{i=0}^{n} P(X_i) \textnormal{log} P(X_i)   
\end{equation}

%Also we need to make it clear that the results we showed for entropy were on a modified version of the QStack that did trade performance for lower entropy.
%"Entropy can be applied to measure the tradeoff between relaxed semantics and performance." but did we do this? What is the result?
To gather inversion data for the Treiber Stack, EBS and QStack , the real time order of items pushed to the stack was measured using instrumented code as done in \cite{Dodds2014-gi}.  The push order was compared to the pop order and inversion events were counted.  Raw data points are shown in Figures \ref{fig:Events32}, \ref{fig:Events64} and the discrete probability distributions obtained are shown in Figures \ref{fig:PMD32}, \ref{fig:PMD64}. %The raw data and probability distributions use the same colors as the performance results, blue for the QStack, orange for the EBS, and green for the Treiber.  
%From the raw data on the left, distributions are calculated on the right.  
%And what is the significance of: "It is notable that the Treiber stack and the EBS both show a hard limit on the maximum inversion, being close to the number of threads in each case. " ... Can this be interpreted that EBS and Treiber have the potential to scale up to the number of threads? Isn't this sufficient?
It is notable that the Treiber stack and the EBS both show a hard limit on the maximum inversion, being close to the number of threads in each case.  This is intuitive for linearizable or near-linearizable data structures because this is the maximum number of overlapping method calls. 

%Also, in our results, what is the importance that QStack entropy is double? Does this translate to the potential for better performance? Is there a relation between entropy and the complexity of correctness verification? 
At 32 threads, there is some dispersion in the results for Treiber and EBS, but the entropy caused by the QStack is double since there is more uncertainty in method call ordering for the QStack in comparison to the Treiber stack and EBS. 
At 64 threads, the dispersion is much greater for the EBS and Treiber, and only slightly more for the QStack.   The entropy increase reflects this trend since a larger variance in the probability distribution yields higher entropy. Convergence of results at higher thread counts for the concurrent data structures tested is reflected in the similar entropy $P(X)$. 
The performance results presented in Figure~\ref{fig:QStack} showcase how the QStack design leverages the uncertainty in a concurrent system
%as measured by entropy, 
to deliver high scalability obtained through relaxed semantics.

\section{Conclusion}
Quantifiability is a new concurrent correctness condition compatible with drivers of scalability: architecture, semantics and complexity.  Quantifiability is compositional without dependence upon timing or data structure semantics and is free of inherent locking or waiting.  The convenient expression of quantifiability in a linear algebra model offers the promise of reduced verification time complexity and powerful abstractions to facilitate concurrent programming innovation. 
The relaxed semantics permitted by quantifiability allow for significant performance gains through contention avoidance in the implementation of concurrent data structures. Entropy can be applied to evaluate the tradeoff between relaxed semantics and performance.

%Quantifiability models a world where a vast number of things happen at once so imposing microscopic order on them is futile.  But it is also a world where every request is conserved and receives a meaningful response.  We think it is the world of the future.

%\begin{acks} 
%  Thanks to the PODC 2019 committee for rejecting this paper so we could put it on ARXIV immediately!
%\end{acks}

%% Acknowledgments
%\begin{acks}                            %% acks environment is optional
                                        %% contents suppressed with 'anonymous'
  %% Commands \grantsponsor{<sponsorID>}{<name>}{<url>} and
  %% \grantnum[<url>]{<sponsorID>}{<number>} should be used to
  %% acknowledge financial support and will be used by metadata
  %% extraction tools.
%  This material is based upon work supported by the
%  \grantsponsor{GS100000001}{National Science
%    Foundation}{http://dx.doi.org/10.13039/100000001} under Grant
%  No.~\grantnum{GS100000001}{nnnnnnn} and Grant
%  No.~\grantnum{GS100000001}{mmmmmmm}.  Any opinions, findings, and
%  conclusions or recommendations expressed in this material are those
%  of the author and do not necessarily reflect the views of the
%  National Science Foundation.
%\end{acks}

%\balance 
%% Bibliography
\bibliography{ARXIV-Quantifiability.bib}

%%% -*-BibTeX-*-
%%% Do NOT edit. File created by BibTeX with style
%%% ACM-Reference-Format-Journals [18-Jan-2012].

\begin{thebibliography}{00}

%%% ====================================================================
%%% NOTE TO THE USER: you can override these defaults by providing
%%% customized versions of any of these macros before the \bibliography
%%% command.  Each of them MUST provide its own final punctuation,
%%% except for \shownote{}, \showDOI{}, and \showURL{}.  The latter two
%%% do not use final punctuation, in order to avoid confusing it with
%%% the Web address.
%%%
%%% To suppress output of a particular field, define its macro to expand
%%% to an empty string, or better, \unskip, like this:
%%%
%%% \newcommand{\showDOI}[1]{\unskip}   % LaTeX syntax
%%%
%%% \def \showDOI #1{\unskip}           % plain TeX syntax
%%%
%%% ====================================================================

\ifx \showCODEN    \undefined \def \showCODEN     #1{\unskip}     \fi
\ifx \showDOI      \undefined \def \showDOI       #1{#1}\fi
\ifx \showISBNx    \undefined \def \showISBNx     #1{\unskip}     \fi
\ifx \showISBNxiii \undefined \def \showISBNxiii  #1{\unskip}     \fi
\ifx \showISSN     \undefined \def \showISSN      #1{\unskip}     \fi
\ifx \showLCCN     \undefined \def \showLCCN      #1{\unskip}     \fi
\ifx \shownote     \undefined \def \shownote      #1{#1}          \fi
\ifx \showarticletitle \undefined \def \showarticletitle #1{#1}   \fi
\ifx \showURL      \undefined \def \showURL       {\relax}        \fi
% The following commands are used for tagged output and should be
% invisible to TeX
\providecommand\bibfield[2]{#2}
\providecommand\bibinfo[2]{#2}
\providecommand\natexlab[1]{#1}
\providecommand\showeprint[2][]{arXiv:#2}

\bibitem[\protect\citeauthoryear{Abelson, Sussman, and Sussman}{Abelson
  et~al\mbox{.}}{[n. d.]}]%
        {Abelson_undated-vk}
\bibfield{author}{\bibinfo{person}{Harold Abelson}, \bibinfo{person}{Gerald~Jay
  Sussman}, {and} \bibinfo{person}{Julie Sussman}.} \bibinfo{year}{[n.
  d.]}\natexlab{}.
\newblock \bibinfo{booktitle}{{\em Structure and Interpretation of Computer
  Programs - 2nd Edition}}.
\newblock \bibinfo{publisher}{Justin Kelly}.
\newblock


\bibitem[\protect\citeauthoryear{Adhikari, Street, Wang, Liu, and
  Zhang}{Adhikari et~al\mbox{.}}{2013}]%
        {Adhikari2013-iu}
\bibfield{author}{\bibinfo{person}{Kiran Adhikari}, \bibinfo{person}{James
  Street}, \bibinfo{person}{Chao Wang}, \bibinfo{person}{Yang Liu}, {and}
  \bibinfo{person}{Shaojie Zhang}.} \bibinfo{year}{2013}\natexlab{}.
\newblock \showarticletitle{Verifying a Quantitative Relaxation of
  Linearizability via Refinement}.
\newblock In \bibinfo{booktitle}{{\em Model Checking Software}},
  \bibfield{editor}{\bibinfo{person}{Ezio Bartocci} {and} \bibinfo{person}{C~R
  Ramakrishnan}} (Eds.). \bibinfo{series}{Lecture Notes in Computer Science},
  Vol.~\bibinfo{volume}{7976}. \bibinfo{publisher}{Springer Berlin Heidelberg},
  \bibinfo{address}{Berlin, Heidelberg}, \bibinfo{pages}{24--42}.
\newblock


\bibitem[\protect\citeauthoryear{Afek, Korland, and Yanovsky}{Afek
  et~al\mbox{.}}{2010}]%
        {Afek2010-xv}
\bibfield{author}{\bibinfo{person}{Yehuda Afek}, \bibinfo{person}{Guy Korland},
  {and} \bibinfo{person}{Eitan Yanovsky}.} \bibinfo{year}{2010}\natexlab{}.
\newblock \showarticletitle{{Quasi-Linearizability}: Relaxed Consistency for
  Improved Concurrency}. In \bibinfo{booktitle}{{\em Principles of Distributed
  Systems}}. \bibinfo{publisher}{Springer Berlin Heidelberg},
  \bibinfo{pages}{395--410}.
\newblock


\bibitem[\protect\citeauthoryear{Alistarh, Brown, Kopinsky, Li, and
  Nadiradze}{Alistarh et~al\mbox{.}}{2018}]%
        {Alistarh2018-tw}
\bibfield{author}{\bibinfo{person}{Dan Alistarh}, \bibinfo{person}{Trevor
  Brown}, \bibinfo{person}{Justin Kopinsky}, \bibinfo{person}{Jerry~Z Li},
  {and} \bibinfo{person}{Giorgi Nadiradze}.} \bibinfo{year}{2018}\natexlab{}.
\newblock \showarticletitle{Distributionally Linearizable Data Structures}. In
  \bibinfo{booktitle}{{\em Proceedings of the 30th on Symposium on Parallelism
  in Algorithms and Architectures}}. \bibinfo{publisher}{ACM},
  \bibinfo{pages}{133--142}.
\newblock


\bibitem[\protect\citeauthoryear{Alur, McMillan, and Peled}{Alur
  et~al\mbox{.}}{1996}]%
        {Alur1996-mj}
\bibfield{author}{\bibinfo{person}{R Alur}, \bibinfo{person}{K McMillan}, {and}
  \bibinfo{person}{D Peled}.} \bibinfo{year}{1996}\natexlab{}.
\newblock \showarticletitle{Model-checking of correctness conditions for
  concurrent objects}. In \bibinfo{booktitle}{{\em Proceedings 11th Annual
  {IEEE} Symposium on Logic in Computer Science}}. \bibinfo{pages}{219--228}.
\newblock


\bibitem[\protect\citeauthoryear{Amit, Rinetzky, Reps, Sagiv, and Yahav}{Amit
  et~al\mbox{.}}{2007}]%
        {Amit2007-ti}
\bibfield{author}{\bibinfo{person}{Daphna Amit}, \bibinfo{person}{Noam
  Rinetzky}, \bibinfo{person}{Thomas Reps}, \bibinfo{person}{Mooly Sagiv},
  {and} \bibinfo{person}{Eran Yahav}.} \bibinfo{year}{2007}\natexlab{}.
\newblock \showarticletitle{Comparison Under Abstraction for Verifying
  Linearizability}.
\newblock In \bibinfo{booktitle}{{\em Computer Aided Verification}},
  \bibfield{editor}{\bibinfo{person}{Werner Damm} {and} \bibinfo{person}{Holger
  Hermanns}} (Eds.). \bibinfo{series}{Lecture Notes in Computer Science},
  Vol.~\bibinfo{volume}{4590}. \bibinfo{publisher}{Springer Berlin Heidelberg},
  \bibinfo{address}{Berlin, Heidelberg}, \bibinfo{pages}{477--490}.
\newblock


\bibitem[\protect\citeauthoryear{Aspnes, Herlihy, and Shavit}{Aspnes
  et~al\mbox{.}}{1994}]%
        {aspnes1994counting}
\bibfield{author}{\bibinfo{person}{James Aspnes}, \bibinfo{person}{Maurice
  Herlihy}, {and} \bibinfo{person}{Nir Shavit}.}
  \bibinfo{year}{1994}\natexlab{}.
\newblock \showarticletitle{Counting networks}.
\newblock \bibinfo{journal}{{\em Journal of the ACM (JACM)\/}}
  \bibinfo{volume}{41}, \bibinfo{number}{5} (\bibinfo{year}{1994}),
  \bibinfo{pages}{1020--1048}.
\newblock


\bibitem[\protect\citeauthoryear{Badrinath and Ramamritham}{Badrinath and
  Ramamritham}{1987}]%
        {Badrinath1987-oe}
\bibfield{author}{\bibinfo{person}{B~R Badrinath} {and} \bibinfo{person}{K
  Ramamritham}.} \bibinfo{year}{1987}\natexlab{}.
\newblock \showarticletitle{Semantics-based concurrency control: Beyond
  commutativity}. In \bibinfo{booktitle}{{\em 1987 {IEEE} Third International
  Conference on Data Engineering}}. \bibinfo{pages}{304--311}.
\newblock


\bibitem[\protect\citeauthoryear{Baier and Katoen}{Baier and Katoen}{2008}]%
        {Baier2008-wy}
\bibfield{author}{\bibinfo{person}{Christel Baier} {and}
  \bibinfo{person}{Joost-Pieter Katoen}.} \bibinfo{year}{2008}\natexlab{}.
\newblock \bibinfo{booktitle}{{\em Principles of Model Checking}}.
\newblock \bibinfo{publisher}{MIT Press}.
\newblock


\bibitem[\protect\citeauthoryear{Bar-Nissan, Hendler, and Suissa}{Bar-Nissan
  et~al\mbox{.}}{2011}]%
        {Bar-Nissan2011-cb}
\bibfield{author}{\bibinfo{person}{Gal Bar-Nissan}, \bibinfo{person}{Danny
  Hendler}, {and} \bibinfo{person}{Adi Suissa}.}
  \bibinfo{year}{2011}\natexlab{}.
\newblock \showarticletitle{A Dynamic {Elimination-Combining} Stack Algorithm}.
\newblock In \bibinfo{booktitle}{{\em Principles of Distributed Systems}},
  \bibfield{editor}{\bibinfo{person}{Antonio Fern{\`a}ndez~Anta},
  \bibinfo{person}{Giuseppe Lipari}, {and} \bibinfo{person}{Matthieu Roy}}
  (Eds.). \bibinfo{series}{Lecture Notes in Computer Science},
  Vol.~\bibinfo{volume}{7109}. \bibinfo{publisher}{Springer Berlin Heidelberg},
  \bibinfo{address}{Berlin, Heidelberg}, \bibinfo{pages}{544--561}.
\newblock


\bibitem[\protect\citeauthoryear{B{\"a}umler, Schellhorn, Tofan, and
  Reif}{B{\"a}umler et~al\mbox{.}}{2011}]%
        {Baumler2011-jc}
\bibfield{author}{\bibinfo{person}{Simon B{\"a}umler}, \bibinfo{person}{Gerhard
  Schellhorn}, \bibinfo{person}{Bogdan Tofan}, {and} \bibinfo{person}{Wolfgang
  Reif}.} \bibinfo{year}{2011}\natexlab{}.
\newblock \showarticletitle{Proving linearizability with temporal logic}.
\newblock \bibinfo{journal}{{\em Form. Asp. Comput.\/}} \bibinfo{volume}{23},
  \bibinfo{number}{1} (\bibinfo{date}{Jan.} \bibinfo{year}{2011}),
  \bibinfo{pages}{91--112}.
\newblock


\bibitem[\protect\citeauthoryear{Beezer}{Beezer}{2008}]%
        {Beezer2008-eq}
\bibfield{author}{\bibinfo{person}{Robert~Arnold Beezer}.}
  \bibinfo{year}{2008}\natexlab{}.
\newblock \bibinfo{booktitle}{{\em A first course in linear algebra}}.
\newblock \bibinfo{publisher}{Beezer}.
\newblock


\bibitem[\protect\citeauthoryear{Beman~Dawes and Rivera}{Beman~Dawes and
  Rivera}{2018}]%
        {Beman_Dawes2018-bs}
\bibfield{author}{\bibinfo{person}{David~Abrahams Beman~Dawes} {and}
  \bibinfo{person}{Rene Rivera}.} \bibinfo{year}{2018}\natexlab{}.
\newblock \bibinfo{title}{Boost C++ Libraries}.
\newblock
  \bibinfo{howpublished}{\url{https://www.boost.org/users/history/version\_1\_69\_0.html}}.
    (\bibinfo{date}{Dec.} \bibinfo{year}{2018}).
\newblock
\newblock
\shownote{Accessed: 2019-1-15.}


\bibitem[\protect\citeauthoryear{Best and Randell}{Best and Randell}{1981}]%
        {best1981formal}
\bibfield{author}{\bibinfo{person}{Eike Best} {and} \bibinfo{person}{Brian
  Randell}.} \bibinfo{year}{1981}\natexlab{}.
\newblock \showarticletitle{A formal model of atomicity in asynchronous
  systems}.
\newblock \bibinfo{journal}{{\em Acta informatica\/}} \bibinfo{volume}{16},
  \bibinfo{number}{1} (\bibinfo{year}{1981}), \bibinfo{pages}{93--124}.
\newblock


\bibitem[\protect\citeauthoryear{Bouajjani, Emmi, Enea, and Hamza}{Bouajjani
  et~al\mbox{.}}{2015}]%
        {bouajjani2015tractable}
\bibfield{author}{\bibinfo{person}{Ahmed Bouajjani}, \bibinfo{person}{Michael
  Emmi}, \bibinfo{person}{Constantin Enea}, {and} \bibinfo{person}{Jad Hamza}.}
  \bibinfo{year}{2015}\natexlab{}.
\newblock \showarticletitle{Tractable refinement checking for concurrent
  objects}.
\newblock \bibinfo{journal}{{\em Proceedings of the 42nd Annual ACM
  SIGPLAN-SIGACT Symposium on Principles of Programming Languages (POPL
  '15)\/}} \bibinfo{volume}{50}, \bibinfo{number}{1} (\bibinfo{year}{2015}),
  \bibinfo{pages}{651--662}.
\newblock


\bibitem[\protect\citeauthoryear{Bouajjani, Emmi, Enea, and
  Mutluergil}{Bouajjani et~al\mbox{.}}{2017}]%
        {Bouajjani2017-ar}
\bibfield{author}{\bibinfo{person}{Ahmed Bouajjani}, \bibinfo{person}{Michael
  Emmi}, \bibinfo{person}{Constantin Enea}, {and} \bibinfo{person}{Suha~Orhun
  Mutluergil}.} \bibinfo{year}{2017}\natexlab{}.
\newblock \showarticletitle{Proving Linearizability Using Forward Simulations}.
\newblock In \bibinfo{booktitle}{{\em Computer Aided Verification}},
  \bibfield{editor}{\bibinfo{person}{Rupak Majumdar} {and}
  \bibinfo{person}{Viktor Kun{\v c}ak}} (Eds.). \bibinfo{series}{Lecture Notes
  in Computer Science}, Vol.~\bibinfo{volume}{10427}.
  \bibinfo{publisher}{Springer International Publishing},
  \bibinfo{address}{Cham}, \bibinfo{pages}{542--563}.
\newblock


\bibitem[\protect\citeauthoryear{Burckhardt, Dern, Musuvathi, and
  Tan}{Burckhardt et~al\mbox{.}}{2010}]%
        {burckhardt2010line}
\bibfield{author}{\bibinfo{person}{Sebastian Burckhardt},
  \bibinfo{person}{Chris Dern}, \bibinfo{person}{Madanlal Musuvathi}, {and}
  \bibinfo{person}{Roy Tan}.} \bibinfo{year}{2010}\natexlab{}.
\newblock \showarticletitle{Line-up: a complete and automatic linearizability
  checker}.
\newblock \bibinfo{journal}{{\em Proceedings of the 31st ACM SIGPLAN Conference
  on Programming Language Design and Implementation (PLDI'10)\/}}
  \bibinfo{volume}{45}, \bibinfo{number}{6} (\bibinfo{year}{2010}),
  \bibinfo{pages}{330--340}.
\newblock


\bibitem[\protect\citeauthoryear{Chockler, Lynch, Mitra, and Tauber}{Chockler
  et~al\mbox{.}}{2005}]%
        {chockler2005proving}
\bibfield{author}{\bibinfo{person}{Gregory Chockler}, \bibinfo{person}{Nancy
  Lynch}, \bibinfo{person}{Sayan Mitra}, {and} \bibinfo{person}{Joshua
  Tauber}.} \bibinfo{year}{2005}\natexlab{}.
\newblock \showarticletitle{Proving atomicity: An assertional approach}. In
  \bibinfo{booktitle}{{\em International Symposium on Distributed Computing}}.
  Springer, \bibinfo{pages}{152--168}.
\newblock


\bibitem[\protect\citeauthoryear{Debals and De~Lathauwer}{Debals and
  De~Lathauwer}{2015}]%
        {debals2015stochastic}
\bibfield{author}{\bibinfo{person}{Otto Debals} {and} \bibinfo{person}{Lieven
  De~Lathauwer}.} \bibinfo{year}{2015}\natexlab{}.
\newblock \showarticletitle{Stochastic and deterministic tensorization for
  blind signal separation}. In \bibinfo{booktitle}{{\em International
  Conference on Latent Variable Analysis and Signal Separation}}. Springer,
  \bibinfo{pages}{3--13}.
\newblock


\bibitem[\protect\citeauthoryear{Dechev, Pirkelbauer, and Stroustrup}{Dechev
  et~al\mbox{.}}{2006}]%
        {Dechev2006-cb}
\bibfield{author}{\bibinfo{person}{Damian Dechev}, \bibinfo{person}{Peter
  Pirkelbauer}, {and} \bibinfo{person}{Bjarne Stroustrup}.}
  \bibinfo{year}{2006}\natexlab{}.
\newblock \showarticletitle{{Lock-Free} Dynamically Resizable Arrays}. In
  \bibinfo{booktitle}{{\em Principles of Distributed Systems}}.
  \bibinfo{publisher}{Springer Berlin Heidelberg}, \bibinfo{pages}{142--156}.
\newblock


\bibitem[\protect\citeauthoryear{Derrick, Dongol, Schellhorn, Tofan, Travkin,
  and Wehrheim}{Derrick et~al\mbox{.}}{2014}]%
        {Derrick2014-cl}
\bibfield{author}{\bibinfo{person}{John Derrick}, \bibinfo{person}{Brijesh
  Dongol}, \bibinfo{person}{Gerhard Schellhorn}, \bibinfo{person}{Bogdan
  Tofan}, \bibinfo{person}{Oleg Travkin}, {and} \bibinfo{person}{Heike
  Wehrheim}.} \bibinfo{year}{2014}\natexlab{}.
\newblock \showarticletitle{Quiescent Consistency: Defining and Verifying
  Relaxed Linearizability}. In \bibinfo{booktitle}{{\em {FM} 2014: Formal
  Methods}}. \bibinfo{publisher}{Springer International Publishing},
  \bibinfo{pages}{200--214}.
\newblock


\bibitem[\protect\citeauthoryear{Derrick, Schellhorn, and Wehrheim}{Derrick
  et~al\mbox{.}}{2007}]%
        {Derrick2007-te}
\bibfield{author}{\bibinfo{person}{John Derrick}, \bibinfo{person}{Gerhard
  Schellhorn}, {and} \bibinfo{person}{Heike Wehrheim}.}
  \bibinfo{year}{2007}\natexlab{}.
\newblock \showarticletitle{Proving Linearizability Via Non-atomic Refinement}.
\newblock In \bibinfo{booktitle}{{\em Integrated Formal Methods}},
  \bibfield{editor}{\bibinfo{person}{Jim Davies} {and} \bibinfo{person}{Jeremy
  Gibbons}} (Eds.). \bibinfo{series}{Lecture Notes in Computer Science},
  Vol.~\bibinfo{volume}{4591}. \bibinfo{publisher}{Springer Berlin Heidelberg},
  \bibinfo{address}{Berlin, Heidelberg}, \bibinfo{pages}{195--214}.
\newblock


\bibitem[\protect\citeauthoryear{Derrick, Schellhorn, and Wehrheim}{Derrick
  et~al\mbox{.}}{2011}]%
        {Derrick2011-hy}
\bibfield{author}{\bibinfo{person}{John Derrick}, \bibinfo{person}{Gerhard
  Schellhorn}, {and} \bibinfo{person}{Heike Wehrheim}.}
  \bibinfo{year}{2011}\natexlab{}.
\newblock \showarticletitle{Verifying Linearisability with Potential
  Linearisation Points}. In \bibinfo{booktitle}{{\em {FM} 2011: Formal
  Methods}}. \bibinfo{publisher}{Springer Berlin Heidelberg},
  \bibinfo{pages}{323--337}.
\newblock


\bibitem[\protect\citeauthoryear{Descartes}{Descartes}{1903}]%
        {Descartes1903-wy}
\bibfield{author}{\bibinfo{person}{Ren{\'e} Descartes}.}
  \bibinfo{year}{1903}\natexlab{}.
\newblock \bibinfo{booktitle}{{\em The Meditations, and Selections from the
  Principles of Ren{\'e} Descartes (1596-1650)}}.
\newblock \bibinfo{publisher}{Open Court}.
\newblock


\bibitem[\protect\citeauthoryear{Dodds, Haas, and Kirsch}{Dodds
  et~al\mbox{.}}{2014}]%
        {Dodds2014-gi}
\bibfield{author}{\bibinfo{person}{Mike Dodds}, \bibinfo{person}{Andreas Haas},
  {and} \bibinfo{person}{Christoph~M Kirsch}.} \bibinfo{year}{2014}\natexlab{}.
\newblock \showarticletitle{Fast concurrent data-structures through explicit
  timestamping}.
\newblock \bibinfo{journal}{{\em Department of Computer Sciences, Universitt
  Salzburg, Tech. Rep\/}}  \bibinfo{volume}{3} (\bibinfo{year}{2014}).
\newblock


\bibitem[\protect\citeauthoryear{Elmas, Qadeer, Sezgin, Subasi, and
  Tasiran}{Elmas et~al\mbox{.}}{2010}]%
        {Elmas2010-bw}
\bibfield{author}{\bibinfo{person}{Tayfun Elmas}, \bibinfo{person}{Shaz
  Qadeer}, \bibinfo{person}{Ali Sezgin}, \bibinfo{person}{Omer Subasi}, {and}
  \bibinfo{person}{Serdar Tasiran}.} \bibinfo{year}{2010}\natexlab{}.
\newblock \showarticletitle{Simplifying Linearizability Proofs with Reduction
  and Abstraction}.
\newblock In \bibinfo{booktitle}{{\em Tools and Algorithms for the Construction
  and Analysis of Systems}}, \bibfield{editor}{\bibinfo{person}{Javier Esparza}
  {and} \bibinfo{person}{Rupak Majumdar}} (Eds.). \bibinfo{series}{Lecture
  Notes in Computer Science}, Vol.~\bibinfo{volume}{6015}.
  \bibinfo{publisher}{Springer Berlin Heidelberg}, \bibinfo{address}{Berlin,
  Heidelberg}, \bibinfo{pages}{296--311}.
\newblock


\bibitem[\protect\citeauthoryear{Emmi and Enea}{Emmi and Enea}{2017}]%
        {Emmi2017-pl}
\bibfield{author}{\bibinfo{person}{Michael Emmi} {and}
  \bibinfo{person}{Constantin Enea}.} \bibinfo{year}{2017}\natexlab{}.
\newblock \showarticletitle{Sound, complete, and tractable linearizability
  monitoring for concurrent collections}.
\newblock \bibinfo{journal}{{\em Proceedings of the ACM on Programming
  Languages\/}} \bibinfo{volume}{2}, \bibinfo{number}{POPL}
  (\bibinfo{date}{Dec.} \bibinfo{year}{2017}), \bibinfo{pages}{25}.
\newblock


\bibitem[\protect\citeauthoryear{Emmi, Enea, and Hamza}{Emmi
  et~al\mbox{.}}{2015}]%
        {emmi2015monitoring}
\bibfield{author}{\bibinfo{person}{Michael Emmi}, \bibinfo{person}{Constantin
  Enea}, {and} \bibinfo{person}{Jad Hamza}.} \bibinfo{year}{2015}\natexlab{}.
\newblock \showarticletitle{Monitoring refinement via symbolic reasoning}. In
  \bibinfo{booktitle}{{\em Proceedings of the 36th ACM SIGPLAN Conference on
  Programming Language Design and Implementation (PLDI '15)}},
  Vol.~\bibinfo{volume}{50}. ACM, \bibinfo{pages}{260--269}.
\newblock


\bibitem[\protect\citeauthoryear{Feldman, Enea, Morrison, Rinetzky, and
  Shoham}{Feldman et~al\mbox{.}}{2018}]%
        {Feldman2018-be}
\bibfield{author}{\bibinfo{person}{Yotam M~Y Feldman},
  \bibinfo{person}{Constantin Enea}, \bibinfo{person}{Adam Morrison},
  \bibinfo{person}{Noam Rinetzky}, {and} \bibinfo{person}{Sharon Shoham}.}
  \bibinfo{year}{2018}\natexlab{}.
\newblock \showarticletitle{Order out of Chaos: Proving Linearizability Using
  Local Views}.
\newblock  (\bibinfo{date}{May} \bibinfo{year}{2018}).
\newblock
\showeprint[arxiv]{cs.DC/1805.03992}


\bibitem[\protect\citeauthoryear{Flanagan, Flanagan, and Freund}{Flanagan
  et~al\mbox{.}}{2004}]%
        {flanagan2004atomizer}
\bibfield{author}{\bibinfo{person}{Cormac Flanagan}, \bibinfo{person}{Cormac
  Flanagan}, {and} \bibinfo{person}{Stephen~N Freund}.}
  \bibinfo{year}{2004}\natexlab{}.
\newblock \showarticletitle{Atomizer: a dynamic atomicity checker for
  multithreaded programs}. In \bibinfo{booktitle}{{\em Proceedings of the 31st
  ACM SIGPLAN-SIGACT symposium on Principles of programming languages}},
  Vol.~\bibinfo{volume}{39}. ACM, \bibinfo{pages}{256--267}.
\newblock


\bibitem[\protect\citeauthoryear{Flanagan, Freund, and Yi}{Flanagan
  et~al\mbox{.}}{2008}]%
        {flanagan2008velodrome}
\bibfield{author}{\bibinfo{person}{Cormac Flanagan}, \bibinfo{person}{Stephen~N
  Freund}, {and} \bibinfo{person}{Jaeheon Yi}.}
  \bibinfo{year}{2008}\natexlab{}.
\newblock \showarticletitle{Velodrome: a sound and complete dynamic atomicity
  checker for multithreaded programs}.
\newblock \bibinfo{journal}{{\em Proceedings of the 29th ACM SIGPLAN Conference
  on Programming Language Design and Implementation\/}} \bibinfo{volume}{43},
  \bibinfo{number}{6} (\bibinfo{year}{2008}), \bibinfo{pages}{293--303}.
\newblock


\bibitem[\protect\citeauthoryear{Flanagan and Godefroid}{Flanagan and
  Godefroid}{2005}]%
        {flanagan2005dynamic}
\bibfield{author}{\bibinfo{person}{Cormac Flanagan} {and}
  \bibinfo{person}{Patrice Godefroid}.} \bibinfo{year}{2005}\natexlab{}.
\newblock \showarticletitle{Dynamic partial-order reduction for model checking
  software}. In \bibinfo{booktitle}{{\em Proceedings of the 32nd ACM
  SIGPLAN-SIGACT symposium on Principles of programming languages}},
  Vol.~\bibinfo{volume}{40}. ACM, \bibinfo{pages}{110--121}.
\newblock


\bibitem[\protect\citeauthoryear{Flanagan and Qadeer}{Flanagan and
  Qadeer}{2003}]%
        {flanagan2003type}
\bibfield{author}{\bibinfo{person}{Cormac Flanagan} {and} \bibinfo{person}{Shaz
  Qadeer}.} \bibinfo{year}{2003}\natexlab{}.
\newblock \showarticletitle{A type and effect system for atomicity}. In
  \bibinfo{booktitle}{{\em Proceedings of the ACM SIGPLAN 2003 conference on
  Programming language design and implementation}}, Vol.~\bibinfo{volume}{38}.
  ACM, \bibinfo{pages}{338--349}.
\newblock


\bibitem[\protect\citeauthoryear{Gogolla, Drosten, Lipeck, and Ehrich}{Gogolla
  et~al\mbox{.}}{1984}]%
        {Gogolla1984-sx}
\bibfield{author}{\bibinfo{person}{M Gogolla}, \bibinfo{person}{K Drosten},
  \bibinfo{person}{U Lipeck}, {and} \bibinfo{person}{H-D Ehrich}.}
  \bibinfo{year}{1984}\natexlab{}.
\newblock \showarticletitle{Algebraic and operational semantics of
  specifications allowing exceptions and errors}.
\newblock \bibinfo{journal}{{\em Theor. Comput. Sci.\/}} \bibinfo{volume}{34},
  \bibinfo{number}{3} (\bibinfo{date}{Jan.} \bibinfo{year}{1984}),
  \bibinfo{pages}{289--313}.
\newblock


\bibitem[\protect\citeauthoryear{Goodfellow, Bengio, and Courville}{Goodfellow
  et~al\mbox{.}}{2016}]%
        {Goodfellow-et-al-2016}
\bibfield{author}{\bibinfo{person}{Ian Goodfellow}, \bibinfo{person}{Yoshua
  Bengio}, {and} \bibinfo{person}{Aaron Courville}.}
  \bibinfo{year}{2016}\natexlab{}.
\newblock \bibinfo{booktitle}{{\em Deep Learning}}.
\newblock \bibinfo{publisher}{MIT Press}.
\newblock
\newblock
\shownote{\url{http://www.deeplearningbook.org}.}


\bibitem[\protect\citeauthoryear{Grasedyck}{Grasedyck}{2010}]%
        {grasedyck2010polynomial}
\bibfield{author}{\bibinfo{person}{Lars Grasedyck}.}
  \bibinfo{year}{2010}\natexlab{}.
\newblock \bibinfo{booktitle}{{\em Polynomial approximation in hierarchical
  Tucker format by vector-tensorization}}.
\newblock \bibinfo{publisher}{Inst. f{\"u}r Geometrie und Praktische
  Mathematik}.
\newblock


\bibitem[\protect\citeauthoryear{Gruber, Tr{\"a}ff, and Wimmer}{Gruber
  et~al\mbox{.}}{2016}]%
        {Gruber2016-rz}
\bibfield{author}{\bibinfo{person}{Jakob Gruber},
  \bibinfo{person}{Jesper~Larsson Tr{\"a}ff}, {and} \bibinfo{person}{Martin
  Wimmer}.} \bibinfo{year}{2016}\natexlab{}.
\newblock \showarticletitle{Benchmarking Concurrent Priority Queues:
  Performance of {k-LSM} and Related Data Structures}.
\newblock  (\bibinfo{date}{March} \bibinfo{year}{2016}).
\newblock
\showeprint[arxiv]{cs.DS/1603.05047}


\bibitem[\protect\citeauthoryear{Guerraoui, Kuncak, and Losa}{Guerraoui
  et~al\mbox{.}}{2012}]%
        {Guerraoui2012-hw}
\bibfield{author}{\bibinfo{person}{Rachid Guerraoui}, \bibinfo{person}{Viktor
  Kuncak}, {and} \bibinfo{person}{Giuliano Losa}.}
  \bibinfo{year}{2012}\natexlab{}.
\newblock \showarticletitle{Speculative Linearizability}. In
  \bibinfo{booktitle}{{\em Proceedings of the 33rd {ACM} {SIGPLAN} Conference
  on Programming Language Design and Implementation}} {\em
  (\bibinfo{series}{PLDI '12})}. \bibinfo{publisher}{ACM},
  \bibinfo{address}{New York, NY, USA}, \bibinfo{pages}{55--66}.
\newblock


\bibitem[\protect\citeauthoryear{Guttag}{Guttag}{1976}]%
        {Guttag1976-nq}
\bibfield{author}{\bibinfo{person}{John Guttag}.}
  \bibinfo{year}{1976}\natexlab{}.
\newblock \showarticletitle{Abstract data types and the development of data
  structures}. In \bibinfo{booktitle}{{\em Proceedings of the 1976 conference
  on Data : Abstraction, definition and structure}}, Vol.~\bibinfo{volume}{11}.
  \bibinfo{publisher}{ACM}, \bibinfo{pages}{72}.
\newblock


\bibitem[\protect\citeauthoryear{Guttag, Horowitz, and Musser}{Guttag
  et~al\mbox{.}}{1978}]%
        {guttag1978abstract}
\bibfield{author}{\bibinfo{person}{John~V Guttag}, \bibinfo{person}{Ellis
  Horowitz}, {and} \bibinfo{person}{David~R Musser}.}
  \bibinfo{year}{1978}\natexlab{}.
\newblock \showarticletitle{Abstract data types and software validation}.
\newblock \bibinfo{journal}{{\it Commun. ACM}} \bibinfo{volume}{21},
  \bibinfo{number}{12} (\bibinfo{year}{1978}), \bibinfo{pages}{1048--1064}.
\newblock


\bibitem[\protect\citeauthoryear{Haas}{Haas}{2015}]%
        {TSQueue2014}
\bibfield{author}{\bibinfo{person}{A Haas}.} \bibinfo{year}{2015}\natexlab{}.
\newblock {\em \bibinfo{title}{Fast concurrent data structures through
  timestamping}}.
\newblock \bibinfo{thesistype}{Ph.D. Dissertation}. \bibinfo{school}{PhD
  thesis, University of Salzburg, Salzburg, Austria}.
\newblock


\bibitem[\protect\citeauthoryear{Haas, Lippautz, Henzinger, Payer, Sokolova,
  Kirsch, and Sezgin}{Haas et~al\mbox{.}}{2013}]%
        {Haas2013-bn}
\bibfield{author}{\bibinfo{person}{Andreas Haas}, \bibinfo{person}{Michael
  Lippautz}, \bibinfo{person}{Thomas~A Henzinger}, \bibinfo{person}{Hannes
  Payer}, \bibinfo{person}{Ana Sokolova}, \bibinfo{person}{Christoph~M Kirsch},
  {and} \bibinfo{person}{Ali Sezgin}.} \bibinfo{year}{2013}\natexlab{}.
\newblock \showarticletitle{Distributed Queues in Shared Memory: Multicore
  Performance and Scalability Through Quantitative Relaxation}. In
  \bibinfo{booktitle}{{\em Proceedings of the {ACM} International Conference on
  Computing Frontiers}} {\em (\bibinfo{series}{CF '13})}.
  \bibinfo{publisher}{ACM}, \bibinfo{address}{New York, NY, USA},
  \bibinfo{pages}{17:1--17:9}.
\newblock


\bibitem[\protect\citeauthoryear{Harris, Fraser, and Pratt}{Harris
  et~al\mbox{.}}{2002}]%
        {harris2002practical}
\bibfield{author}{\bibinfo{person}{Timothy~L Harris}, \bibinfo{person}{Keir
  Fraser}, {and} \bibinfo{person}{Ian~A Pratt}.}
  \bibinfo{year}{2002}\natexlab{}.
\newblock \showarticletitle{A practical multi-word compare-and-swap operation}.
  In \bibinfo{booktitle}{{\em International Symposium on Distributed
  Computing}}. Springer, \bibinfo{pages}{265--279}.
\newblock


\bibitem[\protect\citeauthoryear{Hendler, Shavit, and Yerushalmi}{Hendler
  et~al\mbox{.}}{2004}]%
        {Hendler2004-ho}
\bibfield{author}{\bibinfo{person}{Danny Hendler}, \bibinfo{person}{Nir
  Shavit}, {and} \bibinfo{person}{Lena Yerushalmi}.}
  \bibinfo{year}{2004}\natexlab{}.
\newblock \showarticletitle{A Scalable Lock-free Stack Algorithm}. In
  \bibinfo{booktitle}{{\em Proceedings of the Sixteenth Annual {ACM} Symposium
  on Parallelism in Algorithms and Architectures}} {\em (\bibinfo{series}{SPAA
  '04})}. \bibinfo{publisher}{ACM}, \bibinfo{address}{New York, NY, USA},
  \bibinfo{pages}{206--215}.
\newblock


\bibitem[\protect\citeauthoryear{Henzinger, Kirsch, Payer, Sezgin, and
  Sokolova}{Henzinger et~al\mbox{.}}{2013}]%
        {Henzinger2013-pv}
\bibfield{author}{\bibinfo{person}{Thomas~A Henzinger},
  \bibinfo{person}{Christoph~M Kirsch}, \bibinfo{person}{Hannes Payer},
  \bibinfo{person}{Ali Sezgin}, {and} \bibinfo{person}{Ana Sokolova}.}
  \bibinfo{year}{2013}\natexlab{}.
\newblock \showarticletitle{Quantitative Relaxation of Concurrent Data
  Structures}. In \bibinfo{booktitle}{{\em Proceedings of the 40th Annual {ACM}
  {SIGPLAN-SIGACT} Symposium on Principles of Programming Languages}} {\em
  (\bibinfo{series}{POPL '13})}. \bibinfo{publisher}{ACM},
  \bibinfo{address}{New York, NY, USA}, \bibinfo{pages}{317--328}.
\newblock


\bibitem[\protect\citeauthoryear{Herlihy}{Herlihy}{1991}]%
        {Herlihy1991-jf}
\bibfield{author}{\bibinfo{person}{Maurice Herlihy}.}
  \bibinfo{year}{1991}\natexlab{}.
\newblock \showarticletitle{Wait-free Synchronization}.
\newblock \bibinfo{journal}{{\em ACM Trans. Program. Lang. Syst.\/}}
  \bibinfo{volume}{13}, \bibinfo{number}{1} (\bibinfo{date}{Jan.}
  \bibinfo{year}{1991}), \bibinfo{pages}{124--149}.
\newblock


\bibitem[\protect\citeauthoryear{Herlihy and Shavit}{Herlihy and
  Shavit}{2012}]%
        {Herlihy2011-yj}
\bibfield{author}{\bibinfo{person}{Maurice Herlihy} {and} \bibinfo{person}{Nir
  Shavit}.} \bibinfo{year}{2012}\natexlab{}.
\newblock \bibinfo{booktitle}{{\em The Art of Multiprocessor Programming}}.
\newblock \bibinfo{publisher}{Morgan Kaufmann. Revised Reprint. ISBN:
  0123973375}.
\newblock
\showISBNx{978-0123973375}


\bibitem[\protect\citeauthoryear{Herlihy and Wing}{Herlihy and Wing}{1990a}]%
        {herlihy1990linearizability}
\bibfield{author}{\bibinfo{person}{Maurice~P Herlihy} {and}
  \bibinfo{person}{Jeannette~M Wing}.} \bibinfo{year}{1990}\natexlab{a}.
\newblock \showarticletitle{Linearizability: A correctness condition for
  concurrent objects}.
\newblock \bibinfo{journal}{{\em ACM Transactions on Programming Languages and
  Systems (TOPLAS)\/}} \bibinfo{volume}{12}, \bibinfo{number}{3}
  (\bibinfo{year}{1990}), \bibinfo{pages}{463--492}.
\newblock


\bibitem[\protect\citeauthoryear{Herlihy and Wing}{Herlihy and Wing}{1990b}]%
        {Herlihy1990-ms}
\bibfield{author}{\bibinfo{person}{Maurice~P Herlihy} {and}
  \bibinfo{person}{Jeannette~M Wing}.} \bibinfo{year}{1990}\natexlab{b}.
\newblock \showarticletitle{Linearizability: A Correctness Condition for
  Concurrent Objects}.
\newblock \bibinfo{journal}{{\em ACM Trans. Program. Lang. Syst.\/}}
  \bibinfo{volume}{12}, \bibinfo{number}{3} (\bibinfo{date}{July}
  \bibinfo{year}{1990}), \bibinfo{pages}{463--492}.
\newblock


\bibitem[\protect\citeauthoryear{Hoare}{Hoare}{1978}]%
        {hoare1978proof}
\bibfield{author}{\bibinfo{person}{Charles Antony~Richard Hoare}.}
  \bibinfo{year}{1978}\natexlab{}.
\newblock \showarticletitle{Proof of correctness of data representations}.
\newblock In \bibinfo{booktitle}{{\em Programming Methodology}}.
  \bibinfo{publisher}{Springer}, \bibinfo{pages}{269--281}.
\newblock


\bibitem[\protect\citeauthoryear{Khyzha, Dodds, Gotsman, and Parkinson}{Khyzha
  et~al\mbox{.}}{2017}]%
        {Khyzha2017-ez}
\bibfield{author}{\bibinfo{person}{Artem Khyzha}, \bibinfo{person}{Mike Dodds},
  \bibinfo{person}{Alexey Gotsman}, {and} \bibinfo{person}{Matthew Parkinson}.}
  \bibinfo{year}{2017}\natexlab{}.
\newblock \showarticletitle{Proving Linearizability Using Partial Orders}. In
  \bibinfo{booktitle}{{\em Programming Languages and Systems}}.
  \bibinfo{publisher}{Springer Berlin Heidelberg}, \bibinfo{pages}{639--667}.
\newblock


\bibitem[\protect\citeauthoryear{Khyzha, Gotsman, and Parkinson}{Khyzha
  et~al\mbox{.}}{2016}]%
        {Khyzha2016-lv}
\bibfield{author}{\bibinfo{person}{Artem Khyzha}, \bibinfo{person}{Alexey
  Gotsman}, {and} \bibinfo{person}{Matthew Parkinson}.}
  \bibinfo{year}{2016}\natexlab{}.
\newblock \showarticletitle{A Generic Logic for Proving Linearizability}.
\newblock In \bibinfo{booktitle}{{\em {FM} 2016: Formal Methods}},
  \bibfield{editor}{\bibinfo{person}{John Fitzgerald},
  \bibinfo{person}{Constance Heitmeyer}, \bibinfo{person}{Stefania Gnesi},
  {and} \bibinfo{person}{Anna Philippou}} (Eds.). \bibinfo{series}{Lecture
  Notes in Computer Science}, Vol.~\bibinfo{volume}{9995}.
  \bibinfo{publisher}{Springer International Publishing},
  \bibinfo{address}{Cham}, \bibinfo{pages}{426--443}.
\newblock


\bibitem[\protect\citeauthoryear{Kirsch, Lippautz, and Payer}{Kirsch
  et~al\mbox{.}}{2013}]%
        {kirsch2013fast}
\bibfield{author}{\bibinfo{person}{Christoph~M Kirsch},
  \bibinfo{person}{Michael Lippautz}, {and} \bibinfo{person}{Hannes Payer}.}
  \bibinfo{year}{2013}\natexlab{}.
\newblock \showarticletitle{Fast and scalable, lock-free k-FIFO queues}. In
  \bibinfo{booktitle}{{\em International Conference on Parallel Computing
  Technologies}}. Springer, \bibinfo{pages}{208--223}.
\newblock


\bibitem[\protect\citeauthoryear{Knuth}{Knuth}{1998}]%
        {Knuth1998-td}
\bibfield{author}{\bibinfo{person}{Donald~E Knuth}.}
  \bibinfo{year}{1998}\natexlab{}.
\newblock \bibinfo{booktitle}{{\em The Art of Computer Programming: Volume 3:
  Sorting and Searching}}.
\newblock \bibinfo{publisher}{Addison-Wesley Professional}.
\newblock


\bibitem[\protect\citeauthoryear{Kogan and Petrank}{Kogan and Petrank}{2012}]%
        {kogan2012methodology}
\bibfield{author}{\bibinfo{person}{Alex Kogan} {and} \bibinfo{person}{Erez
  Petrank}.} \bibinfo{year}{2012}\natexlab{}.
\newblock \showarticletitle{A methodology for creating fast wait-free data
  structures}. In \bibinfo{booktitle}{{\em Proceedings of the 17th ACM SIGPLAN
  symposium on Principles and Practice of Parallel Programming}},
  Vol.~\bibinfo{volume}{47}. ACM, \bibinfo{pages}{141--150}.
\newblock


\bibitem[\protect\citeauthoryear{Lamport}{Lamport}{1978}]%
        {Lamport1978-ip}
\bibfield{author}{\bibinfo{person}{Leslie Lamport}.}
  \bibinfo{year}{1978}\natexlab{}.
\newblock \showarticletitle{Time, clocks, and the ordering of events in a
  distributed system}.
\newblock \bibinfo{journal}{{\em Commun. ACM\/}} \bibinfo{volume}{21},
  \bibinfo{number}{7} (\bibinfo{date}{July} \bibinfo{year}{1978}),
  \bibinfo{pages}{558--565}.
\newblock


\bibitem[\protect\citeauthoryear{Lamport}{Lamport}{1979}]%
        {Lamport1979-rd}
\bibfield{author}{\bibinfo{person}{Leslie Lamport}.}
  \bibinfo{year}{1979}\natexlab{}.
\newblock \showarticletitle{How to make a multiprocessor computer that
  correctly executes multiprocess program}.
\newblock \bibinfo{journal}{{\em IEEE Trans. Comput.\/}} \bibinfo{number}{9}
  (\bibinfo{year}{1979}), \bibinfo{pages}{690--691}.
\newblock


\bibitem[\protect\citeauthoryear{Liang and Feng}{Liang and Feng}{2013}]%
        {Liang2013-eh}
\bibfield{author}{\bibinfo{person}{Hongjin Liang} {and} \bibinfo{person}{Xinyu
  Feng}.} \bibinfo{year}{2013}\natexlab{}.
\newblock \showarticletitle{Modular Verification of Linearizability with
  Non-fixed Linearization Points}.
\newblock \bibinfo{journal}{{\em SIGPLAN Not.\/}} \bibinfo{volume}{48},
  \bibinfo{number}{6} (\bibinfo{date}{June} \bibinfo{year}{2013}),
  \bibinfo{pages}{459--470}.
\newblock


\bibitem[\protect\citeauthoryear{Lipton}{Lipton}{1975}]%
        {lipton1975reduction}
\bibfield{author}{\bibinfo{person}{Richard~J Lipton}.}
  \bibinfo{year}{1975}\natexlab{}.
\newblock \showarticletitle{Reduction: A method of proving properties of
  parallel programs}.
\newblock \bibinfo{journal}{{\it Commun. ACM}} \bibinfo{volume}{18},
  \bibinfo{number}{12} (\bibinfo{year}{1975}), \bibinfo{pages}{717--721}.
\newblock


\bibitem[\protect\citeauthoryear{Liskov and Wing}{Liskov and Wing}{1994}]%
        {Liskov1994-rn}
\bibfield{author}{\bibinfo{person}{Barbara~H Liskov} {and}
  \bibinfo{person}{Jeannette~M Wing}.} \bibinfo{year}{1994}\natexlab{}.
\newblock \showarticletitle{A Behavioral Notion of Subtyping}.
\newblock \bibinfo{journal}{{\em ACM Trans. Program. Lang. Syst.\/}}
  \bibinfo{volume}{16}, \bibinfo{number}{6} (\bibinfo{date}{Nov.}
  \bibinfo{year}{1994}), \bibinfo{pages}{1811--1841}.
\newblock


\bibitem[\protect\citeauthoryear{Michael and Scott}{Michael and Scott}{1995}]%
        {Michael1995-vf}
\bibfield{author}{\bibinfo{person}{M~M Michael} {and} \bibinfo{person}{M~L
  Scott}.} \bibinfo{year}{1995}\natexlab{}.
\newblock \showarticletitle{Simple, Fast, and Practical {Non-Blocking} and
  Blocking Concurrent Queue Algorithms}.
\newblock \bibinfo{journal}{{\em Technical Report 600\/}}
  (\bibinfo{year}{1995}).
\newblock


\bibitem[\protect\citeauthoryear{Morrison and Afek}{Morrison and Afek}{2013}]%
        {Morrison2013-mc}
\bibfield{author}{\bibinfo{person}{A Morrison} {and} \bibinfo{person}{Y Afek}.}
  \bibinfo{year}{2013}\natexlab{}.
\newblock \showarticletitle{Fast concurrent queues for x86 processors}.
\newblock \bibinfo{journal}{{\em Proceedings of the 18th ACM SIGPLAN symposium
  on Principles and practice of parallel programming\/}}
  (\bibinfo{year}{2013}).
\newblock


\bibitem[\protect\citeauthoryear{Nanevski, Banerjee, Delbianco, and
  F{\'a}bregas}{Nanevski et~al\mbox{.}}{2019}]%
        {nanevski2019specifying}
\bibfield{author}{\bibinfo{person}{Aleksandar Nanevski},
  \bibinfo{person}{Anindya Banerjee}, \bibinfo{person}{Germ{\'a}n~Andr{\'e}s
  Delbianco}, {and} \bibinfo{person}{Ignacio F{\'a}bregas}.}
  \bibinfo{year}{2019}\natexlab{}.
\newblock \showarticletitle{Specifying Concurrent Programs in Separation Logic:
  Morphisms and Simulations}.
\newblock \bibinfo{journal}{{\em arXiv preprint arXiv:1904.07136\/}}
  (\bibinfo{year}{2019}).
\newblock


\bibitem[\protect\citeauthoryear{{National Research Council}, {Division on
  Engineering and Physical Sciences}, {Computer Science and Telecommunications
  Board}, and {Committee on Sustaining Growth in Computing
  Performance}}{{National Research Council} et~al\mbox{.}}{2011}]%
        {National_Research_Council2011-ia}
\bibfield{author}{\bibinfo{person}{{National Research Council}},
  \bibinfo{person}{{Division on Engineering and Physical Sciences}},
  \bibinfo{person}{{Computer Science and Telecommunications Board}}, {and}
  \bibinfo{person}{{Committee on Sustaining Growth in Computing Performance}}.}
  \bibinfo{year}{2011}\natexlab{}.
\newblock \bibinfo{booktitle}{{\em The Future of Computing Performance: Game
  Over or Next Level?}}
\newblock \bibinfo{publisher}{National Academies Press}.
\newblock


\bibitem[\protect\citeauthoryear{O'Hearn, Rinetzky, Vechev, Yahav, and
  Yorsh}{O'Hearn et~al\mbox{.}}{2010}]%
        {OHearn2010-dn}
\bibfield{author}{\bibinfo{person}{Peter~W O'Hearn}, \bibinfo{person}{Noam
  Rinetzky}, \bibinfo{person}{Martin~T Vechev}, \bibinfo{person}{Eran Yahav},
  {and} \bibinfo{person}{Greta Yorsh}.} \bibinfo{year}{2010}\natexlab{}.
\newblock \showarticletitle{Verifying Linearizability with Hindsight}. In
  \bibinfo{booktitle}{{\em Proceedings of the 29th {ACM} {SIGACT-SIGOPS}
  Symposium on Principles of Distributed Computing}} {\em
  (\bibinfo{series}{PODC '10})}. \bibinfo{publisher}{ACM},
  \bibinfo{address}{New York, NY, USA}, \bibinfo{pages}{85--94}.
\newblock


\bibitem[\protect\citeauthoryear{Ou and Demsky}{Ou and Demsky}{2017}]%
        {ou2017checking}
\bibfield{author}{\bibinfo{person}{Peizhao Ou} {and} \bibinfo{person}{Brian
  Demsky}.} \bibinfo{year}{2017}\natexlab{}.
\newblock \showarticletitle{Checking concurrent data structures under the C/C++
  11 memory model}. In \bibinfo{booktitle}{{\em Proceedings of the 22nd ACM
  SIGPLAN Symposium on Principles and Practice of Parallel Programming (PPoPP
  '17)}}, Vol.~\bibinfo{volume}{52}. ACM, \bibinfo{pages}{45--59}.
\newblock


\bibitem[\protect\citeauthoryear{Papadimitriou}{Papadimitriou}{1979}]%
        {papadimitriou1979serializability}
\bibfield{author}{\bibinfo{person}{Christos~H Papadimitriou}.}
  \bibinfo{year}{1979}\natexlab{}.
\newblock \showarticletitle{The serializability of concurrent database
  updates}.
\newblock \bibinfo{journal}{{\em Journal of the ACM (JACM)\/}}
  \bibinfo{volume}{26}, \bibinfo{number}{4} (\bibinfo{year}{1979}),
  \bibinfo{pages}{631--653}.
\newblock


\bibitem[\protect\citeauthoryear{Rihani, Sanders, and Dementiev}{Rihani
  et~al\mbox{.}}{2015}]%
        {Rihani2015-cf}
\bibfield{author}{\bibinfo{person}{Hamza Rihani}, \bibinfo{person}{Peter
  Sanders}, {and} \bibinfo{person}{Roman Dementiev}.}
  \bibinfo{year}{2015}\natexlab{}.
\newblock \showarticletitle{Brief Announcement: MultiQueues: Simple Relaxed
  Concurrent Priority Queues}. In \bibinfo{booktitle}{{\em Proceedings of the
  27th ACM Symposium on Parallelism in Algorithms and Architectures}} {\em
  (\bibinfo{series}{SPAA '15})}. \bibinfo{publisher}{ACM},
  \bibinfo{address}{New York, NY, USA}, \bibinfo{pages}{80--82}.
\newblock
\showISBNx{978-1-4503-3588-1}
\showDOI{%
\url{https://doi.org/10.1145/2755573.2755616}}


\bibitem[\protect\citeauthoryear{Schellhorn, Derrick, and Wehrheim}{Schellhorn
  et~al\mbox{.}}{2014}]%
        {Schellhorn2014-mr}
\bibfield{author}{\bibinfo{person}{Gerhard Schellhorn}, \bibinfo{person}{John
  Derrick}, {and} \bibinfo{person}{Heike Wehrheim}.}
  \bibinfo{year}{2014}\natexlab{}.
\newblock \showarticletitle{A Sound and Complete Proof Technique for
  Linearizability of Concurrent Data Structures}.
\newblock \bibinfo{journal}{{\em ACM Trans. Comput. Log.\/}}
  \bibinfo{volume}{15}, \bibinfo{number}{4} (\bibinfo{date}{Sept.}
  \bibinfo{year}{2014}), \bibinfo{pages}{31:1--31:37}.
\newblock


\bibitem[\protect\citeauthoryear{Scherer and Scott}{Scherer and Scott}{2004}]%
        {scherer2004nonblocking}
\bibfield{author}{\bibinfo{person}{William~N Scherer} {and}
  \bibinfo{person}{Michael~L Scott}.} \bibinfo{year}{2004}\natexlab{}.
\newblock \showarticletitle{Nonblocking concurrent data structures with
  condition synchronization}. In \bibinfo{booktitle}{{\em International
  Symposium on Distributed Computing}}. Springer, \bibinfo{pages}{174--187}.
\newblock


\bibitem[\protect\citeauthoryear{Sergey, Nanevski, Banerjee, and
  Delbianco}{Sergey et~al\mbox{.}}{2016}]%
        {sergey2016hoare}
\bibfield{author}{\bibinfo{person}{Ilya Sergey}, \bibinfo{person}{Aleksandar
  Nanevski}, \bibinfo{person}{Anindya Banerjee}, {and}
  \bibinfo{person}{Germ{\'a}n~Andr{\'e}s Delbianco}.}
  \bibinfo{year}{2016}\natexlab{}.
\newblock \showarticletitle{Hoare-style specifications as correctness
  conditions for non-linearizable concurrent objects}.
\newblock \bibinfo{journal}{{\em Proceedings of the 2016 ACM SIGPLAN
  International Conference on Object-Oriented Programming, Systems, Languages,
  and Applications (OOPSLA '16)\/}} \bibinfo{volume}{51}, \bibinfo{number}{10}
  (\bibinfo{year}{2016}), \bibinfo{pages}{92--110}.
\newblock


\bibitem[\protect\citeauthoryear{Shacham, Bronson, Aiken, Sagiv, Vechev, and
  Yahav}{Shacham et~al\mbox{.}}{2011}]%
        {shacham2011testing}
\bibfield{author}{\bibinfo{person}{Ohad Shacham}, \bibinfo{person}{Nathan
  Bronson}, \bibinfo{person}{Alex Aiken}, \bibinfo{person}{Mooly Sagiv},
  \bibinfo{person}{Martin Vechev}, {and} \bibinfo{person}{Eran Yahav}.}
  \bibinfo{year}{2011}\natexlab{}.
\newblock \showarticletitle{Testing atomicity of composed concurrent
  operations}. In \bibinfo{booktitle}{{\em Proceedings of the 2011 ACM
  international conference on Object oriented programming systems languages and
  applications}}, Vol.~\bibinfo{volume}{46}. ACM, \bibinfo{pages}{51--64}.
\newblock


\bibitem[\protect\citeauthoryear{Shannon}{Shannon}{1948}]%
        {shannon1948mathematical}
\bibfield{author}{\bibinfo{person}{Claude~Elwood Shannon}.}
  \bibinfo{year}{1948}\natexlab{}.
\newblock \showarticletitle{A mathematical theory of communication}.
\newblock \bibinfo{journal}{{\em Bell system technical journal\/}}
  \bibinfo{volume}{27}, \bibinfo{number}{3} (\bibinfo{year}{1948}),
  \bibinfo{pages}{379--423}.
\newblock


\bibitem[\protect\citeauthoryear{Shavit}{Shavit}{2011}]%
        {Shavit2011-gd}
\bibfield{author}{\bibinfo{person}{Nir Shavit}.}
  \bibinfo{year}{2011}\natexlab{}.
\newblock \showarticletitle{Data structures in the multicore age}.
\newblock \bibinfo{journal}{{\em Commun. ACM\/}} \bibinfo{volume}{54},
  \bibinfo{number}{3} (\bibinfo{date}{March} \bibinfo{year}{2011}),
  \bibinfo{pages}{76--84}.
\newblock


\bibitem[\protect\citeauthoryear{Shavit and Taubenfeld}{Shavit and
  Taubenfeld}{2015}]%
        {Shavit2015-ce}
\bibfield{author}{\bibinfo{person}{Nir Shavit} {and} \bibinfo{person}{Gadi
  Taubenfeld}.} \bibinfo{year}{2015}\natexlab{}.
\newblock \showarticletitle{The Computability of Relaxed Data Structures:
  Queues and Stacks as Examples}.
\newblock In \bibinfo{booktitle}{{\em Structural Information and Communication
  Complexity}}, \bibfield{editor}{\bibinfo{person}{Christian Scheideler}}
  (Ed.). \bibinfo{series}{Lecture Notes in Computer Science},
  Vol.~\bibinfo{volume}{9439}. \bibinfo{publisher}{Springer International
  Publishing}, \bibinfo{address}{Cham}, \bibinfo{pages}{414--428}.
\newblock


\bibitem[\protect\citeauthoryear{Sheehy}{Sheehy}{2015}]%
        {Sheehy2015-uh}
\bibfield{author}{\bibinfo{person}{Justin Sheehy}.}
  \bibinfo{year}{2015}\natexlab{}.
\newblock \bibinfo{title}{There is no now}.
\newblock   (\bibinfo{year}{2015}), \bibinfo{numpages}{36--41}~pages.
\newblock


\bibitem[\protect\citeauthoryear{Singh, Neamtiu, and Gupta}{Singh
  et~al\mbox{.}}{2016}]%
        {Singh2016-hh}
\bibfield{author}{\bibinfo{person}{V Singh}, \bibinfo{person}{I Neamtiu}, {and}
  \bibinfo{person}{R Gupta}.} \bibinfo{year}{2016}\natexlab{}.
\newblock \showarticletitle{Proving Concurrent Data Structures Linearizable}.
  In \bibinfo{booktitle}{{\em 2016 {IEEE} 27th International Symposium on
  Software Reliability Engineering ({ISSRE})}}. \bibinfo{pages}{230--240}.
\newblock


\bibitem[\protect\citeauthoryear{Strachey}{Strachey}{1967}]%
        {Strachey1967-qn}
\bibfield{author}{\bibinfo{person}{C.~S. Strachey}.}
  \bibinfo{year}{1967}\natexlab{}.
\newblock \bibinfo{booktitle}{{\em Fundamental Concepts of Programming
  Languages}}.
\newblock \bibinfo{publisher}{Programming Research Group}.
\newblock


\bibitem[\protect\citeauthoryear{Stroustrup}{Stroustrup}{2013}]%
        {Stroustrup2013-yo}
\bibfield{author}{\bibinfo{person}{Bjarne Stroustrup}.}
  \bibinfo{year}{2013}\natexlab{}.
\newblock \bibinfo{booktitle}{{\em The C++ Programming Language (4th Ed.)}}.
\newblock \bibinfo{publisher}{Pearson Education}.
\newblock


\bibitem[\protect\citeauthoryear{Tofan, Travkin, Schellhorn, and
  Wehrheim}{Tofan et~al\mbox{.}}{2014}]%
        {Tofan2014-ja}
\bibfield{author}{\bibinfo{person}{Bogdan Tofan}, \bibinfo{person}{Oleg
  Travkin}, \bibinfo{person}{Gerhard Schellhorn}, {and} \bibinfo{person}{Heike
  Wehrheim}.} \bibinfo{year}{2014}\natexlab{}.
\newblock \showarticletitle{Two approaches for proving linearizability of
  multiset}.
\newblock \bibinfo{journal}{{\em Science of Computer Programming\/}}
  \bibinfo{volume}{96} (\bibinfo{date}{Dec.} \bibinfo{year}{2014}),
  \bibinfo{pages}{297--314}.
\newblock


\bibitem[\protect\citeauthoryear{Treiber}{Treiber}{1986}]%
        {Treiber1986-cr}
\bibfield{author}{\bibinfo{person}{R~Kent Treiber}.}
  \bibinfo{year}{1986}\natexlab{}.
\newblock \bibinfo{booktitle}{{\em Systems programming: Coping with
  parallelism}}.
\newblock \bibinfo{publisher}{International Business Machines Incorporated,
  Thomas J. Watson Research Center New York}.
\newblock


\bibitem[\protect\citeauthoryear{Vechev, Yahav, and Yorsh}{Vechev
  et~al\mbox{.}}{2009}]%
        {vechev2009experience}
\bibfield{author}{\bibinfo{person}{Martin Vechev}, \bibinfo{person}{Eran
  Yahav}, {and} \bibinfo{person}{Greta Yorsh}.}
  \bibinfo{year}{2009}\natexlab{}.
\newblock \showarticletitle{Experience with model checking linearizability}. In
  \bibinfo{booktitle}{{\em International SPIN Workshop on Model Checking of
  Software}}. Springer, \bibinfo{pages}{261--278}.
\newblock


\bibitem[\protect\citeauthoryear{Wen, Song, and You}{Wen et~al\mbox{.}}{2018}]%
        {Wen2018-le}
\bibfield{author}{\bibinfo{person}{Tangliu Wen}, \bibinfo{person}{Lan Song},
  {and} \bibinfo{person}{Zhen You}.} \bibinfo{year}{2018}\natexlab{}.
\newblock \showarticletitle{Proving Linearizability Using Reduction}.
\newblock \bibinfo{journal}{{\em Comput. J.\/}} (\bibinfo{date}{Nov.}
  \bibinfo{year}{2018}).
\newblock


\bibitem[\protect\citeauthoryear{Williams, Oliker, Vuduc, Shalf, Yelick, and
  Demmel}{Williams et~al\mbox{.}}{2007}]%
        {williams2007optimization}
\bibfield{author}{\bibinfo{person}{Samuel Williams}, \bibinfo{person}{Leonid
  Oliker}, \bibinfo{person}{Richard Vuduc}, \bibinfo{person}{John Shalf},
  \bibinfo{person}{Katherine Yelick}, {and} \bibinfo{person}{James Demmel}.}
  \bibinfo{year}{2007}\natexlab{}.
\newblock \showarticletitle{Optimization of sparse matrix-vector multiplication
  on emerging multicore platforms}. In \bibinfo{booktitle}{{\em SC'07:
  Proceedings of the 2007 ACM/IEEE Conference on Supercomputing}}. IEEE,
  \bibinfo{pages}{1--12}.
\newblock


\bibitem[\protect\citeauthoryear{Wimmer, Gruber, Tr{\"a}ff, and Tsigas}{Wimmer
  et~al\mbox{.}}{2015}]%
        {Wimmer2015-yh}
\bibfield{author}{\bibinfo{person}{Martin Wimmer}, \bibinfo{person}{Jakob
  Gruber}, \bibinfo{person}{Jesper~Larsson Tr{\"a}ff}, {and}
  \bibinfo{person}{Philippas Tsigas}.} \bibinfo{year}{2015}\natexlab{}.
\newblock \showarticletitle{The Lock-free {k-LSM} Relaxed Priority Queue}. In
  \bibinfo{booktitle}{{\em Proceedings of the 20th {ACM} {SIGPLAN} Symposium on
  Principles and Practice of Parallel Programming}} {\em
  (\bibinfo{series}{PPoPP 2015})}. \bibinfo{publisher}{ACM},
  \bibinfo{address}{New York, NY, USA}, \bibinfo{pages}{277--278}.
\newblock


\bibitem[\protect\citeauthoryear{Wing}{Wing}{1989}]%
        {wing1989verifying}
\bibfield{author}{\bibinfo{person}{Jeannette~M Wing}.}
  \bibinfo{year}{1989}\natexlab{}.
\newblock \showarticletitle{Verifying atomic data types}.
\newblock \bibinfo{journal}{{\em International Journal of Parallel
  Programming\/}} \bibinfo{volume}{18}, \bibinfo{number}{5}
  (\bibinfo{year}{1989}), \bibinfo{pages}{315--357}.
\newblock


\bibitem[\protect\citeauthoryear{Yang}{Yang}{2018}]%
        {Yang2018-lw}
\bibfield{author}{\bibinfo{person}{Chaoran Yang}.}
  \bibinfo{year}{2018}\natexlab{}.
\newblock \bibinfo{title}{Fast Wait Free Queue}.
\newblock
  \bibinfo{howpublished}{\url{https://github.com/chaoran/fast-wait-free-queue}}.
    (\bibinfo{date}{Oct.} \bibinfo{year}{2018}).
\newblock
\newblock
\shownote{Accessed: 2019-2-5.}


\bibitem[\protect\citeauthoryear{Yang and Mellor-Crummey}{Yang and
  Mellor-Crummey}{2016}]%
        {Yang2016-io}
\bibfield{author}{\bibinfo{person}{Chaoran Yang} {and} \bibinfo{person}{John
  Mellor-Crummey}.} \bibinfo{year}{2016}\natexlab{}.
\newblock \showarticletitle{A wait-free queue as fast as fetch-and-add}. In
  \bibinfo{booktitle}{{\em Proceedings of the 21st {ACM} {SIGPLAN} Symposium on
  Principles and Practice of Parallel Programming - {PPoPP} '16}}.
  \bibinfo{publisher}{ACM Press}, \bibinfo{address}{New York, New York, USA},
  \bibinfo{pages}{1--13}.
\newblock


\bibitem[\protect\citeauthoryear{Zaremski and Wing}{Zaremski and Wing}{1995}]%
        {Zaremski1995-vm}
\bibfield{author}{\bibinfo{person}{Amy~Moormann Zaremski} {and}
  \bibinfo{person}{Jeannette~M Wing}.} \bibinfo{year}{1995}\natexlab{}.
\newblock \showarticletitle{Specification Matching of Software Components}. In
  \bibinfo{booktitle}{{\em Proceedings of the 3rd {ACM} {SIGSOFT} Symposium on
  Foundations of Software Engineering}} {\em (\bibinfo{series}{SIGSOFT '95})}.
  \bibinfo{publisher}{ACM}, \bibinfo{address}{New York, NY, USA},
  \bibinfo{pages}{6--17}.
\newblock


\bibitem[\protect\citeauthoryear{Zhang, Chattopadhyay, and Wang}{Zhang
  et~al\mbox{.}}{2015}]%
        {zhang2015round}
\bibfield{author}{\bibinfo{person}{Lu Zhang}, \bibinfo{person}{Arijit
  Chattopadhyay}, {and} \bibinfo{person}{Chao Wang}.}
  \bibinfo{year}{2015}\natexlab{}.
\newblock \showarticletitle{Round-up: runtime verification of quasi
  linearizability for concurrent data structures}.
\newblock \bibinfo{journal}{{\em IEEE Transactions On Software Engineering\/}}
  \bibinfo{volume}{41}, \bibinfo{number}{12} (\bibinfo{year}{2015}),
  \bibinfo{pages}{1202--1216}.
\newblock


\end{thebibliography}

%% Appendix
%\appendix
%\section{Appendix}

%Text of appendix \ldots

\end{document}